\documentclass[a4paper,12pt]{article}
\usepackage{geometry}
\usepackage{cuted}
\usepackage{lipsum}
\usepackage{graphicx}
\usepackage{authblk}
\usepackage{chessfss}
\usepackage{xkeyval}
\usepackage{xifthen}
\usepackage{pgfcore}
\usepgfmodule{shapes}
\usepackage{pst-node}
\usepackage{chessboard}
\usepackage{float}
\usepackage[english]{babel}
\usepackage{csquotes}
\usepackage{tikz}
\usepackage[vlined,ruled]{algorithm2e}
\usepackage{xcolor}
\usepackage{algorithm2e}
\usepackage[title]{appendix}
\usepackage{tikz-cd}
\usepackage{ amssymb }

\tikzset{
  symbol/.style={
    draw=none,
    every to/.append style={
      edge node={node [sloped, allow upside down, auto=false]{$#1$}}}
  }
}

\usepackage{graphicx}
\usepackage{lineno}
\usepackage{amsmath,amsfonts,amsthm}
\usepackage{amsthm}
\usepackage{chngcntr}

\newtheorem{theorem}{Theorem}%[section]

\usepackage{tcolorbox}
\tcbuselibrary{theorems}

\newtcbtheorem{mydefinition}{Definition}%
{colback=green!8,colframe=green!35!black,fonttitle=\bfseries}{th}

\newtcbtheorem{mytheo}{Theorem}%
{colback=blue!8,colframe=blue!35!black,fonttitle=\bfseries}{th}

\usepackage{csquotes} % Recommended by biblatex
\usepackage[style=numeric, sorting=none, backend=biber, urldate=long]{biblatex}
\usepackage{booktabs}												% Nicer tables
\usepackage{fix-cm}													% Custom fontsizes
\usepackage{todonotes}
\usepackage[section]{placeins}
\let\Oldsection\section
\renewcommand{\section}{\FloatBarrier\Oldsection}

\let\Oldsubsection\subsection
\renewcommand{\subsection}{\FloatBarrier\Oldsubsection}

\let\Oldsubsubsection\subsubsection
\renewcommand{\subsubsection}{\FloatBarrier\Oldsubsubsection}

\usepackage{mathtools}

\theoremstyle{definition}
\newtheorem{definition}{Definition}[section]

\theoremstyle{lemma}
\newtheorem{lemma}{Lemma}[section]

\usepackage{color}   %May be necessary if you want to color links
\usepackage{hyperref}
\hypersetup{
    colorlinks=true, %set true if you want colored links
    linktoc=all,     %set to all if you want both sections and subsections linked
    linkcolor=blue,  %choose some color if you want links to stand out
    filecolor=magenta,      
    urlcolor=blue,
}
\title {\textbf{Semantic Embeddings in Semilattices}}
\author[1,2]{Fernando Martin-Maroto}
\author[1,2]{Gonzalo G. de Polavieja}
\affil[1]{\small Champalimaud Research, Champalimaud Foundation, Lisbon, Portugal} \affil[2]{\small Algebraic AI, Madrid, Spain }

\begin{document}
\maketitle
\begin{abstract}

To represent anything from mathematical concepts to real-world objects, we have to resort to an encoding. Encodings, such as written language, usually assume a decoder that understands a rich shared code. A semantic embedding is a form of encoding that assumes a decoder with no knowledge, or little knowledge, beyond the basic rules of a mathematical formalism such as an algebra. Here we give a formal definition of a semantic embedding in a semilattice which can be used to resolve machine learning and classic computer science problems. Specifically, a semantic embedding of a problem is here an encoding of the problem as sentences in an algebraic theory that extends the theory of semilattices. We use the recently introduced formalism of finite atomized semilattices \cite{AtomizedSL} to study the properties of the embeddings and their finite models. For a problem embedded in a semilattice we show that every solution has a model atomized by an irreducible subset of the non-redundant atoms of the freest  model of the embedding. We give examples of semantic embeddings that can be used to find solutions for the N-Queen's completion, the Sudoku, and the Hamiltonian Path problems.

\end{abstract}

\section{Introduction} 

An embedding is a structure-preserving map that allows to identify a mathematical structure A within another B, for example a group into another group. For such mapping to be possible, A has to be a subgroup of B.  This quite restrictive notion can be extended to the more general concept of semantic embedding \cite{Burris, Marker}. A semantic embedding of A into B can be defined even when A and B are very different mathematical objects with different properties and operations. For example, a group can be semantically embedded within a graph.

Although the concept of semantic embedding is well-known to mathematical logicians and theoretical computer scientists, it is rarely used to resolve practical problems. Semantic embeddings have been used mostly to study undecidability \cite{Burris}. 

Here we take this notion to the extreme when A is as complex as real world problems and B is as simple as a semilattice, a minimalist algebra with just one binary idempotent operator. Despite the simplicity of semilattices, we will see that it is possible to find semantic embeddings for real-world problems that can be used in practice. For example, A can be solving a sudoku or classifying data in a machine learning problem. 

The word \emph{embedding} is often used as a synonym of \emph{encoding}. Encodings, such as written language, usually assume a decoder that understands a rich shared code. A semantic embedding is a form of encoding that assumes a decoder with no knowledge, or little knowledge, beyond the basic rules of a mathematical formalism such an algebra. For example, encoding a novel A into a book B requires understanding written words and such understating of words is part of the interpretation mechanism. The interpretation mechanism for the novel is somehow fixed and rich. On the other hand, in a semantic embedding rather than relying on a fixed and complex interpretation we try to discover an interpretation that allows us to find A into B. Searching for a simple and accurate interpretation of A into B is a natural problem that can, potentially, teach us things about A and B.  

To study the properties of semantic embeddings in semilattices, we will extensively apply the formalism of finite atomized semilattices \cite{AtomizedSL}. In this formalism, semilattice structures (we call them here semilattice models) are represented as sets of atoms, where atoms are closely related to the subdirectly irreducible components of the semilattice. There are substantial advantages in using atomized semilattice models. For example, atomized models have a linearity property; if $M$ and $N$ are two models of a set $R$ of atomic sentences then the model $M + N$ spawned by the atoms of $M$ and the atoms of $N$ is also a model of $R$. We will see how the atomized description of models naturally provides a framework that is ideal to study semantic embeddings. 

Using atomized semilattices we show that, for a problem semantically embedded in a semilattice every solution of the problem has a model that is irreducible (all its atoms are necessary) and it is atomized by a subset of the non-redundant atoms of the freest model of the embedding. This result reduces the search space for solution models and shows that the mere statement of a problem with a semantic embedding is a step towards finding solutions. Notice that a semantic embedding is only a problem statement. To construct a semantic embedding there is no need to know how to resolve the problem. 

Semantic embeddings and atomized descriptions of semilattice models are central to \emph{Algebraic Machine Learning} \cite{AMLPresentation}, a machine learning method that can be used for pattern recognition as well as for combining learning from data and from formal descriptions of problems. Constructing a semantic embedding is the first step in this method. The embeddings given here can be used for Algebraic Machine Learning with no need for further modifications. With these embeddings the sparse crossing algorithm described in \cite{AMLPresentation} can easily find full boards of queens, solve Sudokus or create new ones. Sparse crossing can also find Hamiltonian paths better than a naive backtracking method \footnote{Sparse Crossing \cite{AMLPresentation} can find Hamiltonian paths with the embeddings presented here. Specifically, randomly generated graphs with a probability of an edge between nodes of $5\%$ are hard for backtracking because they have enough edges to have a large search space but not enough to easily find a path that does not require crossing the same node twice. For graphs of about 50 vertexes naive backtracking starts becoming impractical as it requires exponentially more time as the number of nodes grow. Sparse Crossing could find Hamiltonian paths in a few attempts for graphs with more than 100 vertexes, for which backtracking could take an astronomically large amount of time. 
Although the use of sparse crossing is barely discussed in this document we thought is worth mentioning it to motivate the reader and to draw attention on the practical relevance of semilattice embeddings.}. 

We have organized this paper as follows. In section \ref{definions} we give a definition for semantic embeddings into semilattices. In section \ref{classesOfEmbeddings} we introduce four different kinds of embeddings. We also show how to use atomized semilattices to characterize the different kinds of embeddings. In sections \ref{trivialExample} to \ref{sec:hamiltonian_path} we provide examples of semantic embeddings, starting with a very simple example in section \ref{trivialExample}. Section \ref{prerequisites} is a summary of the notions from Atomized Semilattices that we use in this document. Atomized Semilattices are further developed here with the introduction of grounded models in section \ref{groundedModels}. Section \ref{sec:theorems} is devoted to the theorems, together with section \ref{sec:context_constants} that focuses on theorems related to context constants.

\section{Definitions} \label{definions} 

An algebraic semantic embedding of a problem $P$ is an encoding of $P$ in the form of an algebraic theory. A semantic embedding relies on the constants and operations of an algebra to describe $P$ with a set of sentences. In this paper we focus on algebraic embeddings over finite semilattices in its simplest form; an algebraic theory consisting of a set $R$ of positive and negative atomic sentences.

A semantic embedding maps each solution of the problem $P$ to a model of $R$ and, conversely, permits to interpret unambiguously some models of $R$ as solutions of $P$. For each solution $S$ of $P$ there should be at least one model $M_S \models R$ that maps to S. Such mapping is ``the interpretation'' of P in the algebra.

We use the word ``duple'' to refer to an ordered pair $(a, b)$ of elements. We also refer to positive atomic sentence $a \leq b$ as ``positive duple'' and to the negative atomic sentence $a \not\leq b$ as ``negative duple''.
We say that the duple is ``over $C$'' if the left and right-hand sides of the duple, $a$ and $b$ respectively, are either constants or idempotent summation of constants from a set $C$ of constants, i.e. they are elements of the freest semilattice generated by the set $C$.

\bigskip
\begin{definition} \label{embeddingSet} 
An embedding set $R = R^{+} \cup R^{-}$ is a set $R^{+}$ of positive duples and a set $R^{-}$ of negative duples over a set of constants $C$. We refer to the constants in $C$ as the embedding constants.
\end{definition}

\begin{definition} \label{interpretation} 
An interpretation sentence for a problem $P$ into an algebra is a sentence ${\bf{\varphi}}$ that uses the functions and relations of the algebra and mentions constants in a set $Q$ such that a model $M$ is a solution of $P$ if and only if $M \models {\bf{\varphi}}$. We refer to the constants in $Q$ as the interpretation constants. 
\end{definition}

In the case of semilattices we can use to construct ${\bf{\varphi}}$ the idempotent operator $\odot$ and the order relation $<$ that can be defined from the idempotent operator. The set $Q$ of interpretation constants is a subset of the embedding constants $C$.

For example, if $P$ is the problem of resolving a Sudoku we can choose an interpretation sentence ${\bf{\varphi}}$ that says that there is a number written at each cell and the numbers are not repeated in the rows, columns and subgrids. It is possible to define an interpretation sentence ${\bf{\varphi}}$ without knowing how to find a solution or even without knowing if a solution exists. 

The interpretation sentence is assumed to be a first order sentence with or without quantifiers. We will be able to build embeddings for complex interpretation sentences with quantifiers as long as we can find a first order sentence without quantifiers, the ``scope sentence'', that is equivalent to ${\bf{\varphi}}$ in the presence of the embedding set $R$. The scope sentence has no quantifiers but it may have negations, conjunctions and disjunctions. We will see that for problems like the N-Queens Completion, the Sudoku or the Hamiltonian Path such scope sentence exists.

\bigskip
\begin{definition} \label{scope} 
A scope sentence $\Xi$ for a problem $P$ and an embedding set $R$ is a first order sentence without quantifiers and with constants in $Q$ that follows from the interpretation sentence ${\bf{\varphi}}$ and such that if a model of $R$ satisfies $\Xi$ then it also satisfies ${\bf{\varphi}}$, i.e. \[
\Xi \wedge R \vdash {\bf{\varphi}} \,\,\, and \,\,\, {\bf{\varphi}} \vdash \Xi .
\]
We use the entail symbol, $\vdash$, to state that every semilattice model of the sentence in the left is a model of the sentence in the right.
For short, we say that a sentence is over $Q$ when it only mentions constants in the set of constants $Q$. 
\end{definition}
\bigskip

It is clear that the scope and interpretation sentences are equivalent when R holds: \[
R \wedge  {\bf{\varphi}} \Leftrightarrow \Xi \wedge R .
\]
where we are using the double implication symbol to say that the sentences at each side of the symbol have the same semilattice models.   

In the Sudoku example, the embedding set $R$ may say that there is at most one number on each cell and that the rows columns and subgrids cannot have repeated numbers, which we can succinctly say using a conjunction of atomic sentences. The scope sentence may say that there is at least one number written at each cell. An scope sentence for this problem can be written as a conjunction, for each cell, of a disjunction that says that either 0, or 1 or 2 ... or 9 are written in the cell. There are models of $R$ that do not assign a number to all of the cells and, hence, do not correspond to solutions of the problem. The scope sentence is needed to distinguish the models of $R$ that are solutions of the problem $P$ from the models of $R$ that are not. 

In the examples that we analyze in this document, the scope sentences are conjunctions of a few disjunctions of (positive and/or negative) atomic sentences. We refer to these atomic sentences as the subclauses of $\Xi$. We cannot directly make the scope sentence part of the embedding set because that will require a very large number of atomic sentences. We will learn how to deal with this problem and build semantic embeddings that can be used in practice.

\bigskip

A problem $P$ may have one or many solutions. We can write:
\[
\bf{\varphi} \Leftrightarrow \vee_s {\bf{\varphi}}_s,
\]
where every ${\bf{\varphi}}_s$ is a first order sentence that mentions constants in $Q$ and describes the solution $S$ of $P$. In most problems, ${\bf{\varphi}}$ would be known but not the sentences ${\bf{\varphi}}_s$. 

Consider the transformation of the scope sentence to disjunctive normal form $\Xi  = \vee_k \Xi_k$. The index $k$ may take a large number of values which makes impractical to directly make the scope sentence part of R. Each clause $\Xi_k$ is a conjunction of a subset of the subclauses mentioned in $\Xi$. Every subclause of $\Xi_k$ is either a subclause of $\Xi$, a negated subclause of $\Xi$. 

Since ${\bf{\varphi}} \vdash \Xi$ then for every solution $S$ we have ${\bf{\varphi}}_s \vdash \Xi$. In fact, for each solution $S$ there is at least one value of $k$ such that ${\bf{\varphi}}_s \vdash \Xi_k$. In case that there are multiple clauses $\Xi_k$ implied by ${\bf{\varphi}}_s$ we can, without loss of generality, do the conjunction of such clauses and define a single clause $\Xi_S$ such that:\[
\Xi_S \wedge R \vdash {\bf{\varphi}}_S \,\,\, and \,\,\, {\bf{\varphi}}_S \vdash \Xi_S,
\] so it follows: \[
R \wedge  {\bf{\varphi}}_S \Leftrightarrow \Xi_S \wedge R .
\]

Notice that some clauses $\Xi_k$, generally most of them, may not be realized by any solution.

\bigskip

To be specific about what makes two solutions of $P$ different we introduce the concept of a separator set.

\bigskip
\begin{definition} \label{separatorSet} 
Given a scope sentence $\Xi$ for a problem $P$ a separator set $\Gamma$ is a set of positive atomic sentences over the interpretation constants Q such that if either $\sigma$ or $\neg \sigma$ is an atomic subclause of $\Xi$ then $\sigma \in \Gamma$.
\end{definition}
\bigskip

Now we can identify a solution $S$ with a set of sentences as follows:

\bigskip
\begin{definition} \label{solution} 
A solution $S$ of a problem $P$ with scope sentence $\Xi$ and separator set $\Gamma$ is a subset $S \subseteq \Gamma$ such that the sentence $\gamma_S = \wedge \{\pi: \pi \in S\} \wedge \{\neg \nu: \nu \in \Gamma \cap \overline{S}\}$ implies $\Xi$ and $\gamma_S \wedge R$ has a model. Resolving $P$ amounts to finding such subsets of $\Gamma$.
\end{definition}
\bigskip
\begin{definition} \label{solutionModel}
A solution model for solution S is any model of $R \wedge \gamma_S$. 
\end{definition}
\bigskip

We use $\overline{S}$ for the complementary set of $S$, i.e. $\Gamma \cap \overline{S} = \Gamma - S$. Notice that we have defined the separator set $\Gamma$ as a superset of the set of atomic sentences mentioned in the scope sentence.  A separator set can be defined even when a problem has no need for a non trivial scope sentence. The separator set gives us the freedom to choose what constitutes a different solution and to separate solutions beyond what the scope sentence does. We cannot properly talk about solutions unless they are distinguished by the separator set. Without loss of generality we say that an embedding has solutions: \[
{\bf{\varphi}}_S := \gamma_S,
\]
and then ${\bf{\varphi}} \Leftrightarrow \vee_s \gamma_S$

We can now define our semantic embedding:

\bigskip
\begin{definition} \label{algebraicSemanticEmbedding} Let a problem $P$ be described by an interpretation sentence ${\bf{\varphi}}$ over $Q$. An algebraic semantic embedding, $(C, R, \Xi)$, of a problem $P$ with interpretation $(Q, {\bf{\varphi}}, \Gamma)$ is a set $R$ of positive and negative atomic sentences over a superset of constants $C \supseteq Q$ such that $\Xi \wedge R \vdash {\bf{\varphi}} \vdash \Xi$, where the scope sentence $\Xi$ is a first order sentence without quantifiers and over $Q$ and the separator set $\Gamma$ contains the atomic sentences mentioned by $\Xi$ as positive atomic sentences. 
\end{definition}
\bigskip

To build an embedding we usually need an extended set $C$ of constants, a set that contains $Q$. With the help of the extended set of constants it becomes possible to describe $P$ using a set $R$ of positive and negative atomic sentences. We left for the scope sentence everything that we cannot describe as a conjunction of positive and negative atomic sentences. We will see that there is a way to encode the scope sentence using atomic sentences and new constants. This encoding is weaker that the one provided but the embedding set $R$ but it suffices, in practice to find solutions for the problem at hand. 

From the text Finite Atomized Semilattices \cite{AtomizedSL} the reader should at least be familiar with the concepts of atomization, redundant and non-redundant atoms, and full crossing. Besides these concepts and methods, there are a handful of results required that are listed in section \ref{prerequisites} of this paper. The formalism of atomized semilattices is extended here with the introduction of ``grounded models'' in section \ref{groundedModels}. 

For a solution $S$ a solution model of particular interest is the freest one:

\bigskip
\begin{definition} \label{freestSolution}
The freest solution model for a solution $S$ of an embedding is the model $F_S \equiv F_{C}(R^{+} \cup S)$.
\end{definition}

To build the freest model $F_C(R^{+})$ of a set of positive atomic sentences $R^{+}$ we start from the freest model $F_C(\emptyset)$ over $C$ and use the full crossing algorithm to enforce each atomic sentence, one after another, in any order. The reader can find a proof for this in \cite{AtomizedSL}. 

We can use the full crossing method to build $F_S$ which is the freest model of the set of positive duples $R^{+} \cup S$. 

Theorem \ref{freestSolutionModel} shows that, $F_S$ satisfies the scope sentence and, hence, is a solution model.
Solution models are strictly less free than $F_S$ so they can be atomized with non-redundant atoms of $F_S$. A model $M_S$ is a solution model of the embedding for solution $S$ if it satisfies both, the positive atomic sentences, i.e. $M_S \models R^{+} \cup S$, and the negative atomic sentences $M_S  \models R^{-} \cup \{ \neg \nu : \nu \in \overline{S} \cap \Gamma \}$. As a consequence, $M_S \subseteq F_S$ where the inclusion says that the atoms of $M_S$ are all atoms of $F_S$, i.e. the atoms of $M_S$ are either non-redundant atoms of $F_S$ or unions of non-redundant atoms of $F_S$.   

Solution models have a linearity property; if $M_S$ and $N_S$ are solution models of $P$ for solution $S$ then the model \[
M_S + N_S
\] spawned by the atoms of $M_S$ and $N_S$ is also a solution model of $S$. Even more, if $M_S$ is a solution model of $S$ we can add any atom of $F_S$ to $M_S$ and still get a solution model of $S$.

\bigskip
\begin{definition} \label{irreducibleModel} 
Given a set of atomic sentences $R$ we say that a model $M$ is an irreducible model of $R$ if $M \models R$ and it is atomized by a set $A$ of non-redundant atoms of $M$ such that no subset of $A$ is a model of $R$. 
\end{definition}

Theorem \ref{minimalModelsTheroem} shows that an irreducible model has at most $\vert R^{-} \vert + 1$ atoms. Theorem \ref{positiveScopeTheroem} shows that we can find models for every solution of $P$ among the irreducible models of $R$. Even more, for tight embeddings (defined below) we can find models for every solution of $P$ in the set of irreducible models of $R$ atomized with non-redundant atoms of $F_C(R^{+})$. This is especially relevant as theorem \ref{solutionsAsSubsets} describes a practical mechanism to make every embedding tight.

\subsection{Concise, complete, tight and explicit embeddings} \label{classesOfEmbeddings} 

\begin{definition} \label{restriction} 
A restriction $M^{|Q}$ to $Q$ of an atomized semilattice model $M$ over $C$ is the subalgebra of $M$ spawned by the constants in the subset $Q \subset C$. 
\end{definition}

Assume $Q$ is a subset of $C$. Theorem \ref{restrictionLemma} proves that the restriction to $Q$ of an atomized semilattice with constants in $C$ can be calculated by restricting the upper constant segment of each atom $\phi$ to the constants in $Q$, i.e. replace each atom $\phi$ with upper constant segment $U^{c}(\phi) \subseteq C$ by other atom $\phi^{|Q}$ with upper constant segment $U^{c}(\phi^{|Q})  = U^{c}(\phi) \cap Q$. 

The restriction of a model $M_S$ to the subset of interpretation constants $Q$ still satisfies ${\bf{\varphi_S}}$. Therefore, the restriction to Q preserves the interpretation. However, since the duples that represent atomic sentences of $R$ use constants of $C$ that are not in $Q$, the duples of $R$ may not be defined in the restriction to $Q$. Therefore the restriction to $Q$ of a solution model $M$ may no longer be a model of an algebraic embedding but it is still interpretable as a solution of $P$.

\bigskip
\begin{definition} \label{concise} 
An algebraic embedding is concise if for every solution $S$ of problem $P$ and for any positive atomic sentence $r^{+}$ over the interpretation constants $Q$ we have $R \wedge {\bf{\varphi_S}} \vdash r^{+}$ if and only if ${\bf{\varphi_S}} \vdash r^{+}$.
\end{definition}
\bigskip

An embedding with embedding set $R$ is concise when it does not have implications irrelevant to the interpretation in the subalgebra spawned by interpretation constants $Q$.  For an embedding that is concise the formulas over $Q$ that are true in the freest model of a solution $S$ do not depend upon the embedding set $R$, they are determined by $S$ alone. An interesting result (see theorem \ref{weakEquivalenceTheorem}) is that any two embeddings, which could potentially be very different, of a problem $P$ with constants $C_1$ and $C_2$ that are concise and share the same interpretation $(Q, {\bf{\varphi}}, \Gamma)$ produce for each solution $S$ of $P$ a freest solution model of $S$ with the same non-redundant atoms restricted to $Q$. Also, the atoms in a model of an embedding that is not concise become, when restricted to $Q$, redundant with the atoms restricted to $Q$ of some model of any concise embedding.

\bigskip
\begin{definition} \label{completeEmbedding} 
An embedding set $R$ for a problem $P$ with solutions $S$ is complete if for each positive atomic sentence $r^{+}$ over the interpretation constants $Q$ such that $\forall S({\bf{\varphi_S}} \vdash r^{+})$ then $R \vdash r^{+}$.
\end{definition}

\begin{definition} \label{stronglyCompleteEmbedding} 
An embedding set $R$ for a problem $P$ with solutions $S$ is strongly complete if for each positive atomic sentence $r^{+}$ over the embedding constants $C$ such that $\forall S({\bf{\varphi_S}} \vdash r^{+})$ then $R \vdash r^{+}$.
\end{definition}
\bigskip

In Finite Atomized Semilattices\cite{AtomizedSL} we introduced the concept of redundant, non-redundant and external atoms. We say that a model $M$ has an atom $\phi$ or that $\phi$ is an atom of M, written $\phi \in M$, if $\phi$ is a non-redundant atom of $M$ or is redundant with $M$. We say that $\phi$ is not in $M$, written $\phi \not\in M$ if $\phi$ is external to $M$, i.e. if $\phi$ is not a non-redundant atom of $M$ neither is redundant with $M$.

\bigskip
\begin{definition} \label{residualComponent} 
An embedding of a problem $P$ with embedding set $R$, interpretation constants $Q$ and embedding constants $C$ has a residual atom $\alpha$ if $\alpha$ is a non-redundant atom of $F_C(R^{+})^{|Q}$ with upper segment different than $Q$ that is external to the restriction to $Q$ of every solution model of $P$, i.e. $\alpha \not\in M_S^{|Q}$ for every model $M_S \models R$ that models a solution of $P$. 
\end{definition}
\bigskip

Residual atoms are non-redundant atoms of $F_C(R^{+})$ that have consequences in the subalgebra spawned by $Q$ but do not play a role in any solution of $P$. Theorem \ref{residualAtomTheorem} proves that an embedding that has no residual atoms is complete; and conversely, an embedding that is complete and concise has no residual atoms.
Additionally, theorem \ref{equivalenceTheorem} shows that two concise and complete embeddings sharing the same interpretation $(Q, {\bf{\varphi}}, \Gamma)$ have the same atoms restricted to $Q$, i.e $F_{C1}(R_1^{+})^{|Q}= F_{C2}(R_2^{+})^{|Q}$. 

\bigskip
\begin{definition} \label{subsetAndTightSubmodels} 
A model $M$ is a subset model of $N$, written $M \subset N$ if every atom that is in model $M$ is in model $N$.  A model $M$ is a tight subset model of $N$, written $M \sqsubset N$ if the non-redundant atoms of $M$ are a subset of the non-redundant atoms of $N$. Subset $\subset$ and tight subset $\sqsubset$ are both transitive relations. 
\end{definition}

\begin{definition} \label{tightEmbedding} 
An embedding of a problem $P$ with embedding set $R$ and embedding constants $C$ is tight if for every solution $S$ of $P$ every non-redundant atom in the freest solution model $F_S =  F_{C}(R^{+} \cup S)$ is also a non-redundant atom in $F_C(R^{+})$. 
\end{definition}
\bigskip

A non-redundant atom of the freest solution of a model $F_S$ is always an atom of $F_{C}(R^{+})$ but it is not necessarily non-redundant in $F_{C}(R^{+})$, it can be a redundant atom of $F_{C}(R^{+})$. There may be many non-redundant atoms in a model but there are always many more redundant atoms. Every union of non-redundant atoms is a redundant atom, so we have, in the worse case, a number of redundant atoms in the order of 2 to the power of the number of non-redundant atoms. Tight embeddings have solution models spawn by atoms that are non-redundant atoms of $F_{C}(R^{+})$ so, when searching for non-redundant atoms of $F_S$ in the atoms of $F_{C}(R^{+})$ having a tight embedding severely reduces the search space.

\bigskip
\begin{definition} \label{grounding} 
Let $K \subset C$ and let $M$ be a model over $C$. The model $M^{{}^{\vee}K}$ (read $M$ grounded to $K$) is the model spawned by the subset of atoms $\phi$ of $M$ such that $U^{c}(\phi) \subseteq K$. 
\end{definition}

Model $M^{{}^{\vee}K}$ is as free or less free than $M$. Theorem \ref{groundingProperties} shows that $M^{{}^{\vee}K}$ is a well-defined model that is a tight subset of $M$, i.e. $M^{{}^{\vee}K} \sqsubset M$.

We introduce the concept of explicit embeddings that permits finding the freest model solutions by calculating a ground operation with respect to some subset $K$ of the constants.
A method that guarantees that our embedding is tight, i.e. $\forall S (F_S  \sqsubset F_C(R^{+}))$, is to build the embedding as an explicit embedding, something that can always be achieved following a simple procedure. 

\bigskip
\begin{definition} \label{explicitEmbedding} 
An embedding of a problem $P$ with interpretation constants $Q$ and embedding constants $C$ is explicit if for each solution $S$ of $P$ there is a subset $K_S$ of the constants $Q \subseteq K_S \subset C$ such that the freest solution for $S$ satisfies $F_S = F_{C}(R^{+})^{{}^{\vee}K_S} \oplus F_{C - K_S}( \emptyset)$.
\end{definition}
\bigskip

Theorem \ref{explicitIsTightTheorem} proves that an explicit embedding is tight, so a method that guarantees that our embedding is tight, i.e. $\forall S (F_S  \sqsubset F_C(R^{+}))$, is to build an explicit embedding, something we can always do following a recipe given in theorem \ref{solutionsAsSubsets}. 

Suppose we have an embedding set $R$ over a set of embedding constants $C$ for a problem $P$. Theorem \ref{solutionsAsSubsets} provides a mechanism to extend the embedding constants and the embedding set $R$ to obtain an explicit and, therefore, tight embedding of $P$. This extension mechanism can be used in practice. In addition, when the embedding is complete the extended embedding is also complete and when the embedding is concise the extended embedding is also concise.

\section{A trivial example}  \label{trivialExample}

Consider the following problem $P$: we have three elements $a$, $b$ and $c$ and either $a$ or $b$ is mapped to $c$. This problem has three solutions $S_1, S_2, S_3$: $a$ is mapped to $c$, $b$ is mapped to $c$, and both $a$ and $b$ are mapped to $c$. 

Consider the following interpretation sentence ${\bf{\varphi}} = (a \leq c) \vee (b \leq c)$ and the solution interpretation sentences
\[ 
{\bf{\varphi_1}} = (a \leq c) \wedge (b \not\leq c), 
\]\[ 
{\bf{\varphi_2}} = (a \not\leq c) \wedge (b \leq c),
\]\[ 
{\bf{\varphi_3}} = (a \leq c) \wedge (b \leq c),
\]
with interpretation constants $Q = \{a, b, c\}$, separator set $\Gamma = \{(a \leq c), (b \leq c)\}$ and an scope sentence equal to the interpretation sentence $\Xi = (a \leq c) \vee (b \leq c)$.

There is no atomic sentence that is a consequence of the interpretation sentence ${\bf{\varphi}}$ so all embeddings will be complete.

\underline{First embedding}:  $C =\{a, b, c\}$ and $R = \emptyset$. In this case $Q =C$. Any model over $C$ is a model of the empty embedding set, including models that satisfy the solution interpretation sentences ${\bf{\varphi_1}}, {\bf{\varphi_2}}$ and ${\bf{\varphi_3}}$. The embedding set $R = \emptyset$ is obviously concise. The model $F_C(\emptyset)$ is atomized by the three non-redundant atoms $\phi_a$,  $\phi_b$ and $\phi_c$ with upper segments $U^{c}(\phi_a) = \{a\}$, $U^{c}(\phi_b) = \{b\}$ and $U^{c}(\phi_c) = \{c\}$. The freest model and freest solution models are are:
\[
F_C(R^{+}) = [\phi_a, \phi_b, \phi_c]
\]\[
F_1 = [\phi_{ac}, \phi_{b}, \phi_{c}]
\]\[
F_2 = [\phi_{a}, \phi_{bc}, \phi_{c}]
\]\[
F_3 = [\phi_{ac}, \phi_{bc}, \phi_{c}]
\]
where $U^{c}(\phi_{ac}) = \{a, c\}$ and $U^{c}(\phi_{bc}) = \{b, c\}$. The embedding is neither explicit nor tight but it is complete and concise.

\underline{Second embedding}: $C =\{a, b, c, d, e\}$ and $R = (d \leq a) \wedge (d \leq c ) \wedge (e \leq b) \wedge (e \leq c)$.  The atomization of $F_C(R^{+})$ consists of five atoms with upper segments:
\[ 
U^{c}(\phi_{acd}) = \{a, c, d\}, \,\,\,\,\,  U^{c}(\phi_{bce}) = \{b, c, e\}
\]\[ 
U^{c}(\phi_a) = \{a\}, \,\,\,\,\, U^{c}(\phi_b) = \{b\}, \,\,\,\,\, U^{c}(\phi_c) = \{c\}.
\]
The freest model of the embedding set and freest solution models are are:
\[
F_C(R^{+}) =\square_{e \leq c} \square_{e \leq b} \square_{d \leq c} \square_{d \leq a}  [\phi_a, \phi_b, \phi_c, \phi_d,  \phi_e] = [\phi_a, \phi_b, \phi_c, \phi_{acd},  \phi_{bce}] 
\]\[
F_1 = F_C(R^{+} \cup (a \leq c)) =\square_{a \leq c}  [\phi_a, \phi_b, \phi_c, \phi_{acd},  \phi_{bce}] = [\phi_{ac}, \phi_b, \phi_c, \phi_{acd},  \phi_{bce}] 
\]\[
F_2 = F_C(R^{+} \cup (b \leq c)) =\square_{b \leq c}  [\phi_a, \phi_b, \phi_c, \phi_{acd},  \phi_{bce}] = [\phi_a, \phi_{bc}, \phi_c, \phi_{acd},  \phi_{bce}] 
\]\[
F_3 = F_C(R^{+} \cup (a \leq c) \cup (b \leq c)) =\square_{b \leq c} \square_{a \leq c} [\phi_a, \phi_b, \phi_c, \phi_{acd},  \phi_{bce}] = [\phi_{ac}, \phi_{bc}, \phi_c, \phi_{acd},  \phi_{bce}] 
\]
It is a good exercise to check that the three models satisfy $R \wedge {\bf{\varphi_1}}$, $R \wedge {\bf{\varphi_2}}$ and $R \wedge {\bf{\varphi_3}}$ respectively. There is no positive atomic sentence ${\bf{\xi}}^{+}$ over $Q = \{a, b, c\}$ satisfied by any of three models that is not implied by the interpretation sentences alone so the embedding is concise. The embedding is also complete but it is neither tight nor explicit.

Theorem \ref{restrictionLemma} promises that we can calculate a restriction model by calculating the restriction of each non-redundant atom:
\[
F_1^{|Q} = [\phi_{ac}, \phi_b, \phi_c, \phi_{acd},  \phi_{bce}]^{|Q} = [\phi_{ac}, \phi_b, \phi_c, \phi_{ac},  \phi_{bc}] = [\phi_{ac}, \phi_b, \phi_c]
\]\[
F_2^{|Q} = [\phi_a, \phi_{bc}, \phi_c, \phi_{acd},  \phi_{bce}]^{|Q}  =  [\phi_a, \phi_{bc}, \phi_c, \phi_{ac},  \phi_{bc}] = [\phi_a, \phi_{bc}, \phi_c] 
\]\[
F_3^{|Q} = [\phi_{ac}, \phi_{bc}, \phi_c, \phi_{acd},  \phi_{bce}]^{|Q} =  [\phi_{ac}, \phi_{bc}, \phi_c, \phi_{ac},  \phi_{bc}] =  [\phi_{ac}, \phi_{bc}, \phi_c] 
\]
Theorem \ref{conciseEmbeddingTheorem} says that an embedding is concise if and only if the freest solution model restricted to $Q=\{a, b, c\}$ is equal to $F_{Q}(S)$ for each solution $S$ of $P$. The reader can check that this is the case by calculating $F_{Q}(S)$ for the three solutions, in this case: $\square_{a \leq c} [\phi_a, \phi_b, \phi_c]$, $\square_{b \leq c} [\phi_a, \phi_b, \phi_c]$ and $\square_{a \leq c} \square_{b \leq c} [\phi_a, \phi_b, \phi_c]$. In addition, since the first embedding was also concise the restricted solution models are equal for both embeddings.

Finally the freest model restricted to $Q$ is:
\[
F_C(R^{+})^{|Q}  = [\phi_a, \phi_b, \phi_c, \phi_{acd},  \phi_{bce}]^{|Q} = [\phi_a, \phi_b, \phi_c, \phi_{ac},  \phi_{bc}] = [\phi_a, \phi_b, \phi_c] 
\]
Theorem \ref{equivalenceTheorem} tells us that two concise and complete embeddings with different embedding constants but sharing the same interpretation have the same atoms restricted to the interpretation constants $Q$, i.e $F_{C1}(R_1^{+})^{|Q}= F_{C2}(R_2^{+})^{|Q}$. For the first and second embeddings the freest model restricted to $Q$ is  $[\phi_a, \phi_b, \phi_c]$.

\underline{Third embedding}: The third embedding is very similar to the second embedding but this time is going to be a tight embedding. $C =\{a, b, c, d, e, m\}$ and $R = (d \leq a) \wedge (d \leq c ) \wedge (e \leq b) \wedge (e \leq c) \wedge (c \leq a \odot b \odot m)$. The atomization for the freest model is the result of $\square_{c \leq a \odot b \odot m} \square_{e \leq c} \square_{e \leq b} \square_{d \leq c} \square_{d \leq a}  [\phi_a, \phi_b, \phi_c, \phi_d,  \phi_e,  \phi_m]$ and consists of eight non-redundant atoms:
\[
F_C(R^{+}) = [\phi_a, \phi_b, \phi_m,  \phi_{ac}, \phi_{bc}, \phi_{cm}, \phi_{acd},  \phi_{bce}]. 
\]\
The freest solution models are:
\[
F_1 = \square_{a \leq c} [\phi_a, \phi_b, \phi_m,  \phi_{ac}, \phi_{bc}, \phi_{cm}, \phi_{acd},  \phi_{bce}] =  [\phi_b, \phi_m,  \phi_{ac}, \phi_{bc}, \phi_{cm}, \phi_{acd},  \phi_{bce}] 
\]\[
F_2 = \square_{b \leq c} [\phi_a, \phi_b, \phi_m,  \phi_{ac}, \phi_{bc}, \phi_{cm}, \phi_{acd},  \phi_{bce}] =  [\phi_a, \phi_m,  \phi_{ac}, \phi_{bc}, \phi_{cm}, \phi_{acd},  \phi_{bce}] 
\]\[
F_3 = \square_{b \leq c} \square_{a \leq c} [\phi_a, \phi_b, \phi_m,  \phi_{ac}, \phi_{bc}, \phi_{cm}, \phi_{acd},  \phi_{bce}] =  [\phi_m,  \phi_{ac}, \phi_{bc}, \phi_{cm}, \phi_{acd},  \phi_{bce}]. 
\]
It is quite straightforward to check that the embedding is tight:
\[
F_1 \sqsubset F_C(R^{+})  \,\,\,\,\,  F_2 \sqsubset F_C(R^{+})  \,\,\,\,\,  F_3 \sqsubset F_C(R^{+}), 
\]
so we have a complete, concise and tight embedding but not an explicit embedding.

The restrictions to $Q$ of the freest solution models are again the same to those obtained in the previous embeddings (as the three embeddings are concise).
\[
F_1^{|Q} = [\phi_b, \phi_m,  \phi_{ac}, \phi_{bc}, \phi_{cm}, \phi_{acd},  \phi_{bce}]^{|Q} = [\phi_{ac}, \phi_b, \phi_c]
\]\[
F_2^{|Q} =  [\phi_a, \phi_m,  \phi_{ac}, \phi_{bc}, \phi_{cm}, \phi_{acd},  \phi_{bce}]^{|Q}  =  [\phi_a, \phi_{bc}, \phi_c] 
\]\[
F_3^{|Q} =[\phi_m,  \phi_{ac}, \phi_{bc}, \phi_{cm}, \phi_{acd},  \phi_{bce}]^{|Q} =  [\phi_{ac}, \phi_{bc}, \phi_c] 
\] \[
F_C(R^{+})^{|Q}  = [\phi_a, \phi_b, \phi_m,  \phi_{ac}, \phi_{bc}, \phi_{cm}, \phi_{acd},  \phi_{bce}]^{|Q} = [\phi_a, \phi_b, \phi_c]. 
\]

\underline{Fourth embedding}: The fourth embedding is again very similar to the second and third embedding but this time is tight but not concise. $C =\{a, b, c, d, e\}$ and $R = (d \leq a) \wedge (d \leq c ) \wedge (e \leq b) \wedge (e \leq c) \wedge (c \leq a \odot b)$. The atomization for the freest model consists of six non-redundant atoms:
\[
F_C(R^{+}) = [\phi_a, \phi_b,  \phi_{ac}, \phi_{bc}, \phi_{acd},  \phi_{bce}]. 
\]
The freest solution models are:
\[
F_1 =  [\phi_b,  \phi_{ac}, \phi_{bc}, \phi_{acd},  \phi_{bce}] 
\]\[
F_2 = [\phi_a,  \phi_{ac}, \phi_{bc}, \phi_{acd},  \phi_{bce}] 
\]\[
F_3 =  [\phi_{ac}, \phi_{bc}, \phi_{acd},  \phi_{bce}]. 
\]
Again the embedding is tight:
\[
F_1 \sqsubset F_C(R^{+})  \,\,\,\,\,  F_2 \sqsubset F_C(R^{+})  \,\,\,\,\,  F_3 \sqsubset F_C(R^{+}), 
\]
but it is clearly not concise because it satisfies $c \leq a \odot b$, a duple over $Q$ that is not implied by the solutions of $P$ alone. The freest solution models restricted to $Q$ are:  
\[
F_1^{|Q}  =  [\phi_b,  \phi_{ac}, \phi_{bc}, \phi_{acd},  \phi_{bce}]^{|Q} =  [\phi_b,  \phi_{ac},  \phi_{bc}] 
\]\[
F_2^{|Q}  = [\phi_a,  \phi_{ac}, \phi_{bc}, \phi_{acd},  \phi_{bce}]^{|Q} =   [\phi_a,  \phi_{ac}, \phi_{bc}]
\]\[
F_3^{|Q}  =  [\phi_{ac}, \phi_{bc}, \phi_{acd},  \phi_{bce}]^{|Q} = [\phi_{ac}, \phi_{bc}]
\]\[
F_C(R^{+})^{|Q} = [\phi_a, \phi_b,  \phi_{ac}, \phi_{bc}, \phi_{acd},  \phi_{bce}]^{|Q} =  [\phi_a, \phi_b,  \phi_{ac}, \phi_{bc}].
\]
We can check that the freest solution models restricted to $Q$ are not equal to $F_{Q}(S)$. For example:
\[
\square_{a \leq c} [\phi_a, \phi_b, \phi_c] =  [\phi_{ac}, \phi_b, \phi_c]  \not=  F_1^{|Q}
\]

\underline{Fifth embedding}: This is a complete, concise, tight and explicit embedding. It is quite different from the others and uses embedding constants $C =\{a, b, c, g, h\}$ and embedding set $R = (a \leq c \odot g) \wedge (b \leq c \odot h)$. The atomization of the freest model $F_C(R^{+})$ consists of seven atoms with upper segments:
\[ 
U^{c}(\phi_{ac}) = \{a, c\}, \,\,\,\,\, U^{c}(\phi_{ag}) = \{a, g\},
\]\[ 
U^{c}(\phi_{bc}) = \{b, c\}, \,\,\,\,\, U^{c}(\phi_{bh}) = \{b, h\},
\]\[ 
U^{c}(\phi_{c}) = \{c\}, \,\,\,\,\, U^{c}(\phi_{g}) = \{g\}, \,\,\,\,\, U^{c}(\phi_{h}) = \{h\},
\]
that can be calculated as:
\[
F_C(R^{+}) =\square_{b \leq c \odot h} \square_{a \leq c \odot g}  [\phi_a, \phi_b, \phi_c, \phi_g, \phi_h] = [\phi_c, \phi_g, \phi_h, \phi_{ac}, \phi_{ag}, \phi_{bc}, \phi_{bh}].
\]
The freest solution models are:
\[
F_1 = [\phi_c, \phi_g, \phi_h, \phi_{ac}, \phi_{bc}, \phi_{bh}],
\]\[
F_2 =[\phi_c, \phi_g, \phi_h, \phi_{ac}, \phi_{ag}, \phi_{bc}],
\]\[
F_3 =[\phi_c, \phi_g, \phi_h, \phi_{ac}, \phi_{bc}],
\]
and it is clear that the embedding is tight:
\[
F_1 \sqsubset F_C(R^{+})  \,\,\,\,\,  F_2 \sqsubset F_C(R^{+})  \,\,\,\,\,  F_3 \sqsubset F_C(R^{+}), 
\]
and concise:
\[
F_1^{|Q} = [\phi_c, \phi_g, \phi_h, \phi_{ac}, \phi_{bc}, \phi_{bh}]^{|Q} =  [\phi_c, \phi_{ac}, \phi_{b}],
\]\[
F_2^{|Q}  =[\phi_c, \phi_g, \phi_h, \phi_{ac}, \phi_{ag}, \phi_{bc}]^{|Q} = [\phi_c, \phi_{a}, \phi_{bc}]
\]\[
F_3^{|Q}  =[\phi_c, \phi_g, \phi_h, \phi_{ac}, \phi_{bc}]^{|Q} = [\phi_c, \phi_{ac}, \phi_{bc}].
\]
It has the same restrictions of the first, second and third embeddings. 

For every solution there is a set of constants $K_S$ such that $Q \subseteq K_S \subset C$ that satisfy:
\[
F_C(R^{+})^{{}^{\vee} \{a, b, c, h\} } =  [\phi_c, \phi_g, \phi_h, \phi_{ac}, \phi_{ag}, \phi_{bc}, \phi_{bh}]^{{}^{\vee} \{a, b, c, h\} } =  [\phi_c, \phi_h, \phi_{ac}, \phi_{bc}, \phi_{bh}],
\]\[
F_C(R^{+})^{{}^{\vee} \{a, b, c, g\} } =  [\phi_c, \phi_g, \phi_h, \phi_{ac}, \phi_{ag}, \phi_{bc}, \phi_{bh}]^{{}^{\vee} \{a, b, c, g\} } =   [\phi_c, \phi_g, \phi_{ac}, \phi_{ag}, \phi_{bc}],
\]\[
F_C(R^{+})^{{}^{\vee} \{a, b, c\} } =  [\phi_c, \phi_g, \phi_h, \phi_{ac}, \phi_{ag}, \phi_{bc}, \phi_{bh}]^{{}^{\vee} \{a, b, c\} } =  [\phi_c, \phi_{ac}, \phi_{bc}],
\]
such that:
\[
F_1 = [\phi_c, \phi_g, \phi_h, \phi_{ac}, \phi_{bc}, \phi_{bh}] = F_C(R^{+})^{{}^{\vee} \{a, b, c, h\} } \cup \{\phi_g\}
\]\[
F_2 =[\phi_c, \phi_g, \phi_h, \phi_{ac}, \phi_{ag}, \phi_{bc}] = F_C(R^{+})^{{}^{\vee} \{a, b, c, g\} }  \cup \{\phi_h\}
\]\[
F_3 =[\phi_c, \phi_g, \phi_h, \phi_{ac}, \phi_{bc}] = F_C(R^{+})^{{}^{\vee} \{a, b, c\} } \cup \{\phi_g, \phi_h\}
\]
which proves the embedding is explicit. The fifth example corresponds to the construction in theorem \ref{solutionsAsSubsets} built over an empty embedding set. 

\bigskip
\bigskip

With two concise, complete and tight embeddings that share the same interpretation $(Q, {\bf{\varphi}}, \Gamma)$, for each non-redundant atom $\phi_1 \in \cup_S F_{1S}$ such that $\phi_1^{|Q}$ exists, theorem \ref{sameSolutionAtomsTheoremPartC} shows that either there is a solution $S$ for which $\phi_1^{|Q}$ is redundant in $F_{1S}^{|Q}$ or there is an atom $\phi_2 \in \cup_S F_{2S}$ with $\phi_1^{|Q} = \phi_2^{|Q}$. We can compare the third and fifth embeddings to see if this is true. With the third embedding:
\[
F_{\text{third}} = \cup_S F_{S} = F_C(R^{+}) =  [\phi_a, \phi_b, \phi_m,  \phi_{ac}, \phi_{bc}, \phi_{cm}, \phi_{acd},  \phi_{bce}]
\]
and with the fifth embedding:
\[
F_{\text{fifth}} = \cup_S F_{S} = F_C(R^{+}) =  [\phi_c, \phi_g, \phi_h, \phi_{ac}, \phi_{ag}, \phi_{bc}, \phi_{bh}].
\]
The sets obtained by calculating the restriction  to $Q$ of the non-redundant atoms are equal:
\[
F_{\text{third}} \xrightarrow{|Q}  \{ \phi_a, \phi_b, \phi_{ac}, \phi_{bc}, \phi_{c}, \phi_{ac},  \phi_{bc}  \} =  \{ \phi_a, \phi_b, \phi_{ac}, \phi_{bc}, \phi_{c} \} 
\]\[
F_{\text{fifth}} \xrightarrow{|Q}  \{ \phi_c, \phi_{ac}, \phi_{a}, \phi_{bc}, \phi_{b} \}, 
\]
Notice that here we are comparing sets and not models so redundant atoms matter and cannot be removed.

\section{Vertical Bars} 

Suppose an $n \times n$ grid of white and black pixels. The embedding consists of describing the properly formed grids, i.e. with pixels at every position and no empty or undefined positions, that have at least one black column. The problem $P$ consists of finding examples of grids with this property.

We are going to build and compare two different embeddings. The first embedding describes the grids in a formal manner, the second embedding provides examples of valid grids with a black bar. 

{\underline{Embedding 1:}} Consider the embedding with interpretation constants $Q =  \{ b_{ij}, w_{ij}, q_{ij}\}$ where constants $b_{ij}$ represent black pixels at position $(i, j)$, the constants $w_{ij}$ that represent white pixels and the constants $q_{ij}$ represent positions of the grid. The embedding constants $C = Q \cup \{ n_j \}$ where $n_j$ is a constant such that $n_j \leq w_{ij}$, i.e. it is lower or equal than every constant of a white pixel in column $j$, and so there is one $n_1,n_2,...,n_n$ per column. Consider the embedding set:
\[
R_{1A} =   (\odot_j n_j  \not\leq  \odot_{r,s}  q_{rs})  \wedge  \forall_{ij}  ((w_{ij} \odot b_{ij})  \not\leq \odot_{r,s}  q_{rs}) \wedge \forall_{ij}  (n_j \leq  w_{ij}), 
\]
and the scope sentence: \[ 
\Xi = \forall_{i,j} ((q_{ij} = w_{ij})  \vee  (q_{ij} = b_{ij})).
\]
Any model of $R_{1A} \wedge \Xi$ can be univocally mapped to a grid with a vertical black bar. This model is concise, in order to make it complete we have to add:
\[
R_{1B} =  \forall_{ij}  (q_{ij} \leq w_{ij} \odot b_{ij}),
\]
a set of sentences that are true in every model. To make the embedding also explicit and tight we can follow the strategy described in theorem \ref{solutionsAsSubsets} and extend $C$ to the set $C = Q \cup \{ n_j, g_{ij}, h_{ij} \}$ where $g_{ij}$ and $h_{ij}$ are ``context constants'' used as follows:   \[
R_{1C} =  \forall_{ij}  ((g_{ij} \odot q_{ij} = w_{ij} \odot g_{ij}) \wedge (h_{ij} \odot q_{ij} = b_{ij} \odot h_{ij})).
\]
where we have relied on theorem \ref{contextConstantReuse} to reuse the same context constant $g_{ij}$ for the atomic sentences $(q_{ij} \leq w_{ij})$ and $(w_{ij} \leq q_{ij})$ and the same context constant $h_{ij}$ for $(q_{ij} \leq b_{ij})$ and $(b_{ij} \leq q_{ij})$.
Finally, our first embedding is given by the embedding set:\[
R_{1} = R_{1A} \wedge R_{1B} \wedge R_{1C}.  
\]
A calculation of $F_C(R^{+}_{1})$ provides a model with 46 atoms in the case of $2 \times 2$ grids. Subsets of non-redundant atoms of $F_C(R^{+}_{1})$ that suffice to preserve $R^{-}_{1}$ produce models that, when they satisfy $\Xi$, represent examples of grids with at least one vertical black bar.  This model satisfies:\[
F =  \cup_s F_s =  F_C(R^{+}_{1}).
\]
Using for indexes i,j $1,1  \rightarrow 1$, $1,2  \rightarrow 2$, $2,1  \rightarrow 3$ and $2,2  \rightarrow 4$ for better readability, the atomization is: \[ 
F_C(R^{+}_{1}) = [\phi_{g_1},\phi_{g_2},\phi_{g_3},\phi_{g_4},\phi_{h_1},\phi_{h_2},\phi_{h_3},\phi_{h_4},\phi_{w_1g_1},\phi_{w_2g_2},\phi_{w_3g_3},\]\[\phi_{w_4g_4},\phi_{b_1h_1},\phi_{b_2h_2},\phi_{b_3h_3},\phi_{b_4h_4},\phi_{w_1b_1q_1},\phi_{w_2b_2q_2},\]\[\phi_{w_3b_3q_3},\phi_{w_4b_4q_4},\phi_{b_1q_1g_1},\phi_{b_2q_2g_2},\phi_{b_3q_3g_3},\phi_{b_4q_4g_4},\phi_{w_1q_1h_1},\]\[\phi_{w_2q_2h_2},\phi_{w_3q_3h_3},\phi_{w_4q_4h_4},\phi_{w_1w_3n_1g_1g_3},\phi_{w_2w_4n_2g_2g_4},\phi_{w_1w_3q_1n_1g_3h_1},\]\[\phi_{w_1w_3b_1q_1n_1g_3},\phi_{w_1w_3q_3n_1g_1h_3},\phi_{w_1w_3b_3q_3n_1g_1},\phi_{w_2w_4b_2q_2n_2g_4},\phi_{w_2w_4b_4q_4n_2g_2},\]\[\phi_{w_2w_4q_2n_2g_4h_2},\phi_{w_2w_4q_4n_2g_2h_4},\phi_{w_1w_3b_1q_1q_3n_1h_3},\phi_{w_1w_3b_1b_3q_1q_3n_1},\]\[\phi_{w_1w_3q_1q_3n_1h_1h_3},\phi_{w_1w_3b_3q_1q_3n_1h_1},\phi_{w_2w_4b_2b_4q_2q_4n_2},\phi_{w_2w_4b_2q_2q_4n_2h_4},\]\[\phi_{w_2w_4b_4q_2q_4n_2h_2},\phi_{w_2w_4q_2q_4n_2h_2h_4}] 
\]
The freest model for solution $q_1= b_1, q_2 = w_2, q_3 = b_3, q_4 = w_4$, has non-redundant atoms: \[ 
F_s = [\phi_{g_1},\phi_{g_2},\phi_{g_3},\phi_{g_4},\phi_{h_1},\phi_{h_2},\phi_{h_3},\phi_{h_4},\phi_{w_1g_1},\phi_{w_3g_3},\phi_{b_2h_2},\]\[\phi_{b_4h_4},\phi_{w_1b_1q_1},\phi_{w_2b_2q_2},\phi_{w_3b_3q_3},\phi_{w_2q_2h_2},\phi_{w_4b_4q_4},\phi_{b_1q_1g_1},\]\[\phi_{w_4q_4h_4},\phi_{b_3q_3g_3},\phi_{w_1w_3n_1g_1g_3},\phi_{w_1w_3b_1q_1n_1g_3},\phi_{w_1w_3b_3q_3n_1g_1},\]\[\phi_{w_1w_3b_1b_3q_1q_3n_1},\phi_{w_2w_4b_2b_4q_2q_4n_2},\phi_{w_2w_4b_2q_2q_4n_2h_4},\phi_{w_2w_4b_4q_2q_4n_2h_2},\]\[\phi_{w_2w_4q_2q_4n_2h_2h_4}] ,
\]
which is a subset of $F_C(R^{+}_{1})$, as promised by tightness. The embedding is also explicit:  \[ 
F_s = F_C(R^{+}_{1})^{{}^{\vee} Q \cup \{n_1, n_2\}  \cup \{g_1, g_3, h_2, h_4\} } \cup \{\phi_{h_1}, \phi_{h_3}, \phi_{g_2}, \phi_{g_4}\}. 
\]
We can find all the models of a vertical black bar, with a well-constructed grid that passes the scope sentence, by grounding $F_C(R^{+}_{1})$ with the right context constants (in this case $g_1, g_3, h_2, h_4$). In fact there are 16 sets of context constants that produce models for the 16 possible grids and all these grounded models pass the scope sentence. This includes $h_1,h_2, h_3, h_4$ that produces a model of a full white grid. However, only the models that satisfy $R^{-}_1$ are valid models of our embedding, and these models all have a vertical black bar. Every model without a vertical black bar that we find by grounding satisfy:  \[ 
n_1 \odot n_2  \leq q_1 \odot  q_2 \odot q_3 \odot  q_4 
\]
i.e. these models violate $\odot_j n_j  \not\leq  \odot_{r,s}  q_{rs}$ which is in $R^{-}$.

We have not yet used the negative duples of our embedding. Once we have the atomization of $F_C(R^{+}_{1})$ we can pick a subset of the atoms that suffice to enforce $R^{-}$. This choosing of atoms checks the $\vert R^{-}\vert$ duples just once and ends up with a set of at most $\vert R^{-}\vert$ atoms. It suffices to keep the first discriminant atom we encounter for one duple and then we can move onto the next duple. If we apply this algorithm to this embedding, almost every model we find is a model of a full board with a black vertical bar that passes the scope sentence. There are a few incomplete boards that do not pass the scope sentence. In the general case we can find many ``bad models'' that satisfy $R$ but do not satisfy the scope sentence. It is precisely here, each time we pick an atom using $R^{-}$, where grounding can be used, as a secondary step to prioritize atoms with compatible context constants, so next time we pick an atom for the next sentence in $R^{-}$ we start looking first into atoms with similar context constants.  Grounding may help to find subsets of atoms that satisfy the scope sentence but it is not the primary mechanism for selection of subsets of $F_C(R^{+}_{1})$. The primary mechanism is $R^{-}$.    

Note also that $\vert R^{-}\vert = 5$, in a $2 \times 2$ grid. This means that there should be subsets of non-redundant atoms of $F_C(R^{+}_{1})$ with at most $5$ atoms (perhaps 6 if we want to add $\ominus_c$ to be consistent with the sixth axiom of atomized semilattices). For example, the model: \[ 
M_s = [\phi_{b_2h_2}, \phi_{b_4h_4}, \phi_{w_1w_3n_1g_1g_3}]
\]
is a subset of the non-redundant atoms of $F_s$ that satisfy $R \wedge  \Xi$. In addition $[M_s + \{\ominus_c\}] \subset F_s$, as $F_s$ is freer than $[M_s + \{\ominus_c\}]$. How does it work?:\[ 
M_s \models (q_1 = q_2 = q_3 = q_4 = b_1 = b_3 = w_2 = w_4)
\]
The embedding can be extended to produce models for which the $q$ constants do not equate each others:\[ 
R^{-} \wedge (q_1 \not\leq q_2 \odot q_3 \odot q_4) \wedge (q_2 \not\leq q_1 \odot q_3 \odot q_4) \wedge (q_3 \not\leq q_1 \odot q_2 \odot q_4)  \wedge (q_4 \not\leq q_1 \odot q_2 \odot q_3)
\]
Now we have $\vert R^{-}\vert = 9$ and, in fact the following 7 atoms suffice:  \[ 
N_s = [\phi_{b_2h_2}, \phi_{b_4h_4}, \phi_{w_1b_1q_1}, \phi_{w_2b_2q_2}, \phi_{w_3b_3q_3}, \phi_{w_4b_4q_4},  \phi_{w_1w_3n_1g_1g_3}]
\]\[
N_s \models (q_1 = b_1) \wedge (q_2 = w_2)  \wedge (q_3 = b_3)  \wedge  (q_4 = b_4), 
\]
and form a model that distinguishes every constant $q_i$.

\bigskip
\bigskip

{\underline{Embedding 2:}}  Now we study a very different embedding based on presenting examples to the algebra. The interpretation constants $Q =  \{ b_{ij}, w_{ij}, q_{ij}\}$ and scope sentence $\Xi$ are the same as in the previous embedding. The embedding constants are now $C' = Q \cup \{ v, g_{ij}, h_{ij} \}$. The embedding set is: \[
R_{2A} = \{ v \leq I : I \text{ has a vertical black bar} \} \cup \{ v \not\leq I : I \text{ has no vertical black bar} \}.
\]
where $I$ is a term with $n \times n$ constants either $w_{ij}$ or $b_{ij}$ at each position of the grid.  For example in $2 \times 2$ is a conjunction of 16 atomic sentences:\[
R_{2A} = (v \leq b_{11}  \odot w_{12}  \odot b_{21}   \odot w_{22}) \wedge  (v \leq b_{11}  \odot b_{12}  \odot b_{21}   \odot w_{22}) ...\wedge 
\]\[
...\wedge (v \not\leq w_{11}  \odot w_{12}  \odot w_{21}   \odot w_{22}) \wedge  (v \not\leq b_{11}  \odot w_{12}  \odot w_{21}   \odot w_{22}) ... 
\]
We can add: \[
R_{2B} = (v  \leq \odot_{i,j} q_{i,j}) \wedge  \forall_{ij}  (w_{ij} \odot b_{ij}  \not\leq \odot_{r,s}  q_{rs}). 
\]
The embedding $R_{2A}  \wedge R_{2B}$ is concise. For completeness we have to add again: \[
R_{2B} =  \forall_{ij}  (q_{ij} \leq w_{ij} \odot b_{ij}).
\]
To make the embedding explicit and tight we proceed as before, by adding:  \[
R_{2C} =  \forall_{ij}  ((g_{ij} \odot q_{ij} = w_{ij} \odot g_{ij}) \wedge (h_{ij} \odot q_{ij} = b_{ij} \odot h_{ij})).
\]
The resulting embedding:\[
R_{2} = R_{2A} \wedge R_{2B} \wedge R_{2C}
\] 
produces a freest model  $F_{C'}(R^{+}_{2})$ with 68 atoms in dimension $2 \times 2$. Again, using $1,1  \rightarrow 1$, $1,2  \rightarrow 2$, $2,1  \rightarrow 3$ and $2,2  \rightarrow 4$ for better readability:  \[ 
F_{C'}(R^{+}_{2}) = [\phi_{g_1},\phi_{g_2},\phi_{g_3},\phi_{g_4},\phi_{h_1},\phi_{h_2},\phi_{h_3},\phi_{h_4},\phi_{w_1g_1},\phi_{w_2g_2},\phi_{w_3g_3},\]\[\phi_{w_4g_4},\phi_{b_1h_1},\phi_{b_2h_2},\phi_{b_3h_3},\phi_{b_4h_4},\phi_{w_1b_1q_1},\phi_{w_2b_2q_2},\]\[\phi_{w_3b_3q_3},\phi_{w_4b_4q_4},\phi_{b_1q_1g_1},\phi_{b_2q_2g_2},\phi_{b_3q_3g_3},\phi_{b_4q_4g_4},\phi_{w_1q_1h_1},\]\[\phi_{w_2q_2h_2},\phi_{w_3q_3h_3},\phi_{w_4q_4h_4},\phi_{vw_1b_1q_1},\phi_{vw_2b_2q_2},\phi_{vw_3b_3q_3},\phi_{vw_4b_4q_4},\]\[\phi_{vb_1b_2q_1g_1h_2},\phi_{vb_1b_2q_2g_2h_1},\phi_{vb_1b_4q_1g_1h_4},\phi_{vb_1b_4q_4g_4h_1},\phi_{vb_2b_3q_2g_2h_3},\]\[\phi_{vb_2b_3q_3g_3h_2},\phi_{vb_3b_4q_3g_3h_4},\phi_{vb_3b_4q_4g_4h_3},\phi_{vb_1b_2q_1q_2g_1g_2},\phi_{vb_1b_4q_1q_4g_1g_4},\]\[\phi_{vb_2b_3q_2q_3g_2g_3},\phi_{vb_3b_4q_3q_4g_3g_4},\phi_{vw_1w_2b_2q_1g_2h_1h_2},\phi_{vw_2w_3b_2q_3g_2h_2h_3},\]\[\phi_{vw_2b_2b_4q_4g_2g_4h_2},\phi_{vw_2w_4b_2q_4g_2h_2h_4},\phi_{vw_1w_2b_1q_2g_1h_1h_2},\phi_{vw_2w_3b_3q_2g_3h_2h_3},\]\[\phi_{vw_2b_1b_4q_2h_1h_2h_4},\phi_{vw_2b_3b_4q_2h_2h_3h_4},\phi_{vw_2w_4b_4q_2g_4h_2h_4},\]\[\phi_{vw_1w_3b_1q_3g_1h_1h_3},\phi_{vw_1b_1b_3q_3g_1g_3h_1},\phi_{vw_1w_4b_1q_4g_1h_1h_4},\phi_{vw_3b_1b_3q_1g_1g_3h_3},\]\[\phi_{vw_3b_1b_2q_3h_1h_2h_3},\phi_{vw_3b_1b_4q_3h_1h_3h_4},\phi_{vw_4b_1b_2q_4h_1h_2h_4},\phi_{vw_1w_3b_3q_1g_3h_1h_3},\]\[\phi_{vw_3w_4b_3q_4g_3h_3h_4},\phi_{vw_1b_3b_4q_1h_1h_3h_4},\phi_{vw_1b_2b_3q_1h_1h_2h_3},\phi_{vw_4b_2b_3q_4h_2h_3h_4},\]\[\phi_{vw_1w_4b_4q_1g_4h_1h_4},\phi_{vw_4b_2b_4q_2g_2g_4h_4},\phi_{vw_3w_4b_4q_3g_4h_3h_4}].
 \]
The freest solution model for the same solution than before $q_1= b_1, q_2 = w_2, q_3 = b_3, q_4 = w_4$, now has non-redundant atoms: \[
F'_s = [\phi_{g_1},\phi_{g_2},\phi_{g_3},\phi_{g_4},\phi_{h_1},\phi_{h_2},\phi_{h_3},\phi_{h_4},\phi_{w_1g_1},\phi_{w_3g_3},\phi_{b_2h_2},\]\[\phi_{b_4h_4},\phi_{w_1b_1q_1},\phi_{w_2b_2q_2},\phi_{w_3b_3q_3},\phi_{w_2q_2h_2},\phi_{w_4b_4q_4},\phi_{b_1q_1g_1},\]\[\phi_{w_4q_4h_4},\phi_{b_3q_3g_3},\phi_{vw_1b_1q_1},\phi_{vw_2b_2q_2},\phi_{vw_3b_3q_3},\phi_{vw_4b_4q_4},\]\[\phi_{vb_1b_2q_1g_1h_2},\phi_{vb_1b_4q_1g_1h_4},\phi_{vb_2b_3q_3g_3h_2},\phi_{vb_3b_4q_3g_3h_4}] 
 \]
which again is a subset of the atomization of $F_{C'}(R^{+}_{2})$ and, since the embedding is explicit:  \[ 
F'_s = F_{C'}(R^{+}_{2})^{{}^{\vee} Q \cup \{v\}  \cup  \{g_1, g_3, h_2, h_4\} } \cup \{\phi_{h_1}, \phi_{h_3}, \phi_{g_2}, \phi_{g_4}\} 
\]
There is something different in this embedding: \[
F' =  \cup_s F_s'  \sqsubset  F_{C'}(R^{+}_{2})  \,\,\,\,but\,\,\,\,  F' \not= F_{C'}(R^{+}_{2})
\]
and $F_{C'}(R^{+}_{2})$ has now 24 atoms more than $F'$, that only has 44 non-redundant atoms: \[ 
F' = [\phi_{g_1},\phi_{g_2},\phi_{g_3},\phi_{g_4},\phi_{h_1},\phi_{h_2},\phi_{h_3},\phi_{h_4},\phi_{w_1g_1},\phi_{w_2g_2},\phi_{w_3g_3},\]\[\phi_{w_4g_4},\phi_{b_1h_1},\phi_{b_2h_2},\phi_{b_3h_3},\phi_{b_4h_4},\phi_{w_1b_1q_1},\phi_{w_2b_2q_2},\]\[\phi_{w_3b_3q_3},\phi_{w_4b_4q_4},\phi_{b_1q_1g_1},\phi_{b_2q_2g_2},\phi_{b_3q_3g_3},\phi_{b_4q_4g_4},\phi_{w_1q_1h_1},\]\[\phi_{w_2q_2h_2},\phi_{w_3q_3h_3},\phi_{w_4q_4h_4},\phi_{vw_1b_1q_1},\phi_{vw_2b_2q_2},\phi_{vw_3b_3q_3},\phi_{vw_4b_4q_4},\]\[\phi_{vb_1b_2q_1g_1h_2},\phi_{vb_1b_2q_2g_2h_1},\phi_{vb_1b_4q_1g_1h_4},\phi_{vb_1b_4q_4g_4h_1},\phi_{vb_2b_3q_3g_3h_2},\]\[\phi_{vb_2b_3q_2g_2h_3},\phi_{vb_3b_4q_3g_3h_4},\phi_{vb_3b_4q_4g_4h_3},\phi_{vb_1b_2q_1q_2g_1g_2},\phi_{vb_1b_4q_1q_4g_1g_4},\]\[\phi_{vb_2b_3q_2q_3g_2g_3},\phi_{vb_3b_4q_3q_4g_3g_4}] 
\]
As with the first embedding, we can obtain every possible grid by grounding the model $F_{C'}(R^{+}_{2})$ to a subset of the constants. The negative duples of $R^{-}_{2}$ again select the subsets of atoms that spawn models of grids with a vertical black bar. In this case $R^{-}_{2A}$ is a set of duples that rule out each grid without a vertical bar individually. 

$R^{+}_{2}$ is a complete embedding, this means that any positive duple over $Q$ that is true in every solution of $P$ is a duple of $R^{+}_{2}$ or implied by $R^{+}_{2}$ . However, we have $F' \not=  F_{C'}(R^{+}_{2})$ because this embedding, although complete is not strongly complete. There are duples over $C$ that are true on every model but are not implied by $R^{+}$. Which ones?

We have required $v  \leq \odot_{i,j} q_{i,j}$. This sentence of $R^{+}_{2}$ is central as it says that we want models of grids that have a vertical bar $v$. However, if we replace in $v  \leq \odot_{i,j} q_{i,j}$ any $q_{i,j}$ by $b_{i,j}$ the sentence remains true on every solution. In addition, if we set two constants in a column to $b_{i,j}$ we can replace another $q_{i,j}$ by a white pixel constant $w_{i,j}$. These sentences turn out to be the missing piece. There are many of them. For $2 \times 2$ there are 18 sentences. If we make these sentences part of $R^{+}_{2}$ we get: $F' =  F_{C'}(R^{+}_{2})$ and the model $F'$ does not change. These sentences act by eliminating the extra 24 atoms that $F_{C'}(R^{+}_{2})$ had. 

There is a quite interesting fact that occurs in the strongly complete model. If we calculate the $16$ grounded models we obtain $7$ models of $R_{2}$ that represent grids with a vertical black bar and the other $9$ are models that violate $1, 1, 3, 1, 3, 1, 3, 3, 9$ duples of $R^{-}_{2A}$ respectively.  From the eight negative duples in $R^{-}_{2A}$ (that correspond to counterexamples of the problem) it suffices with $4$ duples to rule out all the grounded models that represent grids with no vertical black bar. With only $2$ negative duples we can rule out $6$ of the $8$ grounded models for well-formed grids that have no vertical bar. This hints a mechanism for generalization: fewer negative examples are enough to describe the embedding.

\bigskip
\bigskip

{\underline{Comparison:}}  We can start by calculating $F^{|Q}_s$ for solution $q_1= b_1, q_2 = w_2, q_3 = b_3, q_4 = w_4$. As both embeddings are concise, $F^{|Q}_s$ and $F'^{|Q}_s$ should be equal and, indeed:
\[
F^{|Q}_s = F'^{|Q}_s = [\phi_{w_1},\phi_{w_3},\phi_{b_2},\phi_{b_4},\phi_{w_2q_2},\phi_{b_1q_1},\phi_{w_4q_4},\phi_{b_3q_3}], 
\]
a calculation that can be easily done manually (remember to discard redundant atoms). 

Let us compare now the sets of atoms (not the models) produced by restricting the atoms of $F$ and $F'$ to the embedding constants:\[ 
F \xrightarrow{|Q}  A = \{\phi_{w_1},\phi_{w_2},\phi_{w_3},\phi_{w_4},\phi_{b_1},\phi_{b_2},\phi_{b_3},\phi_{b_4},\phi_{w_1q_1},\phi_{w_1w_3},\phi_{w_2w_4},\]\[\phi_{w_2q_2},\phi_{w_3q_3},\phi_{b_1q_1},\phi_{w_4q_4},\phi_{b_2q_2},\phi_{b_3q_3},\phi_{b_4q_4},\phi_{w_1b_1q_1},\]\[\phi_{w_1w_3q_1},\phi_{w_1w_3q_3},\phi_{w_2w_4q_2},\phi_{w_2w_4q_4},\phi_{w_2b_2q_2},\phi_{w_3b_3q_3},\phi_{w_4b_4q_4},\]\[\phi_{w_1w_3b_1q_1},\phi_{w_1w_3q_1q_3},\phi_{w_1w_3b_3q_3},\phi_{w_2w_4b_2q_2},\phi_{w_2w_4b_4q_4},\phi_{w_2w_4q_2q_4},\]\[\phi_{w_1w_3b_1q_1q_3},\phi_{w_1w_3b_3q_1q_3},\phi_{w_2w_4b_2q_2q_4},\phi_{w_2w_4b_4q_2q_4},\phi_{w_1w_3b_1b_3q_1q_3},\]\[\phi_{w_2w_4b_2b_4q_2q_4}\}
 \]\[ 
F' \xrightarrow{|Q} B = \{\phi_{w_1},\phi_{w_2},\phi_{w_3},\phi_{w_4},\phi_{b_1},\phi_{b_2},\phi_{b_3},\phi_{b_4},\phi_{w_1q_1},\phi_{w_2q_2},\phi_{w_3q_3},\]\[\phi_{b_1q_1},\phi_{w_4q_4},\phi_{b_2q_2},\phi_{b_3q_3},\phi_{b_4q_4},\phi_{w_1b_1q_1},\phi_{w_2b_2q_2},\]\[\phi_{w_3b_3q_3},\phi_{b_1b_2q_1},\phi_{b_1b_2q_2},\phi_{b_1b_4q_1},\phi_{w_4b_4q_4},\phi_{b_2b_3q_2},\phi_{b_2b_3q_3},\]\[\phi_{b_1b_4q_4},\phi_{b_3b_4q_3},\phi_{b_3b_4q_4},\phi_{b_1b_2q_1q_2},\phi_{b_1b_4q_1q_4},\phi_{b_2b_3q_2q_3},\phi_{b_3b_4q_3q_4}\}.
\]
Many atoms are the same, but others are not. Consider the set: \[ 
A - B = \{\phi_{w_1w_3},\phi_{w_2w_4},\phi_{w_1w_3q_1},\phi_{w_1w_3q_3},\phi_{w_2w_4q_2},\phi_{w_2w_4q_4},  \]\[ \phi_{w_1w_3b_1q_1},\phi_{w_1w_3q_1q_3},\phi_{w_1w_3b_3q_3},\phi_{w_2w_4b_2q_2},\phi_{w_2w_4b_4q_4},\phi_{w_2w_4q_2q_4},\]\[\phi_{w_1w_3b_1q_1q_3},\phi_{w_1w_3b_3q_1q_3},\phi_{w_2w_4b_2q_2q_4},\phi_{w_2w_4b_4q_2q_4},\phi_{w_1w_3b_1b_3q_1q_3},\phi_{w_2w_4b_2b_4q_2q_4}\}
\]
Since both embeddings are concise and tight, theorem \ref{sameSolutionAtomsTheoremPartC} tells us that every one of these atoms should be redundant in at least one solution of $P$. For example for $\phi_{w_1w_3}$ to be redundant in a solution, the solution should satisfy $w_1 \not\leq q_1$ and $w_3 \not\leq q_3$ and it should have atoms $\phi_{w_1}$ and $\phi_{w_3}$. From the scope sentence we know that such solution should satisfy: $q_1 = b_1$ and  $q_3 = b_3$; like the ones we have calculated above. Atoms $\phi_{w_1w_3}$ and $\phi_{w_2w_4}$ are redundant on each solution that has a black vertical bar on the left or a black vertical bar on the right respectively.  Atoms $\phi_{w_1w_3b_1b_3q_1q_3}$ and $\phi_{w_2w_4b_2b_4q_2q_4}$ are redundant on every solution. The other atoms are redundant in other solutions. Some atoms like $\phi_{w_2w_4b_2q_2}$ are redundant only in solutions that have some particular set of white pixels. 

In the other direction:
\[ 
B - A= \{\phi_{b_1b_2q_1},\phi_{b_1b_2q_2},\phi_{b_1b_4q_1},\phi_{b_2b_3q_2},\phi_{b_2b_3q_3},\]\[\phi_{b_1b_4q_4},\phi_{b_3b_4q_3},\phi_{b_3b_4q_4},\phi_{b_1b_2q_1q_2},\phi_{b_1b_4q_1q_4},\phi_{b_2b_3q_2q_3},\phi_{b_3b_4q_3q_4}\},
\]
we find atoms that are redundant in solutions that contain a vertical bar and at least one or more additional white pixel (atoms with three constants), or one additional black pixel (atoms with four constants). Finally:
\[ 
A \cap B = \{\phi_{w_1},\phi_{w_2},\phi_{w_3},\phi_{w_4},\phi_{b_1},\phi_{b_2},\phi_{b_3},\phi_{b_4},\phi_{w_1q_1},\]\[\phi_{w_2q_2},\phi_{w_3q_3},\phi_{b_1q_1},\phi_{w_4q_4},\phi_{b_2q_2},\phi_{b_3q_3},\phi_{b_4q_4},\phi_{w_1b_1q_1},\]\[\phi_{w_2b_2q_2},\phi_{w_3b_3q_3},\phi_{w_4b_4q_4}\}
 \]
contains non-redundant atoms, each in both $F_S$ and $F'_S$ for some solution $S$. Particularly interesting are $\phi_{w_1b_1q_1},\phi_{w_2b_2q_2},\phi_{w_3b_3q_3},\phi_{w_4b_4q_4}$ that are non-redundant atoms of the freest solution model of every solution and represent the structure of the grid: these atoms say that at each position $p$ we have either the white or the black pixel and the corresponding $q_p$ constant.  

\bigskip

What about the $24$ atoms of $F_{C'}(R^{+}_{2})$ that are external to $F'_s$ \[ 
F_{C'}(R^{+}_{2}) - F'_s = \{ \phi_{vw_1w_2b_2q_1g_2h_1h_2},\phi_{vw_2w_3b_2q_3g_2h_2h_3},\phi_{vw_2b_2b_4q_4g_2g_4h_2},\phi_{vw_2w_4b_2q_4g_2h_2h_4},\]\[\phi_{vw_1w_2b_1q_2g_1h_1h_2},\phi_{vw_2w_3b_3q_2g_3h_2h_3},\phi_{vw_2b_1b_4q_2h_1h_2h_4},\phi_{vw_2b_3b_4q_2h_2h_3h_4},\]\[\phi_{vw_2w_4b_4q_2g_4h_2h_4},\phi_{vw_1w_3b_1q_3g_1h_1h_3},\phi_{vw_1b_1b_3q_3g_1g_3h_1},\phi_{vw_1w_4b_1q_4g_1h_1h_4},\]\[\phi_{vw_3b_1b_3q_1g_1g_3h_3},\phi_{vw_3b_1b_2q_3h_1h_2h_3},\phi_{vw_3b_1b_4q_3h_1h_3h_4},\phi_{vw_4b_1b_2q_4h_1h_2h_4},\]\[\phi_{vw_1w_3b_3q_1g_3h_1h_3},\phi_{vw_3w_4b_3q_4g_3h_3h_4},\phi_{vw_1b_3b_4q_1h_1h_3h_4},\phi_{vw_1b_2b_3q_1h_1h_2h_3},\]\[\phi_{vw_4b_2b_3q_4h_2h_3h_4},\phi_{vw_1w_4b_4q_1g_4h_1h_4},\phi_{vw_4b_2b_4q_2g_2g_4h_4},\phi_{vw_3w_4b_4q_3g_4h_3h_4} \} 
\]
Their restrictions to $Q$ are: \[ 
F_{C'}(R^{+}_{2}) - F'_s  \xrightarrow{|Q} \{ \phi_{w_1w_2b_2q_1},\phi_{w_2w_3b_2q_3},\phi_{w_2b_2b_4q_4},\]\[\phi_{w_2w_4b_2q_4},\phi_{w_1w_2b_1q_2},\phi_{w_2w_3b_3q_2},\phi_{w_2b_1b_4q_2},\phi_{w_2b_3b_4q_2},\]\[\phi_{w_2w_4b_4q_2},\phi_{w_1w_3b_1q_3},\phi_{w_1b_1b_3q_3},\phi_{w_1w_4b_1q_4},\phi_{w_3b_1b_3q_1},\]\[\phi_{w_3b_1b_2q_3},\phi_{w_3b_1b_4q_3},\phi_{w_4b_1b_2q_4},\phi_{w_1w_3b_3q_1},\phi_{w_3w_4b_3q_4},\]\[\phi_{w_1b_3b_4q_1},\phi_{w_1b_2b_3q_1},\phi_{w_4b_2b_3q_4},\phi_{w_1w_4b_4q_1},\phi_{w_4b_2b_4q_2},\phi_{w_3w_4b_4q_3} \} 
\]
Atom $\phi_{w_1w_2b_2q_1}$ is not compatible with the scope sentence $\Xi$ as it implies that neither $w_2$ nor $b_2$ can be equal to $q_2$. 16 of these atoms follow this pattern. The restrictions to $Q$ of the other 8 atoms are: \[ 
\phi_{w_2b_1b_4q_2},\phi_{w_2b_3b_4q_2},\phi_{w_3b_1b_2q_3},\phi_{w_3b_1b_4q_3},\phi_{w_4b_1b_2q_4},\phi_{w_1b_3b_4q_1},\phi_{w_1b_2b_3q_1},\phi_{w_4b_2b_3q_4}
\]
Each of these atoms belong to one of the models of grids with no vertical black bars. They, each, set the minimal requirement of having at least three white pixels. The $8$ atoms of $F_{C'}(R^{+}_{2}) - F'_s$ from which they come are:
\[ 
\phi_{vw_2b_1b_4q_2h_1h_2h_4},\phi_{vw_2b_3b_4q_2h_2h_3h_4},\phi_{vw_3b_1b_2q_3h_1h_2h_3},\phi_{vw_3b_1b_4q_3h_1h_3h_4},\]\[\phi_{vw_4b_1b_2q_4h_1h_2h_4},\phi_{vw_1b_3b_4q_1h_1h_3h_4},\phi_{vw_1b_2b_3q_1h_1h_2h_3},\phi_{vw_4b_2b_3q_4h_2h_3h_4}
\]
These atoms are never atoms of a solution. To get rid of them either we make the second embedding strongly complete or we discard them based on their incompatibility with $R^{-}_{2}$ and the scope sentence.

In conclusion: the embedding $1$ is strongly complete while the embedding $2$ is not, unless we add $18$ extra duples to the embedding. Assume we are not making embedding 2 strongly complete. We can discover solutions by inspecting the atoms of the embedding 1 that, when restricted to $Q$, are not atoms of the embedding 2. In the other direction we can also discover solutions by inspecting the atoms of the embedding 2 that are not atoms of the embedding 1 when restricted to $Q$, however, in this direction, some of these atoms may mislead us into forming bad solutions. These solutions can only be ruled out by checking $R^{-}_{2}$ and the scope sentence.

\section{The N-Queen Completion Problem} 

This problem consists of locating $M$ queens on a chessboard of size $M$ in such way that no two queens can attack each other. That is, they all occupy different columns, rows and diagonals.

For an $M \times M$ board we choose a set of $M \times M$ constants $Q_{ij}$ to represent a queen in position $(i, j)$ and a set of $M \times M$ constants $E_{ij}$ for the empty squares. An additional constant $B$ can be used to represent the board and to determine if the board has an empty square or a queen we can use the order relation $<$ of the semilattice. The following sentences provide an interpretation for this problem:
\[ 
{\bf{\alpha}} = "\forall i \forall j ((Q_{ij} \leq B) \vee (E_{ij} \leq B))"  
\]\[ 
{\bf{\beta}} = "\forall i \forall j (Q_{ij} \odot E_{ij} \not \leq  B)"   
\]\[ 
{\bf{\gamma}} = "\forall i \exists j  (Q_{ij}  \leq  B)   \wedge   \forall j \exists i  (Q_{ij}  \leq  B)" 
\]\[ 
{\bf{\epsilon}} = "\forall i \forall j   \, ( (Q_{ij} \in I) \Rightarrow  (Q_{ij}  \leq  B))"
\]\[ 
{\bf{\pi}} = "\forall i \forall j  \forall r \forall s  \,(((Q_{ij}, Q_{rs}) \in A) \Rightarrow  (Q_{ij} \odot Q_{rs} \not \leq  B))"  
\]
where $A$ is the set of pairs of mutually attacking queens, and $I$ is the set of $n$ fixed queens at the beginning of the game. 

The interpretation constants are $Q = \{ Q_{ij}, E_{ij}, B\}$ and the interpretation sentence is:
\[ 
{\bf{\varphi}} = {\bf{\alpha}}  \wedge {\bf{\beta}} \wedge {\bf{\gamma}}  \wedge {\bf{\epsilon}}  \wedge {\bf{\pi}}. 
\]
The scope sentence is:
\[ 
\Xi = {\bf{\alpha}}.
\]
Any semilattice model over the constants $Q$ that satisfy ${\bf{\varphi}}$ is a solution of the game. We can determine the position of the queens in the board by using:
\[ 
Q_{ij} \leq B  \, \Leftrightarrow \,  \text{the board has a queen at position (i, j)}
\]
The separator set is: 
\[ 
\Gamma =   \{ (Q_{ij} \leq B) : 0 \leq i < M, 0 \leq j < M  \}  \cup \{ (E_{ij} \leq B) : 0 \leq i < M, 0 \leq j < M  \}  
\]
We are going to build a couple of embeddings with the interpretation $(Q, {\bf{\varphi}}, \Gamma)$
\bigskip

\underline{First embedding}:  Consider the set  $C =\{ Q_{ij}, E_{ij}, B, R_i, C_j, g_{ij}, h_{ij} \}$ of $4M^{2} + 2M + 1$ embedding constants.  ${\bf{\beta}}$ can directly be made part of $R^{-}$ as it is a conjunction of $M^2$ negated atomic sentences. Same consideration applies to ${\bf{\pi}}$ that is a conjunction of fewer than $4(M - 1)M^2$ negative atomic sentences.

The encoding of ${\bf{\gamma}}$  can be carried out with the set $RC$ of $M^2$ atomic sentences:
\[ 
R_i \odot C_j  \leq  Q_{ij},
\]
the set $NRC$ of $2M$ negated atomic sentences:
\[ 
R_x  \not<  (\odot_{i,j}  E_{ij})  \odot  (\odot_{i,j: i \not= x}  Q_{ij}) ,
\]\[ 
C_y  \not<  (\odot_{i,j}  E_{ij})  \odot  (\odot_{i,j:j \not= y}  Q_{ij}),
\]
and the additional sentence:
\[ 
\delta =  "\odot_x  R_x  \odot_y  C_y  \leq  B".
\]
${\bf{\epsilon}}$ can be encoded in a straightforward manner as $n$ atomic sentences of the form $(Q_{ij}  \leq  B)$, one for each of the fixed queens, forming a block $IB$ of $n$ positive atomic sentences that can be added to $R^{+}$. 

Based on our separator set $\Gamma$, we can use a set $GH$ of $2M^2$ atomic sentences to make sure the embedding is explicit:
\[ 
g_{ij}  \odot Q_{ij}   \leq  B \odot  g_{ij},
\]\[ 
h_{ij}  \odot E_{ij}   \leq  B \odot  h_{ij},
\]
These sentences can also be seen as a mechanism to embed ${\bf{\alpha}}$. 

Finally, we get the embedding set:
\[ 
R^{+} = GH \cup RC \cup IB \cup \{ \delta \}
\]\[ 
R^{-} = \{ {\bf{\beta}},  {\bf{\pi}} \} \cup NRC
\]
that has $3M^2 + n + 1$ atomic sentences in $R^{+}$ and fewer than $M^2 + 4(M - 1)M^2 + 2M$ negated atomic sentences in $R^{-}$.

\bigskip

The embedding given by $R = R^{+} \cup R^{-}$ has various models, including models for all the full boards and also models that are not full boards. There are models of $R$ that represent a board with the first initial $N$ queens set and no more queens.  $R$ also has models that do not even have empty or queen values set for all the board positions. However, there are models of $R$ for any solution of the problem. It is clear than any non-attacking configuration of queens satisfies $R$, including all the solutions with $M$ queens. 

For a game with no fixed queens, the freest model $F_C(R^{+})$ contains $2M^2 + 1$ atoms with one constant in its upper constant segment: either $B$, or any of the $M^2$ constants $g_{ij}$ or any of the $M^2$ constants $h_{ij}$. Also $4M^2$ atoms with two constants in their upper constant segment, with the form $\{Q_{ij}, g_{ij} \}$, $\{E_{ij}, h_{ij} \}$, $\{Q_{ij}, B \}$, $\{E_{ij}, B \}$. Also $2M$ atoms with $M + 2$ constants in their upper constant segment, each corresponding to a row or a column and containing the queens in the row or column, the corresponding $C_x$ or $R_y$ and the constant $B$. A total of $6M^2 + 2M + 1$ non-redundant atoms:

\begin{table}[!htbp]
\centering
\begin{tabular}{cc|cc|cc}
\multicolumn{2}{c}{$F_{C}(R^{+}$)}  & \multicolumn{2}{|c|}{\,\,\,\,\,\,\,\,\,\,} & \multicolumn{2}{c}{$F_S$}         \\
 \hline
Size & \#                    & \multicolumn{2}{|c|}{\,\,\,\,\,\,\,\,\,\,}                  & Size & \#                    \\
\hline
1    & $2M^2 + 1$ & \multicolumn{2}{|c|}{\,\,\,\,\,\,\,\,\,\,}                  & 1    & $2M^2 + 1$ \\
2    & $4M^2$ & \multicolumn{2}{|c|}{\,\,\,\,\,\,\,\,\,\,}                  & 2    & $3M^2$ \\
M+2  & 2M                    & \multicolumn{2}{|c|}{\,\,\,\,\,\,\,\,\,\,}                  & M+2  & 2M                   
\end{tabular}
\end{table}

The embedding is strongly complete, i.e: 
\[
F =  \cup_s F_s  =  F_{C}(R^{+}) 
\]
except for games for which every solution has the same queen or empty space. For example, in dimension $M = 5$ with no fixed queens the embedding is strongly complete as every square can be empty or occupied in some solution. In dimension $M = 4$ with no fixed queens the embedding is not strongly complete as there are 8 empty squares present in the two solutions of the 4x4 puzzle (we have $F  \sqsubset  F_{C}(R^{+})$ in this case). Enforcing the presence of these 8 empty squares in the board makes the embedding strongly complete. 

The embedding is concise, which can be verified by checking that every solution satisfies $F_S^{|Q} = F_C(S)^{|Q}$, a model with $M^2 + 1$ atoms of length 1 and $M^2$ atoms of length 2. The embedding is also tight and explicit, as expected.

\bigskip

\underline{Second embedding}: A better embedding for this problem relies on encoding ${\bf{\pi}}$ using a set of positive sentences instead of a set of negative sentences. We can simply use the fact that the presence of a queen in the board implies the presence of empty squares in all its attacking positions, and also the fact that a row or column with all empty positions except one implies a queen in the remaining position.

The embedding constants are again $C =\{ Q_{ij}, E_{ij}, B, R_i, C_j, g_{ij}, h_{ij} \}$, the same constants we used in the first embedding. To embed ${\bf{\pi}}$ we now use the set EQB of $M^2$ positive duples:
\[ 
\odot_{rs} \{ E_{rs} : (Q_{rs},Q_{ij})  \in A  \} \leq Q_{ij} \odot B 
\]
and  the set QEB with $2M^2$ positive duples:
\[ 
Q_{ij} \leq  \odot_{s} \{ E_{is} : s \not= j  \}  \odot B
\]\[ 
Q_{ij} \leq  \odot_{r} \{ E_{rj} : r \not= i  \}  \odot  B.
\]
In total, $3M^2$ positive atomic sentences. Notice that ${\bf{\pi}}$ is a conjunction of a much larger set of sentences, in the order of $4(M - 1)M^2$ negative atomic sentences, that we no longer need to include in $R^{-}$. The embedding set is then:
\[ 
R^{+} = GH \cup RC \cup IB \cup EQB  \cup QEB  \cup \{ \delta \}
\]\[ 
R^{-} =\{  {\bf{\beta}} \} \cup NRC
\]
with $6M^2 + n + 1$ atomic sentences in $R^{+}$ and  $M^2 + 2M$ negated atomic sentences in $R^{-}$.

\bigskip 

For a game with no fixed queens, the freest model $F = \cup_s F_s$ has $2M^2 + 1$ atoms with one constant in their upper constant segment: either $B$, or any of the $M^2$ constants $g_{ij}$ or any of the $M^2$ constants $h_{ij}$. Also $2M^2$ atoms with two constants in their upper constant segment, of the form $\{Q_{ij}, B \}$ or $\{E_{ij}, B \}$. Also $2M$ atoms with $M + 2$ constants in their upper constant segment, each corresponding to a row or a column and containing the queens in the row or column, the corresponding $C_x$ or $R_y$ and the constant $B$. $F_C(R^{+})$ also contains as many atoms of size $2M^2$ as solutions the game has, each atom with an upper constant segment that contains the inverse image of the full board of the solution, with the role of queens and empty spaces interchanged, each $Q_{ij}$ and $E_{ij}$ constant accompanied by its corresponding $g_{ij}$ or $h_{ij}$.

\begin{table}[!htbp]
\centering
\begin{tabular}{cc|cc|cc}
\multicolumn{2}{c}{$F$}       
& \multicolumn{2}{|c|}{\,\,\,\,\,\,\,\,\,\,}  & \multicolumn{2}{c}{$F_S$}                             \\
\hline
\multicolumn{1}{c}{Size} & \#                        & \multicolumn{2}{|c|}{\,\,\,\,\,\,\,\,\,\,}                  & \multicolumn{1}{c}{Size} & \#                        \\
 \hline
1                        & $2M^2 + 1$ & \multicolumn{2}{|c|}{\,\,\,\,\,\,\,\,\,\,}                  & 1                        & $2M^2 + 1$ \\
2                        & $2M^2$     & \multicolumn{2}{|c|}{\,\,\,\,\,\,\,\,\,\,}                   & 2                        & $2M^2$     \\
M+2                      & 2M                        & \multicolumn{2}{|c|}{\,\,\,\,\,\,\,\,\,\,}                  & M+2                      & 2M                        \\
$2M^2$    & \#Solutions               &                   &                   & $2M^2$    & 1                        
\end{tabular}
\end{table}

The total number of non-redundant atoms in $F = \cup_s F_s$ is:
\[
1 + 2M + 4M^2 + \#Solutions. 
\]

In dimension $M = 5$, for example, the number of atoms in $F$ is $121$ and there are $10$ atoms of size $50$ corresponding, each, to one of the $10$ solutions of the game. 

Each freest solution model $F_S$ is spawned by the non-redundant atoms of $F_C(R^{+})$ with upper constant segment of size smaller than $2M^2$ and the atom of size $2M^2$ that corresponds to the solution $S$ so, in dimension $5$, the freest solution models have each $112$ non-redundant atoms. In general, for every problem size $M$ the freest solution models contain exactly $2 + 2M + 4M^2$ atoms.

The model $F_C(R^{+})$ has the non-redundant atoms of $F$ and, in addition, many residual atoms of larger size than $2M^2$. In dimensions 4 and 5 we have the spectra in the table below.
\begin{table}[!htbp]
\centering
\begin{tabular}{cc|c|cc}
\multicolumn{2}{c|}{$F_C(R^{+})$\,\,\, M=4} &  & \multicolumn{2}{c}{$F_C(R^{+}$) \,\,\, M =5} \\
 \hline
Size & \#            & & Size & \#              \\
 \hline
1  & 33              & & 1  & 51                \\
2  & 32              & & 2  & 50                \\
6  & 8               & & 7  & 10                \\
32 & 2 (\#Solutions) & & 50 & 10 (\#Solutions)  \\
44 & 4               & & 62 & 8                 \\
48 & 89              & & 64 & 8                 \\
50 & 744             & & 66 & 100               \\
52 & 88              & & 68 & 432               \\
   &                 & & 70 & 2020              \\
   &                 & & 72 & 38548             \\
   &                 & & 74 & 33252             \\
   &                 & & 76 & 942                             
\end{tabular}
\end{table}

The presence of residual atoms implies the embedding is not complete, which is the case in every dimension. However, an inspection of such residual atoms permits finding additional atomic sentences that, together with $R^{+}$, make the embedding strongly complete.  

For example, in dimension $M = 4$ with no fixed queens, enforcing the presence of the 8 squares that are empty in both solutions make the embedding strongly complete by eliminating all the atoms of size larger than 32, obtaining then $F =  F_{C}(R^{+})$ where now $R^{+}$ contains the extra atomic sentence:
\[
E_{11} \odot E_{14} \odot E_{22} \odot E_{23} \odot E_{32} \odot E_{33} \odot E_{41} \odot E_{44} \leq B.
\]

In dimension $M = 5$ with no fixed queens there are no queens or empty squares present in every solution. However, there is a set of positive atomic sentences of the form:
\[
E_{ab} \leq Q_{cd} \odot B
\]\[
E_{ef} \leq E_{ij} \odot  E_{kl} \odot B
\]\[
E_{mn} \leq Q_{pq} \odot  E_{rs} \odot B
\]
that, once added to the set $R^{+}$, make the embedding strongly complete, i.e equal to the model $F$ in the table above. 

The reader may wonder why the first embedding is complete in dimension 5 while the second embedding is not. In the second embedding every one of these extra atomic sentences that make the embedding strongly y complete is true in $F$ and in every freest solution model. In the first embedding every one of these extra atomic sentences is false in $F$ and is false in at least one of the freest solution models. Observe that in the first embedding, for a solution $S$ that does not have the queen $Q_{cl}$ in the board weather or not $Q_{cd} \odot B$ implies $E_{ab}$ is a choice that can be made, so the freest solution model should satisfy $F_S \models (E_{ab} \not\leq Q_{cd} \odot B)$. This is not the case in the second embedding for which the positive atomic sentences used to enforce ${\bf{\pi}}$ imply $F_S \models (E_{ab} \leq Q_{cd} \odot B)$.    

For a similar reason this second embedding is not concise either. Sentences in blocks EQB and QEB imply sentences in the interpretation constants $Q$ that are not implied by the interpretation sentences ${\bf{\varphi}}_s$. For example, the sentence $\odot_{rs} \{ E_{rs} : (Q_{rs},Q_{ij})  \in A  \} \leq Q_{ij} \odot B$ is true in every solution and implied by ${\bf{\varphi}}_s$ as long as $Q_{ij}$ is a queen in the board in said solution. However, if $Q_{ij}$ is not a queen in the board of the solution then this sentence is not implied by the chosen interpretation and, hence, the embedding is not concise. As a result, for every solution $S$, the freest solution model $F^{|Q}_{2S}$ given by this embedding is not equal to the freest solution model $F^{|Q}_{1S}$ of the first embedding (which is concise). However, it is true that $F^{|Q}_{1S}$ is strictly freer than $F^{|Q}_{2S}$ and, then, every atom of $F^{|Q}_{2S}$ is an atom in $F^{|Q}_{1S}$. Weather or not an embedding is concise depends on the interpretation sentence chosen. 

By calculating the freest model $F_C(R^{+})$ we could find all the solutions with or without fixed queens on any dimension by just looking at the atoms of size $2M^2$. The number of calculations that are necessary in a classical computer to calculate the atomization using full crossing is, unfortunately, quite large; even in low dimension billions of atoms appear during the calculation of $F_C(R^{+})$, a calculation that may take many hours to compute. However using a sparse version of the full crossing operation (see \cite{AMLPresentation}) the calculation can be carried out under a second.  If we increase the dimension a little bit the number of atoms that appear during the calculation with full crossing reaches astronomical numbers. However, solutions can still be found very fast using sparse crossing, so this embedding can be used in practice to find solutions of the N-Queen completion problem. 

\underline{Third embedding}: the second embedding is tight and explicit and we also know how to make it strictly complete and, therefore, also complete. However, unlike the first embedding, the second embedding is not concise. We can easily make our second embedding concise with the help of one additional embedding constant $U$.

As we mentioned above, our second embedding is not concise because a solution that does not have $Q_{ij}$ in the board does not imply $\odot_{rs} \{ E_{rs} : (Q_{rs},Q_{ij})  \in A  \} \leq Q_{ij} \odot B$ and we are enforcing this sentence as part of the encoding of ${\bf{\pi}}$ with the block EQB of positive sentences. We can address this issue by replacing EQB with the sentences EQBU:
\[ 
\odot_{rs} \{ E_{rs} : (Q_{rs},Q_{ij})  \in A  \} \leq Q_{ij} \odot B \odot U
\]
and by replacing the set QEB with the set QEBU:
\[ 
Q_{ij} \leq  \odot_{s} \{ E_{is} : s \not= j  \}  \odot B  \odot U
\]\[ 
Q_{ij} \leq  \odot_{r} \{ E_{rj} : r \not= i  \}  \odot  B \odot U.
\]
Since $U$ is not an interpretation constant these sentences do not have consequences in the space spawned by the interpretations constants, and now we have a concise embedding $F_S^{|Q} = F_C(S)^{|Q}$. 
To find solutions using this embedding we have to be careful to replace the negative atomic sentence ${\bf{\beta}} = \forall i \forall j (Q_{ij} \odot E_{ij} \not \leq  B)$ with the set ${\bf{\beta}}U$ of $M^2$ negative atomic sentences defined as:
\[ 
Q_{ij} \odot E_{ij} \not \leq  B \odot U.
\]
The third embedding set is then:
\[ 
R^{+} = GH \cup RC \cup IB \cup EQBU  \cup QEBU  \cup \{ \delta \}
\]\[ 
R^{-} =  {\bf{\beta}}U \cup NRC,
\]
that finds solutions to the problem as well as the second embedding. 

With no queens set in the board at the beginning of the game, the atomizations of $F_C(R^{+})$ for the third and second embeddings differ only in a set of one atom of size 1 and $2M^2$ atoms of size 3 present in the atomization of the third embedding and absent in the second embedding. The atom of size 1 has the constant $U$ in its upper constant segment and the atoms of size 3 have, each, the constant $U$ and either one constant $Q_{i, j}$ and its corresponding context constant $g_{i, j}$ or one constant $E_{i, j}$ and its corresponding $h_{i, j}$. All the other atoms are the same.

The atomizations for the freest solution models $F_S$ corresponding to the second and third embeddings differ in the same atom of size 1 with $U$ in its upper constant segment and $M^2$ atoms of size 3 with constant $U$ and either $Q_{i, j}$ and its corresponding context constant $g_{i, j}$ if $E_{i, j}$ is present in the solution $S$ or $E_{i, j}$ and its corresponding $h_{i, j}$ if $Q_{i, j}$ is present in $S$.

\begin{table}[!htbp]
\centering
\begin{tabular}{cc|c|cc}
\multicolumn{2}{c|}{$F$}                 & & \multicolumn{2}{c}{$F_S$} \\
\hline
Size   & \#                              & & Size   & \#         \\
 \hline
1      & $2M^2 + 2$                      & & 1      & $2M^2 + 2$ \\
2      & $2M^2$                          & & 2      & $2M^2$     \\
3      & 24 if $M = 4$, $2M^2$ otherwise & & 3      & $M^2$      \\
$M+2$  & 2M                              & & M+2    & 2M         \\
$2M^2$ & \#Solutions                     & & $2M^2$ & 1                        
\end{tabular}
\end{table}
   
For the atomization of $F = \cup_S F_S$ we again see in the third embedding an extra atom of size 1 and the $2M^2$ atoms of size 3 that $F_C(R^{+})$ had, except on dimension $M = 4$ that we only get $24 = 2M^2 - 8$ atoms with 8 atoms missing corresponding with the 8 empty cells present in both solutions of the game.

So far we have seen that the atomization of the model $F_C(R^{+})$ for the second embedding and the atomization of $F_C(R^{+})$ for the third embedding are quite similar to each other while the atomization of the first embedding is very different than the other two. However, the atomizations of the restricted model $F_C(R^{+})^{|Q}$ for the first and third embeddings are identical while the atomization of this model for the second embedding is different. The same occurs for the atomization of $F_S^{|Q}$ or $F^{|Q}$.  This is indeed expected from the fact that the first and third embeddings are both concise and from theorems \ref{conciseEmbeddingTheorem}, \ref{weakEquivalenceTheorem} and \ref{equivalenceTheorem}.

The freest solution models restricted to the embedding constants $F_S^{|Q}$ are equal for the first and third embeddings and have the following spectra of atom sizes:

\begin{table}[!htbp]
\centering
\begin{tabular}{c|cc|c|cc|c|cc}
$F_S^{|Q}$ & \multicolumn{2}{c|}{EMB 1} &  & \multicolumn{2}{c|}{EMB 2} &  & \multicolumn{2}{|c}{EMB 3} \\
\hline
    & Size & \#        &  & Size & \#      &  & Size & \#         \\
\hline
    & 1    & $M^2 + 1$ &  & 1     & 1      &  & 1    & $M^2 + 1$  \\
    & 2    & $M^2$     &  & 2     & $2M^2$ &  & 2    & $M^2$      \\
    &      &           &  & $M^2$ & 1      &  &      & 
\end{tabular}
\end{table}
% https://table.6cm.co/

For the first and third embeddings, the model $F_C(R^{+})^{|Q}$ is equal to the freest model over the embedding constants $F_Q(\emptyset)$ that has $M^2 + 1$ atoms of size 1. However, for the second embedding the model $F_C(R^{+})^{|Q}$ has many atoms, $960$ for $M = 4$ and many more for higher dimensions. This atomization is actually very similar to that of $F_C(R^{+})$ but with smaller atom sizes and a single atom of size 1.

Finally, model $F^{|Q}$ is also equal for embeddings 1 and 3, and for $M > 4$, again equal to $F_Q(\emptyset)$. For embedding 2 model $F^{|Q}$ has one atom of size 1, $M^2$ atoms of size 2 and as many atoms of size $2M^2$ as solutions.

Since $F_C(R^{+})^{|Q} = F^{|Q}$ we can use theorem \ref{residualAtomTheorem} to state that the third embedding, just like the first, is complete. However, while the first embedding is also strictly complete for $M > 4$ the third embedding is not as implied by $F_C(R^{+}) \not= F$. To make the third embedding strictly complete we will have to add the same additional sentences than in embedding 2 with the difference that, now, the extra constant $U$ appears in the right hand side.

\section{Sudoku}

For the Sudoku it is possible to use an embedding strategy very similar to that used to embed the N-Queens Completion problem.  For the classical Sudoku game in a $9 \times 9$ grid with 81 cells, we can define 9 constants $N_{nij}$ for each of the $9$ possible numbers written in a cell $(i, j)$, i.e. 729 constants $N_{nij}$. We can also define a constant $G$ for the grid and the order relation $\leq$ can be used to determine if the grid has a number in a cell. An interpretation can be built with the following sentences:
\[ 
{\bf{\alpha}} = "\forall i \forall j \vee_n(N_{nij} \leq G)"
\]\[ 
{\bf{\beta}} = "\forall n \forall m (m \not = n)  \forall i \forall j  (N_{nij} \odot N_{mij} \not\leq G)"   
\]\[ 
{\bf{\gamma}} = "\forall n  \forall i \exists j  (N_{nij} \leq G)   \wedge   \forall n  \forall j \exists i  (N_{nij} \leq G)   \wedge   \forall n  \forall \text{z} \exists i \exists j ((i,j) \in \text{cells(z)}) (N_{nij} \leq G)"
\]\[ 
{\bf{\epsilon}} = "\forall n \forall i \forall j   \, ( (N_{nij}  \in I) \Rightarrow  (N_{nij} \leq G))"
\]\[ 
{\bf{\pi}} =  "\forall n  \forall i \forall j  \forall r \forall s  \,(((N_{nij}, N_{nrs}) \in A) \Rightarrow  (N_{nij} \odot N_{nrs}\not \leq G))"
\]
where $z$ is an index for the nine $3\times 3$ subgrids of the sudoku, $cells(z)$ is the set of square cells in the subgrid $z$, $A$ is the set of pairs of constants representing incompatible choices, i.e. repeated numbers in the same column, row or subgrid. Finally, $I$ is the set of constants representing fixed numbers in the sudoku grid at the beginning of the game. 

The interpretation constants are $Q = \{ N_{nij}, G \}$, the interpretation sentence is:
\[ 
{\bf{\varphi}} = {\bf{\alpha}}  \wedge {\bf{\beta}} \wedge {\bf{\gamma}}  \wedge {\bf{\epsilon}}  \wedge {\bf{\pi}}, 
\]
and the scope sentence is:
\[ 
\Xi = {\bf{\alpha}}.
\]
We can determine the position of the numbers in the grid by using:
\[ 
N_{nij} \leq G  \, \Leftrightarrow \,  \text{the sudoku has a number n at position (i, j)}
\]
The separator set is:
\[ 
\Gamma =   \{ (N_{nij} \leq G) :  1 \leq n \leq 9, 1 \leq i \leq 9, 1 \leq j \leq 9  \} . 
\]
\bigskip

\underline{Embedding}:  We have interpretation constants $Q = \{ N_{nij}, G \}$. Consider the embedding constants $C =\{ N_{nij}, G, W_{nij}, R_{ni}, C_{nj}, Z_{nz}, g_{nij} \}$ where $W_{nij}$ plays the role of a negation of the $N_{nij}$ in a similar way that $E_{ij}$ represents the negation of $Q_{ij}$ in the Queens Completion problem. The set of constants $R_{ni}$ represent the presence of a number $n$ in a row $i$,  $C_{nj}$ represent the presence of a number $n$ in the column $j$ and $Z_{nz}$ represent the presence of a number $n$ in the subgrid $z$; these three sets have $9^2$ constants each. Finally the set of $9^3$ context constants $g_{nij}$ are used to make the embedding explicit. In total we need $\vert C \vert = 3 \times 9^3 + 3 \times 9^2 + 1 = 2431$ constants. 

We start by encoding the meaning of $W_{nij}$. We can use the set $WN$ of $2\times 9^3$ atomic sentences:
\[ 
\odot_{m:m\not=n} W_{mij}  \leq N_{nij} \odot G
\]\[ 
N_{nij}  \leq\odot_{m:m\not=n} W_{mij}  \odot G.
\]
Now it is possible to encode ${\bf{\beta}}$ with the set $NWG$ of $9^3$ negative atomic sentences:
\[ 
N_{nij} \odot W_{nij}  \not\leq G,
\]
instead of directly using the $8 \times 9^3$ negative atomic sentences of  ${\bf{\beta}}$.

The encoding of ${\bf{\gamma}}$ can be carried out with the set $RCZ$ of $9^3$ atomic sentences:
\[ 
R_{ni} \odot C_{nj}  \odot Z_{n z(i,j)}   \leq  N_{nij},
\]
where $z(i,j)$ is a function that maps a position $(i,j)$ of the grid to its subgrid index with $1 \leq z \leq 9$. We also use the set $NRCZ$ of $3 \times 9^2$ negated atomic sentences:
\[ 
R_{mx}  \not<   \odot_{nij: (n, i) \not= (m, x) }  N_{nij},
\]\[ 
C_{my}  \not<   \odot_{nij: (n, j) \not= (m, y) }  N_{nij},
\]\[ 
Z_{mk}  \not<   \odot_{nij: (n, z(i,j)) \not= (m, k) }  N_{nij},
\]
and the atomic sentence:
\[ 
\delta =  "\odot_{nx}  R_{nx}  \odot_{ny}  C_{ny}  \odot_{nz}  Z_{nz}   \leq  G"
\]
that sets the goal of the game.

${\bf{\epsilon}}$ is encoded as $\vert I \vert$ atomic sentences, the block $IG$, of the form $(N_{nij}  \leq  G)$, one for each of the numbers provided as hints and fixed in the sudoku grid. 

To encode ${\bf{\pi}}$ we use the block $WNG$ of $4 \times 9^3$ atomic sentences: 
\[ 
\odot_{rs} \{ W_{nrs} : (N_{nrs},N_{nij})  \in A  \} \leq N_{nij} \odot G 
\]\[ 
N_{nij} \leq  \odot_{s} \{ W_{nis}: s \not= j  \}  \odot G
\]\[ 
N_{nij} \leq  \odot_{r} \{ W_{nrj}: r \not= i  \}  \odot G
\]\[ 
N_{nij} \leq  \odot_{rs} \{ W_{nrs}: z(rs) = z(ik) \wedge (r,s) \not= (i,j) \}  \odot G,
\]
and also the set  $WG$ of $9^2$ negative atomic sentences:
\[ 
\odot_{n} W_{nij} \not\leq G.
\]

Finally we use the set $NGG$ of $9^3$ atomic sentences:
\[ 
N_{nij} \leq G \odot g_{nij} 
\]
to make the embedding explicit.

The embedding set is:
\[ 
R^{+} = WN \cup RZC \cup IG \cup WNG  \cup NGG  \cup \{ \delta \}
\]\[ 
R^{-} = NWG \cup NRCZ \cup WG 
\]

The embedding is tight and explicit but not concise.

The calculation of the complete atomization using full crossing is computationally heavy. To get some insight we can reduce the dimension of the sudoku grid to $4\times 4$ with 4 subgrids. For a sudoku with no numbers set at the beginning of the game we can assign numbers $1, 2, 3 ,4$ to one of the subgrids and still find every solution up to a permutation of the numbers. We obtain the following atomization:    
\begin{table}[!htbp]
\centering
\begin{tabular}{cc}
\multicolumn{2}{c}{$F_C(R^{+})$} \\
 \hline
Size                     & \#        \\
\hline
1                        & 65                      \\
2                        & 128                     \\
6                        & 48                      \\
112                      & 12 (\#Solutions)        \\
143                      & 32                      \\
144                      & 16                      \\
149                      & 48                      \\
150                      & 48                      \\
151                      & 16                      
\end{tabular}
\end{table}

Albeit computationally expensive full crossing could be used to calculate $F_C(R^{+})$ and, hence, find every solution for Sudoku problem. Fortunately, if we use an sparse version of full crossing, the sparse crossing algorithm of \cite{AMLPresentation}, this embedding works very well in practice and can be used to resolve sudokus. Even more, if no numbers are set in the grid sparse crossing is capable of generating complete sudoku grids with this embedding.

The freest solution models $F_S$ have every atom of $F_C(R^{+})$ with size smaller than 112, one atom of size 112 corresponding to the solution $S$ and no more atoms, i.e. 242 atoms in total. $F = \cup_s F_s$ has 253 atoms corresponding with every atom of $F_C(R^{+})$ smaller or equal to 112, which means the embedding is not strongly complete. Each atom of size 112 can be mapped to one of the 12 solutions of the game, up to permutations, as we have fixed the numbers in one of the subgrids.  The embedding can be made strongly complete  by adding a set of $80$ positive atomic sentences with the form:
\[
W_{ab} \leq N_{cd} \odot W_{ef} \odot G
\] \[
W_{ij} \leq W_{kl} \odot W_{mn} \odot G
\]
which, once added to the embedding, have the sole effect of removing every atom of size larger than 112 from the atomization. This set of sentences may not be unique; other sets of sentences may also work to make the embedding strictly complete. The $80$ atomic sentences encode rules of the game that can be used to resolve the sudoku and can be derived from the atoms of size larger than 112. The sparse crossing algorithm \cite{AMLPresentation} uses these larger atoms as well as redundant atoms to (via pinning terms, see \cite{AMLPresentation}) discover the solutions of the game.

\section{Hamiltonian Path} \label{sec:hamiltonian_path}

We consider here the problem of undirected Hamiltonian paths and build embeddings for this problem that can be easily extended to Hamiltonian cycles or directed Hamiltonian paths. Given a graph of $N$ vertices $\{v_1,  v_2,..., v_N\}$ and $E$ edges $\{e_1,  e_2,..., e_E\}$ we want to find a path $P$ that passes through every vertex without repeating any. 

We start by defining constants $e_1,  e_2,..., e_E$ for the edges and a constant $P$ for the path, so a solution that satisfies $e_i \leq P $ is interpreted as a path that passes through edge $e_i$. 

Consider the interpretation sentence:
\[ 
\varphi = ``\exists \,  v^1, v^2,...,v^N \,  \exists \,  e^1, e^2,...,e^{N - 1} \, \alpha \wedge \beta"
\]
with
\[ 
\alpha =  ``\forall k(  (v^k,v^{k+1} \in V(e^k)) \wedge \forall j (j \not=k \Rightarrow  v^k \not= v^j))"
\]\[ 
\beta =  ``\forall j ((e_j \leq P) \Leftrightarrow \exists k (e_j = e^k))"
\]
where $v^i$ and $e^k$ are variables bounded to vertexes and edges respectively and $V(e)$ is the set with the two vertexes of edge $e$. Super-indexes enumerate positions along the path while sub-indexes identify the different edges and vertexes. 

The interpretation constants are $Q = \{ E_j, P \}$ and we can chose the scope sentence to be:
\[
 \Xi = ``\wedge_i^N \vee_j  \{ e_j \leq P : v_i \in V(e_j) \}"
\]
that satisfies $\varphi \Rightarrow \Xi$. The separator set for this scope sentence is:
\[ 
\Gamma =   \{ (e_j \leq P) :  1 \leq j \leq E \}. 
\]

\underline{First embedding (almost)}: We are going to build first an embedding for a different but closely related problem, the problem of finding a path, perhaps disconnected, that joins the N vertexes of the graph. 

For the disconnected path problem we can use the same interpretation constants, scope sentence and separator set than we just defined for the Hamiltonian path but, with the help of:
\[ 
\gamma =  ``\forall i, j,k ((i \not= j \not= k \not= i) \wedge  (e_i, e_j, e_k \leq P)) \Rightarrow \not\exists v (v \in V(e_i)\cap V(e_j)  \cap V(e_k))"
\]
we define the new interpretation sentence:
\[ 
\varphi = \gamma \wedge \Xi.
\]

To build an embedding for this problem we define constants $v_1,  v_2,..., v_N$ for the vertexes. With these constants we can describe the graph with a block VVE of $E$ positive atomic sentences:
\[ 
v_{L(k)} \odot v_{R(k)} \leq e_k \odot U 
\]
where $v_{L(k)}$ and $v_{R(k)}$ are the two vertexes of edge $e_k$ and $U$ is an auxiliary constant used to make the embedding concise. We complement the description of the graph with the block VEVU of $V$ negative atomic sentences:
\[ 
v_i \not\leq \odot_k \{ e_k :  v_i \not\in V(e_k) \} \odot_j \{ v_j :  j \not= i \}   \odot U. 
\]
Let $E(i)$ be the set of edges that have vertex $v_i$ in one of their ends. For each vertex $v_i$, consider the $m_i = \vert E(i) \vert (\vert E(i) \vert - 1) (\vert E(i) \vert - 2) / 3!$ sets of three edges connected to $v_i$. To describe a path we can define a block B3E of $\sum_i m_i$ positive sentences:\[ 
x \leq e_a \odot e_b \odot e_c  
\] 
where $e_a, e_b, e_c$ are three different edges connected to the same vertex and $x$ is an extra constant. To state the problem we use: \[ 
\delta =  ``\odot_i^V  v_i \leq P \odot U"
\]
that specifies the path passes through all the vertexes and \[ 
\epsilon = ``x \not\leq P \odot U"
\]
to make sure the path does not cross a vertex more than once. 

To make the embedding explicit (and also tight) we use the block G of $E$ positive atomic sentences:\[ 
e_k \leq P \odot g_k
\]
where $g_1,  g_2,..., g_E$ are $E$ context constants.

We have used the interpretation constants $Q = \{ E_j, P \}$ and the embedding constants $C = \{ E_j, P, v_i, g_j, x, U\}$, with a total of $2E + N + 3$ embedding constants. The embedding set is:
\[ 
R^{+} = VVE \cup B3E \cup G \cup \{  \delta \}
\]\[ 
R^{-} = VEVU \cup \{ \epsilon \}.
\]
In total, $1 + 2E + \sum_i m_i$ positive atomic sentences and $N + 1$ negative atomic sentences.

\bigskip
\bigskip

We have chosen to encode the fact that a path does not cross itself by using the set B3E of positive atomic sentences plus the negative atomic sentence $\epsilon$. Instead, we could get rid of constant $x$ and directly use $\sum_i m_i$ negative sentences: \[ 
e_a \odot e_b \odot e_c \not\leq P \odot U.
\] 
However, that is a worse option for this embedding. If we use the block B3E of positive atomic sentences we can find solution models with at most $\vert R^{-} \vert = V + 1$ atoms, while an embedding using negative sentences could require as many as $1 + V + \sum_i m_i$ atoms. It is clear that if we use negative sentences the embedding has the geometry of the graph encoded in the atoms but not the path. If we use positive sentences the atoms encode both, the requirements of a well-defined path and the geometry of the graph. 

Since the embedding is tight we will be able to find a freest solution model for the suitable path (or paths) atomized with non redundant atoms of $F_C(R^{+})$. In addition, since the scope sentence $\Xi$ is a conjunction of disjunctions, theorem \ref{positiveScopeTheroem} tells us that for each solution path there is an irreducible model for the solution atomized with non-redundant atoms of $F_C(R^{+})$.

\bigskip

\underline{Second embedding (almost)}: We can extend our first embedding to get a bit closer to the Hamiltonian path. The ingredient missing in the previous embedding is the continuity of the path. Since we have a graph with a fixed number of vertexes $N$, we can use first order logic to express continuity without problems as we did in the sentence $\alpha$ (such a thing is not possible if $N$ not specified).

To express continuity we are going to use $N$ additional constants $\{w_1, w_2,..., w_N\}$ for the sequence of $N$ vertexes along the path and $N - 1$ additional constants $\{p_1, p_2,..., p_{N - 1}\}$ for the edges along the path.   

Define the block WWp of $N - 1$ positive atomic sentences:\[
w_i \odot w_{i + 1} \leq p_i \odot U,
\]
and the block WpV of $N$ negative atomic sentences: \[ 
w_i \not\leq \odot_k \{ p_k :  i \not= k \, \wedge \, i \not= k - 1 \} \odot_j \{ v_j :  j \not= i \} \odot U. 
\]
We should also add: \[ 
P = \odot^{N-1}_i p_i,
\]
that form a block Pp with two positive atomic sentences, and:\[ 
\odot^{N}_i v_i = \odot^{N}_i w_i
\]
that also form a block VW of two positive atomic sentences.

The interpretation constants are the same as before and the embedding constants are now $C = \{ E_j, P, p_i, w_i, v_i, g_j, x, U\}$, with a total of $2E + 3N + 2$ embedding constants. The embedding set is:
\[ 
R^{+} = VVE \cup B3E \cup G \cup \{  \delta \} \cup WWp \cup Pp \cup VW
\]\[ 
R^{-} = VEVU \cup \{ \epsilon \} \cup WpV.
\]
In total, $4 + N + 2E + \sum_i m_i$ positive atomic sentences and $2N + 1$ negative atomic sentences.

This embedding is closer to an embedding for a Hamiltonian path but it is still an embedding for a potentially disconnected path that passes through the $N$ vertexes. 

\bigskip

\underline{Third embedding}: what fails in the second embedding is that there is no guarantee that constants $w_i$ are mapped to vertexes neither constants $p_j$ necessarily map to edges. To build an actual embedding for a Hamiltonian path we have to ensure such mapping which, as far as we know, can only be done by modifying the interpretation.

Instead of using $e_i \leq P$ to indicate that edge $e_i$ is in the path, we can use $e_k = p_l$ to assign edge $e_k$ not only to the path but also to a specific position $p_l$ along the path. In addition, if we want to make sure the sentences WWp, Pp, VW and WpV do indeed enforce a continuous path, also vertexes should each be assigned to a specific position in the path: $v_i = w_j$

The new interpretation requires the interpretation constants: \[
Q = \{ v_i, w_j, e_k, p_l \}
\]
The scope sentence can be chosen to be: \[
\Xi = ``(\wedge^{N}_i \vee^{N}_j (v_i = w_j)) \wedge (\wedge^{N-1}_l \vee^{E}_k (e_k = p_l))",
\] 
that lead to a separator set:
\[ 
\Gamma =  \begin{cases}
      (v_j \leq w_i), (w_j \leq v_i) :  1 \leq i,j \leq N \\
      (e_k \leq p_l), (p_l \leq e_k) :  1 \leq l \leq N - 1  \,\,\,\,\, 1 \leq k \leq E \\
    \end{cases}   
\]
The interpretation sentence has to be modified accordingly by replacing $\beta$ with:
\[
\beta' =  ``\forall l \exists k ((e_l = p_k) \Leftrightarrow  (e_l = e^k)) \wedge \forall j \exists i ((v_j = w_i) \Leftrightarrow  (v_j = v^i))"
\]
in the interpretation sentence $\varphi = ``\exists \,  v^1, v^2,...,v^N \,  \exists \,  e^1, e^2,...,e^{N - 1} \, \alpha \wedge \beta'\,\, "$.

With this change of interpretation it is possible to use the same atomic sentences of the second embedding. However, the block G introduced to make the first and second embeddings explicit should be replaced by a new set of sentences according to the change in the separator set. The following $2N^2 + 2E(N - 1)$ atomic sentences: 
\[ 
e_l \leq p_k \odot u_{lk0}
\]\[ 
p_k \leq e_l \odot u_{lk1}
\]\[ 
v_i \leq w_j \odot z_{ij0}
\]\[ 
w_j \leq v_i \odot z_{ij1}
\]
where the $u's$ and $z's$ are two families of context constants, make the embedding explicit. It is possible (and convenient) to use theorem  \ref{contextConstantReuse} to reduce the number of context constants to $2E(N - 1)$. Consider the block H with $12\times E(N - 1)$ atomic sentences:\[ 
h_{lk0} \odot e_l = p_k \odot h_{lk0}
\]\[ 
h_{lk0} \odot v_{L(e_l)} = w_k \odot h_{lk0}
\]\[ 
h_{lk0} \odot v_{R(e_l)} = w_{k + 1} \odot h_{lk0}
\]\[ 
h_{lk1} \odot e_l = p_k \odot h_{lk1}
\]\[ 
h_{lk1} \odot v_{R(e_l)} = w_k \odot h_{lk1}
\]\[ 
h_{lk1} \odot v_{L(e_l)} = w_{k + 1} \odot h_{lk1}.
\]
Each context constant is reused here 6 times and each edge $e_l$ is associated to a position $p_k$ twice, corresponding with the two ways a path can traverse through an edge.  
Finally the embedding for the Hamiltonian path uses embedding constants:\[
C = \{ E_j, P, p_i, w_i, v_i, h_j, x, U\},
\] with a total of $E + N^2 + 2N + 2$ embedding constants and the embedding set is:
\[
R^{+} = VVE \cup B3E \cup H \cup \{  \delta \} \cup WWp \cup Pp \cup VW
\]\[ 
R^{-} = VEVU \cup \{ \epsilon \} \cup WpV.
\]
The embedding has $4 + N + E + \sum_i m_i + 12\times E(N - 1)$ positive atomic sentences and $2N + 1$ negative atomic sentences.

\section{Prerequisites} \label{prerequisites}

\begin{theorem}  \label{restrictionLemma} 
Let $M$ be an atomized semilattice model over a set of constants $C$ and let $Q$ be a subset of $C$. Let the restriction of $M$ to $Q$, written $M^{|Q}$ be the subalgebra of $M$ spawned by the terms and constants over $Q$. Then $M^{|Q}$ is the model spawned by the restriction to $Q$ of the atoms of $M$. The restriction of the atoms of a model to a subset of its constants produces the same model for every atomization. 
\end{theorem}
\begin{proof}
See theorem 22 of Finite Atomized Semilattices for a proof.
\end{proof}
\bigskip

\begin{theorem}  \label{restrictionAndRedundacyLemma} Let $M$ be an atomized semilattice model over $C$ and let $Q \subset C$. For each non-redundant atom $\alpha$ of $M^{|Q}$ there is at least one non-redundant atom of $M$ that restricted to $Q$ is equal to $\alpha$.
\end{theorem}
\begin{proof}
See theorem 23 of Finite Atomized Semilattices for a proof.
\end{proof}
\bigskip

\begin{theorem}  \label{sumOfModelsWhenNoIntersection} 
Let $M$ and $N$ be two semilattice models over the constants $C_M$ and $C_N$ respectively, such $C_M \cap C_N = \emptyset$ and $M = [A]$ and $N = [B]$ where $A$ and $B$ are sets of atoms. The model $M \oplus N$ is atomized by $A \cup B$, i.e. $M \oplus N  \approx  M + N$.
\end{theorem}
\begin{proof}
See theorem 30 of Finite Atomized Semilattices for a proof.
\end{proof}
\bigskip

We are also going to use the following result:

\begin{theorem} \label{segregationTheorem}
Let $R^{+}$ be a set of positive duples and $R^{-}$ a set of negative duples. \\
i) If $\,R^{+} \,\cup\, R^{-} \Rightarrow r^{+}$ for some duple $r$ then, either $R^{+} \cup\, R^{-} \cup\, \{r^{+}\}$ has no model or $R^{+}$ alone implies $r^{+}$, i.e. $\,R^{+} \Rightarrow r^{+}$. \\
ii) If $\,R^{+} \,\cup\, R^{-} \Rightarrow s^{-}$ for some duple $s$ then there is at least one duple $t \in R^{-}$ such that $\,R^{+} \cup\, \{t^{-}\} \Rightarrow s^{-}$.
\end{theorem}
\begin{proof}
Assume $R^{+} \cup R^{-} \cup\, \{r^{+}\}$ is consistent. There is a model $M \models R^{+} \cup\, R^{-} \cup\, \{r^{+}\}$. Now assume $\,R^{+} \not\Rightarrow r^{+}$, then there is a model $N \models R^{+} \cup\, \{r^{-}\}$. Without loss of generality we can assume $M$ and $N$ are atomized. Let the model $M + N$ be the model spawn by the union of the atoms of $M$ and the atoms of $N$. Model $M + N$ satisfy all the negative duples of $M$ and $N$, i.e. $M + N \models Th_0^{-}(M) \cup Th_0^{-}(N)$ and $M + N \models Th_0^{+}(M) \cap Th_0^{+}(N)$. It follows $M + N \models R^{+} \cup\, R^{-} \cup\, \{r^{-}\}$ contradicting $\,R^{+} \,\cup\, R^{-} \Rightarrow r^{+}$. This proves (i).\\
To prove (ii) assume that for every duple $t \in R^{-}$, $\,R^{+} \cup\, \{t^{-}\} \not\Rightarrow s^{-}$. This means that there are models $N_t$ satisfying $\forall t(t \in R^{-})\exists N_t (N_t \models \,R^{+} \cup\, \{t^{-}\} \cup\, \{s^{+}\})$. Consider the model $Q$ spawn by the atoms of all the models $N_t$. We have $Q \models \,R^{+} \cup\, \{R^{-}\} \cup\, \{s^{+}\}$ contradicting $R^{+} \,\cup\, R^{-} \Rightarrow s^{-}$.
\end{proof}
\bigskip

\begin{theorem} \label{inferenceTheorem}
Let duples $r=(r_L, r_R)$ and $s=(s_L,s_R)$ and a model $M$ such that $M \models s^{-}$,\\
i) $\,Th_0^{+}(M) \cup \{r^{+}\} \Rightarrow s^{+}$ if and only if $M \models ((r_R < s_R) \wedge (s_L < s_R \odot r_L))$, and \\
ii) $\,Th_0^{+}(M) \cup \{s^{-}\} \Rightarrow r^{-}$ if and only if $M \models ((r_R < s_R) \wedge (s_L < s_R \odot r_L))$.  
\end{theorem}
\begin{proof}
$\,Th_0^{+}(M) \cup \{r^{+}\} \Rightarrow s^{+}\,$ means that for any model $N$ such that $N \models Th_0^{+}(M)$ we should have $N \models (\neg r^{+} \vee  s^{+})$, and using $\neg r^{+} = r^{-}$ and $s^{+} =  \neg s^{-}$ we get $N \models (r^{-} \vee  \neg s^{-})$ which proves that statements $Th_0^{+}(M) \cup \{r^{+}\} \Rightarrow s^{+}\,$ and $\,Th_0^{+}(M) \cup \{s^{-}\} \Rightarrow r^{-}\,$ are equivalent, and then (i) and (ii) are also equivalent. 
It suffices with proving (i). Right to left, condition $M \models ((r_R < s_R) \wedge (s_L < s_R \odot r_L))$ implies that $Th_0^{+}(M) \cup \{r^{+}\} \Rightarrow ((r_L < r_R < s_R) \wedge  (s_L < s_R \odot r_L)) \Rightarrow (s_L < s_R)$, and the left side follows. To prove (i) left to right we can use theorem \ref{fullCrossingIsFreestTheorem}. Note first that if $M \models r^{+}$ then the left side of (i) implies $M \models s^{+}$ which we have assumed false. Therefore, we should have $M \models r^{-}$. A duple $s^{+}$ implied by the set $Th_0^{+}(M) \cup \{r^{+}\}$ should be modeled by $F(Th_0^{+}(M) \cup \{r^{+}\})$ and therefore either $M \models s^{+}$, against our assumptions, or $s^{+}$ is obtained from the full crossing of $r$. In the proof of \ref{fullCrossingIsFreestTheorem} it is shown that if $s^{+}$ occurs as a consequence of crossing $r$ then $M \models ((r_R < s_R) \wedge (s_L < s_R \odot r_L))$.
\end{proof}
\bigskip

\section{Grounded models} \label{groundedModels}

Let $K$ be a subset of the constants $K \subset C$. The atom ``$\phi$ grounded to $K$'', written $\phi^{{}^{\vee}K}$ is an atom that exists only when $U^{c}(\phi) \subseteq K$ and then $\phi^{{}^{\vee}K} = \phi$. 

The model $M$ grounded to $K$, written $M^{{}^{\vee}K}$, is the model spawned by the grounded atoms of $M$ and $\ominus_K$ or, in other words, the model spawned by the atoms of $M$ with an upper constant segment that is a subset of $K$ and $\ominus_K$.  It turns out that $M^{{}^{\vee}K}$ is a well defined model, i.e. it does not depend upon the atomization chosen for $M$. 

The role of  $\ominus_K$ is to ensure the sixth lemma of finite atomized semilattices is satisfied. In practice, we need to add $\ominus_K$ only if it is not redundant with the grounded-to-$K$ atoms of $M$.

We say that model $M$ over the constants $K$ is a tight subset model of another model $N$, written $M \sqsubset N$ if each non-redundant atom of $M$ is either a non-redundant atom of $N$ or is equal to $\ominus_K$.

\bigskip
\bigskip

\begin{theorem}  \label{groundingProperties} 
Let $K \subset C$ and $M$ a model over the constants $C$.  \\ 
i) If $\phi$ is redundant with $M$ then $\phi^{{}^{\vee}K}$ either does not exist or is redundant with $M^{{}^{\vee}K}$.  \\ 
ii) Each non-redundant atom of $M^{{}^{\vee}K}$ is a non-redundant atom of $M$ or is equal to $\ominus_K$, i.e. $M^{{}^{\vee}K}$ is a tight subset model of $M$, written $M^{{}^{\vee}K} \sqsubset M$.  \\ 
iii) $M^{{}^{\vee}K}$ is a well-defined model. \\ 
iv) $M^{{}^{\vee}K}$ is as free or less free than $M$
\end{theorem}
\begin{proof}
When we ground $M$ to a subset of its constants $K$ we do a selection of atoms of $M$. For each atom $\phi \in M$ either $\phi^{{}^{\vee}K}$ does not exist or if it does $\phi^{{}^{\vee}K} = \phi$. \\
(i) Assume $\phi$ is redundant. Then $\phi = \bigtriangledown_i \varphi_i$ for some (different) atoms $\{\varphi_1,...,\varphi_n\}$ in $M$. If $\phi^{{}^{\vee}K}$ exists (i.e. $U^{c}(\phi) \subseteq K$) then since $U^{c}(\phi) = \cup_i U^{c}(\varphi_i)$ we have that each $\varphi_i^{{}^{\vee}K}$ satisfies $U^{c}(\varphi_i) \subseteq K$ and then $\varphi_i^{{}^{\vee}K}$ exists and is an atom in $M^{{}^{\vee}K}$, and then $\phi^{{}^{\vee}K} = \bigtriangledown_i \varphi_i^{{}^{\vee}K}$. Since $\varphi_i^{{}^{\vee}K} = \varphi_i  \not = \phi = \phi^{{}^{\vee}K}$, then $\phi^{{}^{\vee}K}$ is a union of atoms different than $\phi^{{}^{\vee}K}$ and, hence, redundant with $M^{{}^{\vee}K}$.   \\
(ii) Suppose now that $\phi^{{}^{\vee}K} \not= \ominus_K$ is a non-redundant atom in $M^{{}^{\vee}K}$. This implies that $\phi^{{}^{\vee}K}$ exists, i.e. $U^{c}(\phi) \subseteq K$, and $\phi^{{}^{\vee}K} = \phi$. From (i) we know that if $\phi$ were redundant with $M$ then $\phi^{{}^{\vee}K}$ would be redundant with $M^{{}^{\vee}K}$ and it is not. Hence, $\phi$ is non redundant in $M$. Every non-redundant atom in $M^{{}^{\vee}K}$ is non-redundant in $M$ or equal to $\ominus_K$, and this is the definition of a proper subset model.  \\
(iii) We have shown that if $\phi^{{}^{\vee}K}$ exist and $\phi$ is redundant with $M$ then $\phi^{{}^{\vee}K}$ is redundant with $M^{{}^{\vee}K}$, therefore $M^{{}^{\vee}K}$ is defined by the non-redundant atoms of $M$ and is independent of the atomization of $M$. \\ 
(iv) Since  $\ominus_K$ discriminates no duple over $K$, hence, any duple over $K$ discriminated in $M^{{}^{\vee}K}$ is discriminated by some grounded atom of $M$ that is also an atom of $M$. It follows that any duple over $K$ that is negative in $M^{{}^{\vee}K}$ is also negative for $M$.
\end{proof}

\bigskip
\bigskip

 When is cleaner, we sometimes use $[|Q]M$ as an alternative symbol for the restriction $M^{|Q}$ and also $[{\vee}K]M$ instead of the grounding $M^{{\vee}K}$; for example, in the next theorem $[{\vee}K]M$ appears in an smaller font at the discriminant:

\bigskip
\begin{theorem}  \label{groundingAndCrossingCommute}  Let $r$ a duple over $K \subset C$ and $M$ a model over $C$.  \\
i) ${\bf{dis}}_{[{}^{\vee}K] M}(r)  = [{}^{\vee}K] {\bf{dis}}_{M}(r)$. \\
ii) Grounding and crossing commute: $\square_{r} [{}^{\vee}K] M = [{}^{\vee}K] \square_{r} M$. 
\end{theorem}
\begin{proof}
First, observe that, since $[{}^{\vee}K] M$ is an algebra over $K$, for $\square_{r} [{}^{\vee}K] M$ or ${\bf{dis}}_{[{}^{\vee}K] M}(r)$ to be defined $r$ should be a duple over $K$.  \\
(i) Let $r = (r_L, r_R)$.  The discriminant of $r$ in $[{}^{\vee}K] M$ is:  \[
{\bf{dis}}_{[{}^{\vee}K] M}(r)  = \{ \phi : (\phi \in [{}^{\vee}K] M) \wedge (\phi < r_L) \wedge (\phi  \not< r_R) \} = 
\]\[
=  \{ \phi : (\phi \in M) \wedge (U^c(\phi) \subseteq K) \wedge (\phi < r_L) \wedge (\phi  \not< r_R)  \} = [{}^{\vee}K] {\bf{dis}}_{M}(r). 
\]
(ii) Using the definition of full crossing: \[
\square_{r} [{}^{\vee}K] M =  ([{}^{\vee}K] M - {\bf{dis}}_{[{}^{\vee}K] M}(r)) +  {\bf{dis}}_{[{}^{\vee}K] M}(r) \bigtriangledown {\bf{L}}^{a}_{[{}^{\vee}K] M}(r_R) = 
\]\[
=  ([{}^{\vee}K] M - {\bf{dis}}_{[{}^{\vee}K] M}(r)) + ([{}^{\vee}K] {\bf{dis}}_{M}(r)) \bigtriangledown  ([{}^{\vee}K] {\bf{L}}^{a}_{M}(r_R)).
\]
where we have used proposition (i) and taken into account that:  \[
{\bf{L}}^{a}_{[{}^{\vee}K] M}(r_R)  =\{ \phi : (\phi \in [{}^{\vee}K] M) \wedge (\phi  < r_R) \} = 
\]\[
= \{ \phi : (\phi \in M) \wedge (U^c(\phi) \subseteq K)  \wedge (\phi  < r_R) \} =  [{}^{\vee}K] {\bf{L}}^{a}_{M}(r_R).
\]
Now, let's look into the set:  \[
[{}^{\vee}K] M - {\bf{dis}}_{[{}^{\vee}K] M}(r)) = \{ \phi : (\phi \in [{}^{\vee}K] M) \wedge ((\phi \not< r_L) \vee (\phi  < r_R)) \} =
\] \[
= \{ \phi :(\phi \in M) \wedge (U^c(\phi) \subseteq K) \wedge  ((\phi \not< r_L) \vee (\phi  < r_R))  \} = [{}^{\vee}K] (M -  {\bf{dis}}_{M}(r)).
\]
Finally, we argue that $([{}^{\vee}K] {\bf{dis}}_{M}(r)) \bigtriangledown  ([{}^{\vee}K] {\bf{L}}^{a}_{M}(r_R)) = [{}^{\vee}K] \, ({\bf{dis}}_{M}(r) \bigtriangledown {\bf{L}}^{a}_{M}(r_R))$. Suppose $\alpha  \in {\bf{dis}}_{M}(r)$ and $\beta  \in {\bf{L}}^{a}_{M}(r_R)$. If $(\alpha \bigtriangledown \beta)^{{}^{\vee}K}$ exists then $\alpha^{{}^{\vee}K}$ and $\beta^{{}^{\vee}K}$ should also exist, which implies $([{}^{\vee}K] {\bf{dis}}_{M}(r)) \bigtriangledown  ([{}^{\vee}K] {\bf{L}}^{a}_{M}(r_R))  \supseteq [{}^{\vee}K] \, ({\bf{dis}}_{M}(r) \bigtriangledown {\bf{L}}^{a}_{M}(r_R))$ and, conversely, if $\alpha^{{}^{\vee}K}$ and $\beta^{{}^{\vee}K}$ exist then $(\alpha \bigtriangledown \beta)^{{}^{\vee}K}$ also exists, hence,   $([{}^{\vee}K] {\bf{dis}}_{M}(r)) \bigtriangledown  ([{}^{\vee}K] {\bf{L}}^{a}_{M}(r_R))  \subseteq [{}^{\vee}K] \, ({\bf{dis}}_{M}(r) \bigtriangledown {\bf{L}}^{a}_{M}(r_R))$. \\
Putting all together \[
\square_{r} [{}^{\vee}K] M =  ([{}^{\vee}K] M - {\bf{dis}}_{[{}^{\vee}K] M}(r)) + ([{}^{\vee}K] {\bf{dis}}_{M}(r)) \bigtriangledown  ([{}^{\vee}K] {\bf{L}}^{a}_{M}(r_R)) =
\]\[
=   [{}^{\vee}K] (M -  {\bf{dis}}_{M}(r)) +  [{}^{\vee}K] \, ({\bf{dis}}_{M}(r) \bigtriangledown {\bf{L}}^{a}_{M}(r_R)) =  [{}^{\vee}K] \square_{r} M,
\]
as we wanted to show.
\end{proof}

\bigskip

If we compare the proofs of theorem \ref{groundingAndCrossingCommute} and of theorem 25 of the paper on Finite Atomized Semilattices, we can see that grounding and restriction are similar. However, while the hypotheses ``$r$ is a duple over $Q$'' is central for the theorem on restriction the equivalent ``$r$ is a duple over $K$'' is required in the proof of \ref{groundingAndCrossingCommute} only to make $\square_{r}$ definable for an algebra over $K$.

\newpage

\section{Theorems} \label{sec:theorems}

We sometimes write $S \vdash { \bf{\eta}}$ for a sentence ${ \bf{\eta}}$ and a set $S$ of sentences as a shorthand for the conjunction $\wedge \{  \sigma:  \sigma  \in S  \} \vdash { \bf{\eta}}$.

\bigskip
\begin{theorem}  \label{freestSolutionModel} 
Assume an embedding of a problem $P$ with embedding set $R$ over the constants $C$ and scope sentence $\Xi$. Let $S$ be a solution of $P$. The model $F_S \equiv F_C(R^{+} \cup S)$ satisfies $\Xi$ and is a solution model of $P$. Moreover, $F_S$ is the freest model that can be interpreted as the solution $S$.
\end{theorem}
\begin{proof}
A valid embedding assumes that for each solution $S$ there should be at least one solution model $M_S \models  R \wedge {\bf{\varphi}}_S$. Let $\Gamma$ be the separator set obtained from $\Xi$. There is a sentence $\Xi_S$ such that $\Xi_S \Rightarrow \Xi$ and $R \wedge {\bf{\varphi}}_S \Leftrightarrow \Xi_S \wedge R$ so the solution $S$ can be unambiguously mapped to the subset $S$ of sentences of $\Gamma$ satisfied by $\Xi_S$. It follows that $M_S$ satisfies $R^{+}  \cup S$ and $R^{-}  \cup \{  \neg \sigma : \sigma \in \Gamma \cap \overline{S} \}$. Since $F_S$ is the freest model of $R^{+} \cup S$ then $F_S$ is freer than $M_S$ and, hence, $F_S$ should also model $R^{-}  \cup \{  \neg \sigma : \sigma \in \Gamma \cap \overline{S} \}$, i.e. $F_S \models R \wedge \Xi_S$. Using the equivalence above between $R \wedge \Xi_S$ and $R \wedge {\bf{\varphi}}_S$ we get to $F_S \models {\bf{\varphi}}_S$ and $F_S \models R \wedge \Xi$, so $F_S$ is a solution model for $P$ that can be interpreted as $S$. Notice that $F_S$ is freer than any model of $R \wedge \Xi_S$ and then also freer than any model of $R \wedge {\bf{\varphi}}_S$.
\end{proof}

\bigskip
\bigskip

\begin{theorem} \label{conciseEmbeddingTheorem} 
An embedding of a problem $P$ with interpretation constants $Q$ and a separator set $\Gamma$ is concise if and only if the freest solution model restricted to $Q$ is equal to $F_{Q}(S)$ for each properly defined solution $S$ of $P$.
\end{theorem}
\begin{proof}
An embedding is concise when for every solution $S$ and for every duple $r$ over the interpretation constants $Q$ we have $R \wedge {\varphi_S} \vdash r^{+}$ if and only if ${\varphi_S} \vdash r^{+}$ where the solution is properly defined, i.e. described by the sentence ${\bf{\varphi}}_S = \wedge \{\pi: \pi \in S\} \wedge \{\neg \nu: \nu \in \Gamma \cap \overline{S}\}$. \\
We showed in theorem \ref{freestSolutionModel} that the freest solution model $F_S$ is the freest model of $R \wedge {\bf{\varphi_S}}$, therefore, $R \wedge {\bf{\varphi_S}} \vdash r^{+}$ if and only if $F_S \models r^{+}$. According to theorem \ref{segregationTheorem} negative sentences do not have positive consequences so, ${\bf{\varphi_S}} \vdash r^{+}$ if and only if $\wedge \{\pi: \pi \in S\} \vdash r^{+}$ which, in turn, occurs if and only if $F_{C}(S) \models r^{+}$. Putting both things together: an embedding is concise if for every $r$ over the interpretation constants $Q$ we have $F_S \equiv F_{C}(R^{+} \cup S) \models r^{+}$ if and only if $F_{C}(S) \models r^{+}$. Since any model and its restriction to $Q$ agree on the sign of every duple over $Q$, in a concise embedding the restriction to $Q$ of $F_S$ should be equal to the restriction to $Q$ of $F_{C}(S)$. Using that all the duples of $S$ are over $Q$, the restriction $F_{C}(S)^{|Q}$ is equal to  $F_{Q}(S)$ and we finally get that in a concise embedding $F_S^{|Q} = F_{Q}(S)$. \\
Conversely, if $F_S^{|Q} = F_{Q}(S)$, for any duple $r$ over the interpretation constants $Q$ holds that $F_S \models r^{+}$ if and only if $F_{Q}(S) \models  r^{+}$. Now, $F_S \equiv F_{C}(R^{+} \cup S) \models r^{+}$ means that the freest model of  $R^{+} \cup S$ or, equivalently, the freest model of $R \wedge {\bf{\varphi_S}}$ models $r^{+}$ and then every model of $R \wedge {\bf{\varphi_S}}$ models $r^{+}$ which can be written as $R \wedge {\varphi_S} \vdash r^{+}$. Same is true for $F_{Q}(S)$ that satisfies $F_{Q}(S) \models r^{+}$ if and only if ${\varphi_S} \vdash r^{+}$. Therefore, $R \wedge {\varphi_S} \vdash r^{+}$ if and only ${\varphi_S} \vdash r^{+}$ and the embedding is concise.
\end{proof}

\bigskip

\begin{theorem} \label{weakEquivalenceTheorem} 
Two concise algebraic semantic embeddings of a problem $P$ using embedding sets $R_1$ and $R_2$ over $C_1$ and $C_2$ respectively that share the same interpretation $(Q, {\bf{\varphi}}, \Gamma)$ produce for each solution $S$ of $P$ a freest model of $S$ with the same non-redundant atoms restricted to $Q$. 
\end{theorem}
\begin{proof}
From theorem \ref{conciseEmbeddingTheorem} we know that for every solution $S$ two concise embeddings will produce the same freest solutions restricted to interpretation constants $F_{1S}^{|Q} = F_{2S}^{|Q} = F_{Q}(S)$.  Theorem \ref{restrictionLemma} says that the restriction to $Q$ of a model is the model spawned by the restriction to $Q$ of its atoms. Therefore, the non-redundant atoms produced by restricting the atoms of $F_{1S}$ to $Q$ should be equal to those obtained from the restriction to $Q$ of the atoms of $F_{2S}$.
\end{proof}

\bigskip

\begin{theorem} \label{restrictionToQExistsTheorem} 
Let $R$ be an embedding of a problem $P$ with embedding constants $C$ and interpretation constants $Q$.  Let $F = \cup_S F_{S}$ where the union runs along the solutions of $P$. \\
i) A non-redundant atom $\phi$ of $F_C(R^{+})$ is external to $F$ if and only if $\phi$ is external to every freest solution model $F_S$. \\
ii) Let $s$ be a duple and $S$ a solution of $P$ such that $F_C(R^{+}) \models s^{-}$ and $F_S \models s^{+}$. Every atom $\phi \in F_C(R^{+})$ that discriminates $s$ has a restriction $\phi^{|Q}$. \\
iii) Every atom $\phi$ of $F_C(R^{+})$ that is external to at least one $F_S$ has a restriction $\phi^{|Q}$. \\
iv) An atom of $F_C(R^{+})$ with no restriction to $Q$ is an atom of every freest solution model $F_S$.
\end{theorem}
\begin{proof}
(i) Since $F = \cup_S F_{S}$, an atom that is external to $F$ must be external to every $F_{S}$. In the other direction,  suppose $\phi$ is a non-redundant atom of $F_C(R^{+})$ that is external to every $F_S$.  Since $F_S \equiv F_{C}(R^{+} \cup S)$ the atoms of $F_S$ are all atoms of $F_{C}(R^{+})$ so we can write $F_S \subset F_{C}(R^{+})$ and from here $F = \cup_S F_S \subseteq F_{C}(R^{+})$. Since the atoms of $F$ are all atoms in $F_{C}(R^{+})$ then $\phi$ cannot be redundant with $F$ because if it were then it would also be redundant with $F_{C}(R^{+})$ and it is not. If $\phi$ is non-redundant in $F$ then it is non-redundant in some $F_S$ and this contradicts our assumption. Therefore $\phi$ is external to $F$. \\
(ii)  Suppose we have a model $M$ and a couple of duples $r = (r_L, r_R)$ and $s = (s_L, s_R)$ such that $M \models s^{-}$ and $\square_r M \models s^{+}$ (we write this as $M \models (r^{+} \Rightarrow s^{+})$). Theorem \ref{inferenceTheorem} proves that $M \models (r^{+} \Rightarrow s^{+})$ if and only if $M \models (r_R \leq s_R)  \wedge (s_L \leq s_R \odot r_L)$.  \\
Suppose a duple $s$ and a solution $S$ of $P$ such that $F_C(R^{+}) \models s^{-}$ and $F_S \models s^{+}$.  Then $F_{C}(R^{+} \cup S) \models s^{+}$. Let $\{ q_1, q_2, ..., q_n\} = S$. There is some $i \leq n$ and $i > 0$ such that $F_{C}(R^{+} \cup \{ q_1, q_2, ..., q_i\}) \models s^{+}$ and $F_{C}(R^{+} \cup \{ q_1, q_2, ..., q_{i-1}\}) \models s^{-}$. Let $r = q_{i}$ and $M =  F_{C}(R^{+} \cup \{ q_1, q_2, ..., q_{i-1}\})$ and let's use the result above to write: $M \models (r_R \leq s_R)  \wedge (s_L \leq s_R \odot r_L)$.  \\
All the duples in $\{ q_1, q_2, ..., q_n\}$ are duples over $Q$ so $r$ is a duple over $Q$. Any atom $\phi$ of $M$ that discriminates $s$ is an atom $(\phi < s_L) \wedge (\phi \not< s_R)$, that together with $M \models (s_L \leq s_R \odot r_L)$ implies $(\phi < r_L)$.  Since $r$ is a duple over $Q$ then $U^c(\phi)  \cap Q \not= \emptyset$.  \\
(iii) We are going to use the result: \\
\emph{- If $\phi$ is a non-redundant atom in a model $M$ with constants $C$ then either $U^c(\phi) = C$ or there is at least one duple $(c, t)$ that is discriminated in $M$ only by $\phi$ where $c$ is a constant in the set $U^c(\phi)$ and $t$ is equal to any idempotent summation of constants in $C - U^c(\phi)$}. \\
Applying this result to a non-redundant atom $\phi$ in $F_C(R^{+})$ with $U^c(\phi) \not= C$ that is external to at least one $F_S$ follows that there is a duple $s = (c, t)$ that is discriminated in $F_C(R^{+})$ only by $\phi$ and then $F_S \models s^{+}$. In this situation proposition (ii) tells us that $\phi^{|Q}$ exists. In case $U^c(\phi) = C$ then $\phi^{|Q}$ exists and has upper segment equal to $Q$. Therefore, every non-redundant atom in $F_C(R^{+})$ that is external to $F_S$ has a restriction $\phi^{|Q}$. \\ What about redundant atoms? A redundant atom $\phi$ with $F_C(R^{+})$ is a union of non-redundant atoms of $F_C(R^{+})$. Let $X$ be a set of non-redundant atoms of $F_C(R^{+})$ which union is equal to $\phi$. If $\phi$ is external to $F_S$ there should be at least one non-redundant atom $\varphi$ in $X$ external to $F_S$, so the upper constant segment of  $\phi$ contains the upper constant segment of $\varphi$ and, because $\varphi^{|Q}$ exists, then $\phi$ cannot have null intersection with $Q$, so $\phi^{|Q}$ also exists. \\
(iv) It is the negation of prop iii.
\end{proof}

\bigskip
\begin{theorem} \label{noResidualAtomModelEquality} 
Consider an embedding with embedding set $R$ and interpretation constants $Q$:  \\
i) $\alpha \not= \ominus_Q$ is a residual atom if and only if $\alpha$ is a non-redundant atom of $F_C(R^{+})^{|Q}$ and is external to the restriction to $Q$ of the freest solution model, i.e. $F_S^{|Q}$, of every solution $S$ of $P$.  \\
ii) A residual atom is external to $F^{|Q}$ where $F = \cup_S F_S$.  \\
iii) An embedding has no residual atoms if and only if $F^{|Q} = F_C(R^{+})^{|Q}$.  \\
iv) If $\alpha$ is residual then there is at least one non-redundant atom $\phi \in F_{C}(R^{+})$ such that $\alpha = \phi^{|Q}$. \\
v) Any non-redundant atom $\phi$ in $F_{C}(R^{+})$ such that $\phi^{|Q}$ is residual is external to $F = \cup_S F_S$.  \\
vi) Suppose the embedding has no residual atoms and let $\phi$ be a non-redundant atom of $F_{C}(R^{+})$ such that $\phi^{|Q}$ exists, is non-redundant in  $F_C(R^{+})^{|Q}$ and is different than $\ominus_Q$; there is some solution $S$ such that $\phi^{|Q} \in F_S^{|Q}$. 
\end{theorem}
\begin{proof}
(i) An atom $\alpha$ is residual if it is a non-redundant atom of $F_C(R^{+})^{|Q}$ different than $\ominus_Q$ that is external to every $M_S^{|Q}$ where $M$ is any solution model of $P$. The freest solution model $F_S$ of a solution $S$ is (equal or) freer than any solution model $M$ of $S$ and then $F_S^{|Q}$ is freer than every $M_S^{|Q}$. Therefore, any atom of that is external to $F_S^{|Q}$ is external to every $M_S^{|Q}$. Since  $F_S$ is also a model of $S$ then we have the equivalent statement for a definition of redundant atom:  an atom $\alpha$ is residual if and only if $\alpha$ is a non-redundant atom of $F_C(R^{+})^{|Q}$ different than $\ominus_Q$ that is external to the solution model $F_S^{|Q}$ of every solution of $P$.   \\
(ii) The proof is similar to that of theorem \ref{restrictionToQExistsTheorem} part i.  Since $F_S = F_C(R^{+}\cup S)$ then $F_C(R^{+})$ is freer than $F_S$ so the atoms of $F_S$ are all atoms of $F_C(R^{+})$, i.e. $F_S  \subset F_C(R^{+})$. Therefore $F \equiv \cup_S F_S \subset F_C(R^{+})$ and then $F^{|Q} = \cup_S F_S^{|Q} \subset F_C(R^{+})^{|Q}$. A residual atom $\alpha$ must be a non-redundant atom of $F_C(R^{+})^{|Q}$, so it is either a non-redundant atom of $F^{|Q}$ or external to $F^{|Q}$. Now, $F^{|Q} = \cup_S F_S^{|Q}$ implies that every non-redundant atom of $F^{|Q}$ should be a non-redundant atom of at least one $F_S^{|Q}$ and proposition i. says that $\alpha$ is external to every $F_S^{|Q}$, hence, $\alpha$ must be external to $F^{|Q}$.   \\
(iii) We just showed that $F^{|Q} \subset F_C(R^{+})^{|Q}$. Assume $\alpha$ is a non-redundant atom of $F_C(R^{+})^{|Q}$ external to $F^{|Q}$. Notice that $\ominus_Q$ is an atom of every model over $Q$ so, if $\alpha$ is external to $F^{|Q}$ then $\alpha \not= \ominus_Q$. Since $F^{|Q} = \cup_S F_S^{|Q}$ then $\alpha$ must be external to every $F_S^{|Q}$ and therefore, from proposition i. $\alpha$ is residual. It follows that, if the embedding has no residual atoms, every non-redundant atom $\alpha$ of $F_C(R^{+})^{|Q}$ must be in $F^{|Q}$ and then $F^{Q} = F_C(R^{+})^{|Q}$. In the other direction, using part ii, a residual atom is external to $F^{|Q}$ so if  $F^{|Q} = F_C(R^{+})^{|Q}$ it would be also external to $F_C(R^{+})^{|Q}$ which contradicts the definition of residual atom. Therefore, $F^{|Q} = F_C(R^{+})^{|Q}$ implies the embedding has no residual atoms.   \\
(iv) A residual atom $\alpha$ is non-redundant atom in $F_{C}(R^{+})^{|Q}$.  From theorem \ref{restrictionAndRedundacyLemma} follows that there is at least one non-redundant atom $\phi \in F_{C}(R^{+})$ such that $\alpha = \phi^{|Q}$. \\
(v) We showed above that $F = \cup_S F_S \subseteq F_{C}(R^{+})$. Since the atoms of $F$ are all atoms in $F_{C}(R^{+})$ then $\phi$ cannot be redundant with $F$ because, if it were, it would also be redundant with $F_{C}(R^{+})$ and it is not. In addition, if $\phi$ is non-redundant in $F$ then it would be non-redundant in some $F_S$ and then $\alpha \in F_S^{|Q}$ against the assumption that $\phi^{|Q}$ is a residual atom. Since $\phi$ is neither redundant with $F$ nor a non-redundant atom in $F$ then $\phi \not\in F$.   \\
(vi) Suppose such solution $S$ does not exist, i.e. for every solution $\phi^{|Q}  \not\in F_S^{|Q}$.  Since $\phi^{|Q}$ is non-redundant in $F_C(R^{+})^{|Q}$ and is not equal to $\ominus_Q$, prop. i tells us that $\phi^{|Q}$ is a residual atom, against our assumptions.   
\end{proof}
\bigskip

\begin{theorem} \label{residualAtomIffPositiveAtomicTheorem} 
An embedding of a problem $P$ has a residual atom if and only if there is a positive atomic sentence $r^{+}$ over $Q$ such that $R \wedge {\bf{\varphi_S}} \vdash r^{+}$ for every solution $S$ of $P$ and $R \not\vdash r^{+}$. 
\end{theorem}
\begin{proof}
Let $R$ be the embedding set over the embedding constants $C$. Define the set of atoms $F = \cup_S F_S$ where $F_S$ is an atomization of the freest solution model $F_{C}(R^{+} \cup S)$.  \\
$(\Rightarrow)$ Suppose $\alpha$ is a residual atom of the embedding. From theorem \ref{noResidualAtomModelEquality} we know that there is at least one non-redundant atom $\phi \in F_{C}(R^{+})$ such that $\alpha = \phi^{|Q}$. If $\alpha$ is non-redundant in $F_{C}(R^{+})^{|Q}$ with $U^{c}(\alpha) \not= Q$ then there is at least one duple $r$ of $Q$ discriminated only by $\phi^{|Q}$ in $F_{C}(R^{+})^{|Q}$ and discriminated by $\phi$ in $F_{C}(R^{+})$ so  $F_{C}(R^{+}) \models r^{-}$ and then $R \not\vdash r^{+}$. Since $\alpha$ is not an atom in any $F_S^{|Q}$ and every $F_S^{|Q} \subset F_{C}(R^{+})^{|Q}$ (see the proof of theorem \ref{noResidualAtomModelEquality}) then $r$ is not discriminated in any solution model $F_S^{|Q}$. This means that $\forall S (F_S \models r^{+})$ or, in other words, $\forall S (F_{C}(R^{+} \cup S) \models r^{+})$ that can also be written as $\forall S  (R \wedge {\bf{\varphi_S}} \vdash r^{+})$. \\
$(\Leftarrow)$ Conversely if there is a duple $r$ over $Q$ such that $R \not\vdash r^{+}$ and $R \wedge {\bf{\varphi_S}} \vdash r^{+}$ for every solution $S$ of $P$ then $F_{C}(R^{+}) \models r^{-}$ and it follows that there is at least one non-redundant atom $\phi$ in the freest model $F_{C}(R^{+})$ that discriminates $r$. If $\phi^{|Q}$ is a non-redundant atom of $F_{C}(R^{+})^{|Q}$ then $\phi^{|Q}$ cannot be an atom of any $F_S^{|Q}$ (otherwise $R \wedge {\bf{\varphi_S}} \vdash r^{-}$) and then $\phi^{|Q}$ is a residual atom of the embedding. The atom $\phi^{|Q}$ may, however, be redundant in $F_{C}(R^{+})^{|Q}$. If $\phi^{|Q}$ is redundant with $F_{C}(R^{+})^{|Q}$ then there should be some non-redundant atom $\alpha$ in $F_{C}(R^{+})^{|Q}$ that discriminates $r$. From  theorem \ref{noResidualAtomModelEquality} we know that $\alpha$ is the restriction of at least one non-redundant atom $\beta$ in $F_{C}(R^{+})$. Since $F_S^{|Q} \models r^{+}$ then $\alpha = \beta^{|Q}$ cannot be in any $F^{|Q}_S$ and $\alpha$ is a residual atom of the embedding.
\end{proof}

\bigskip
\begin{theorem} \label{residualAtomTheorem} 
(i) An embedding that has no residual atoms is complete.   \\
(ii)  An embedding that is complete and concise has no residual atoms.
\end{theorem}
\begin{proof}
Theorem \ref{residualAtomIffPositiveAtomicTheorem} says that an embedding of a problem $P$ has a residual atom if and only if $\exists r((R \not\vdash r^{+}) \wedge \forall S(R \wedge {\bf{\varphi_S}} \vdash r^{+}))$ where $r$ is a duple over the embedding constants $Q$ and the universal quantifier runs along the solutions $S$ of $P$.  \\
i) If an embedding has no residual atoms then $\forall r ((R \vdash r^{+}) \vee \exists S(R \wedge {\bf{\varphi_S}} \not\vdash r^{+}))$. Since $\exists S(R \wedge {\bf{\varphi_S}} \not\vdash r^{+})$ implies $\exists S({\bf{\varphi_S}} \not\vdash r^{+})$ then $\forall r ((R \vdash r^{+}) \vee \exists S({\bf{\varphi_S}} \not\vdash r^{+}))$ which is the definition of a complete embedding.  \\
ii) If an embedding is concise then $R \wedge {\bf{\varphi_S}} \vdash r^{+}$ implies ${\bf{\varphi_S}} \vdash r^{+}$ when $r$ is a duple over $Q$. If the model is also complete then $\forall S({\bf{\varphi_S}} \vdash r^{+})$ implies $R \vdash r^{+}$, in other words $\forall r (\forall S(R \wedge {\bf{\varphi_S}} \vdash r^{+})  \Rightarrow (R \vdash r^{+}))$ and then theorem \ref{residualAtomIffPositiveAtomicTheorem} tells us that the embedding has no residual atoms. 
\end{proof}

\bigskip

\begin{theorem} \label{equivalenceTheorem} 
Two concise and complete embeddings with embedding constants $C_1$ and $C_2$ respectively sharing the same interpretation $(Q, {\bf{\varphi}}, \Gamma)$ have the same non-redundant atoms restricted to $Q$, i.e $F_{C1}(R_1^{+})^{|Q}= F_{C2}(R_2^{+})^{|Q}$.
\end{theorem}
\begin{proof}
From theorem  \ref{conciseEmbeddingTheorem} we know that in a concise embedding $F_S^{|Q} = F_{Q}(S)$ and then $F^{|Q} =  \cup_S F_{Q}(S)$. Since embeddings 1 and 2 are both concise then $F_1^{|Q} =  \cup_S F_{Q}(S) = F_2^{|Q}$. Theorem \ref{residualAtomTheorem} proves that our two embeddings have no residual atoms and then theorem \ref{noResidualAtomModelEquality} (iii) shows that $F_1^{|Q} = F_{C1}(R^{+})^{|Q}$ and $F_2^{|Q} = F_{C2}(R^{+})^{|Q}$. Putting all together $F_{C1}(R^{+})^{|Q} = F_1^{|Q} =  \cup_S F_{Q}(S) = F_2^{|Q} = F_{C2}(R^{+})^{|Q}$ as we wanted to prove.
\end{proof}

\bigskip

\begin{theorem} \label{sameSolutionAtomsTheorem} 
Consider two embeddings $1$ and $2$ of a problem $P$ with the same interpretation $(Q, {\bf{\varphi}}, \Gamma)$ and let $R_1$ and $R_2$ be the embeddings sets and $C_1$ and $C_2$ the embedding constants. Let $\eta$ be a non-redundant atom of $F_{1S}^{|Q}$ where $S$ is a solution of $P$.  \\ 
i) If both embeddings are concise then $\eta$ is a non-redundant atom of $F_{2S}^{|Q}$. \\ 
ii) If both embeddings are concise and tight there are non-redundant atoms $\phi_1 \in F_{C1}(R_1^{+})$ and $\phi_2 \in F_{C2}(R_2^{+})$ such that $\phi_1^{|Q} = \phi_2^{|Q} = \eta$.
\end{theorem}
\begin{proof}
(i) Since the embeddings are concise in agreement with theorem \ref{conciseEmbeddingTheorem}  so, since $\eta$ is a non-redundant atom of $F_{1S}^{|Q}$ then it is also a non-redundant atom of $F_{2S}^{|Q}$. \\ 
(ii) If the embeddings are tight every non-redundant atom $\eta$ of $F_{1S}^{|Q}$ is equal to the the restriction to $Q$ of at least one non-redundant atom $\phi_1$ of $F_{C1}(R_1^{+})$. From prop. i there is also at least one non-redundant atom $\phi_2$ of $F_{C2}(R_2^{+})$ with restriction equal to $\eta$.
\end{proof}

\bigskip

\begin{theorem} \label{sameSolutionAtomsTheoremPartB} 
Consider two embeddings $1$ and $2$ of a problem $P$ with the same interpretation $(Q, {\bf{\varphi}}, \Gamma)$ and let $R_1$ and $R_2$ be the embeddings sets and $C_1$ and $C_2$ the embedding constants. Let $\phi_1$ be a non-redundant atom of $F_{1S}$ where $S$ is a solution of $P$.  \\ 
i) If both embeddings are concise then $\phi_1^{|Q}$ is either redundant with both $F_{1S}^{|Q}$ and $F_{2S}^{|Q}$ or is a non-redundant atom in $F_{1S}^{|Q}$ and in  $F_{2S}^{|Q}$. \\ 
ii) If both embeddings are concise and tight either $\phi_1^{|Q}$ is redundant with $F_{1S}^{|Q}$ or there is a non-redundant atom $\phi_2 \in F_{C2}(R_2^{+})$ such that $\phi_1^{|Q} = \phi_2^{|Q}$.
\end{theorem}
\begin{proof}
Suppose $\phi_1^{|Q}$ is redundant with $F_{1S}^{|Q}$. Because the embeddings are concise $F_{1S}^{|Q} = F_{2S}^{|Q}$ and then $\phi_1^{|Q}$ is redundant with $F_{2S}^{|Q}$. If, on the other hand, $\phi_1^{|Q}$ is non-redundant in $F_{1S}^{|Q}$ make $\eta$ in theorem \ref{sameSolutionAtomsTheorem} equal to $\phi_1^{|Q}$.
\end{proof}

\bigskip

\begin{theorem} \label{sameSolutionAtomsTheoremPartC} 
Consider two embeddings $1$ and $2$ of a problem $P$ with the same interpretation $(Q, {\bf{\varphi}}, \Gamma)$ and let $R_1$ and $R_2$ be the embeddings sets and $C_1$ and $C_2$ the embedding constants. Let $\phi_1$ be a non-redundant atom of $F_{1}$.  \\ 
i) If both embeddings are concise there is some solution $S$ where $\phi_1^{|Q}$  is either redundant with both $F_{1S}^{|Q}$ and $F_{2S}^{|Q}$ or is a non-redundant atom in $F_{1S}^{|Q}$ and in  $F_{2S}^{|Q}$. \\ 
ii) If both embeddings are concise and tight either $\phi_1^{|Q}$ is redundant with $F_{1S}^{|Q}$ for some solution $S$ or there is a non-redundant atom $\phi_2 \in F_{C2}(R_2^{+})$ such that $\phi_1^{|Q} = \phi_2^{|Q}$.
\end{theorem}
\begin{proof}
If $\phi_1$ is non-redundant in $F_1 = \cup_S F_{S1}$ then it is non-redundant in at least one $F_{1S}$. Apply theorem \ref{sameSolutionAtomsTheoremPartB} to $\phi_1$ non-redundant in $F_{1S}$.
\end{proof}

\bigskip

\begin{theorem} \label{stronglyComplete} 
An embedding is strongly complete if and only if $F \equiv \cup_S F_S = F_C(R^{+})$.
\end{theorem}
\begin{proof}
Since $F_C(R^{+})$ is freer than every freest solution $F_S$ it follows that $F \equiv \cup_S F_S \subset F_C(R^{+})$. In the other direction, let $\phi$ be non-redundant in $F_C(R^{+})$ and external to $F$. Then there is at least one duple $r$ over the embedding constants $C$ that is discriminated only by $\phi$. Since $\phi$ is external to $F $ it must be external to each $F_S$ and, hence, $F_S \models r^{+}$ for every solution $S$ while $F_C(R^{+})  \models r^{-}$ contradicting the assumption that the embedding is strongly complete.
\end{proof}

\bigskip

\begin{theorem} \label{sameSolutionAtomsTheoremPartD} 
Consider two embeddings $1$ and $2$ of a problem $P$ with the same interpretation $(Q, {\bf{\varphi}}, \Gamma)$ and let $R_1$ and $R_2$ be the embeddings sets and $C_1$ and $C_2$ the embedding constants. Let $\phi_1$ be a non-redundant atom of $F_{C1}(R_1^{+})$ and assume embedding 1 is strongly complete.  \\ 
i) If both embeddings are concise there is some solution $S$ where $\phi_1^{|Q}$  is either redundant with both $F_{1S}^{|Q}$ and $F_{2S}^{|Q}$ or is a non-redundant atom in $F_{1S}^{|Q}$ and in  $F_{2S}^{|Q}$. \\ 
ii) If both embeddings are concise and tight either $\phi_1^{|Q}$ is redundant with $F_{1S}^{|Q}$ for some solution $S$ or there is a non-redundant atom $\phi_2 \in F_{C2}(R_2^{+})$ such that $\phi_1^{|Q} = \phi_2^{|Q}$.
\end{theorem}
\begin{proof}
If embedding 1 is strongly complete, theorem  \ref{stronglyComplete} tells us that $F_1 = F_{C1}(R_1^{+})$. Therefore, $\phi_1$ is a non-redundant atom of $F_1$ and we can apply theorem \ref{sameSolutionAtomsTheoremPartC}. 
\end{proof}

\bigskip
\bigskip

\begin{theorem} \label{explicitIsTightTheorem} 
An explicit embedding is tight.
\end{theorem}
\begin{proof}
An embedding is explicit if for each solution $S$ there is a subset $K_S$ of the constants $Q \subseteq K_S \subset C$ such that the freest solution model for $S$ is $F_S = F_{C}(R^{+})^{{}^{\vee}K_S} \oplus F_{C - K_S}( \emptyset)$.  Since the set of constants in $F_{C}(R^{+})^{{}^{\vee}K_S}$ and $F_{C - K_S}( \emptyset)$ are disjoint, theorem \ref{sumOfModelsWhenNoIntersection} tells us that $F_S = F_{C}(R^{+})^{{}^{\vee}K_S} + F_{C - K_S}( \emptyset)$. Theorem \ref{groundingProperties} shows that $F_{C}(R^{+})^{{}^{\vee}K_S} \sqsubset  F_{C}(R^{+})$, i.e. the non-redundant atoms of $F_{C}(R^{+})^{{}^{\vee}K_S}$ are all non-redundant atoms of $F_{C}(R^{+})$. Regarding the rest of the atoms of $F_S$, i.e the atoms in the set $F_{C - K_S}( \emptyset)$, consider that it follows from $F_S \equiv F_{C}(R^{+} \cup S)$ that $F_S \subset F_{C}(R^{+})$. Therefore, $F_{C - K_S}( \emptyset) \subset F_{C}(R^{+})$ and, since a non-redundant atom $\phi$ in $F_{C - K_S}( \emptyset)$ has a single constant in its upper constant segment it can be redundant with no model, so $\phi$ should be a non-redundant atom of both $F_S$ and $F_{C}(R^{+})$. We have shown the non-redundant atoms of $F_{C - K_S}( \emptyset)$ and the non-redundant atoms of $F_{C}(R^{+})^{{}^{\vee}K_S}$ are non-redundant atoms of $F_{C}(R^{+})$ so we can write $F_S  \sqsubset F_{C}(R^{+})$ for every solution $S$ of $P$ and the embedding is tight.
\end{proof}

\bigskip
\bigskip

\begin{lemma}  \label{extensionLemma} 
Let M be a semilattice model with constants $C$ and let $R$ be a set of positive duples over $C$. Let $M_R  \equiv F_C(Th^{+}_0(M) \cup R^{+})$, i.e. the freest extension of $M$ that satisfies $R^{+}$. Let $C'$ be a superset of $C$ with a constant $g_{\sigma}$ for each atomic sentence $\sigma$ of $R$. For a duple $\sigma \equiv (a \leq b)$ define $\sigma' \equiv(a \leq b \odot g_{\sigma})$.  Let $M_{R'} \equiv F_{C'}(Th^{+}_0(M) \cup_{i: (\sigma_i \in R) } \{\sigma'_i\})$. The set of non-redundant atoms of $M_R$ is a subset of the set of non-redundant atoms of $M_{R'}$, written $M_R\sqsubset M_{R'}$. In addition, the restriction to $C$ satisfies $M_{R'}^{|C} \approx M$ and grounded to $C$ satisfies $M_{R'}^{{}^{\vee}C} \approx M_R$.
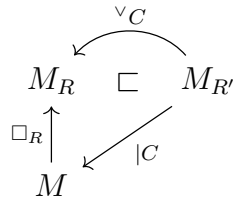
\[
\begin{tikzcd}
M_{R} \arrow[r, symbol=\sqsubset] & M_{R'} \arrow[l, "{}^{\vee}C"', bend right=49] \arrow[ld, "|C"] \\
M \arrow[u, "\square_{R}"]    &                                                                     
\end{tikzcd}
\]
\end{lemma}
\begin{proof}
$M_{R'}$ can be calculated by starting with model $M$ and enforcing each duple of $\sigma'_i$ one by one, in any order, using full crossing. Assume first that $R$ contains a single duple $\sigma \equiv (a \leq b)$ and consider the following two models: the model $M_{\sigma}  \equiv F_{C}(Th^{+}_0(M) \cup \{\sigma\}) = \square_{\sigma} M \equiv  (M -  {\bf{dis}}_M(a, b))  \, \cup \, ({\bf{L}}^{a}(b) \bigtriangledown {\bf{dis}}_M(a, b))$ resulting from full-crossing $\sigma$ over $M$, and the model $M_{\sigma'} \equiv F_{C\cup \{g_{\sigma}\}}(Th^{+}_0(M) \cup \{\sigma'\})$ obtained from the full-crossing of $\sigma'$ over a model atomized by the atoms of $M$ plus an additional atom $\phi_{\sigma}$ with upper constant segment $U^{c}(\phi_{\sigma}) = \{ g_{\sigma} \}$. \\ 
From the definition of full crossing follows that $M_{\sigma'} = M_{\sigma}  \cup \{\phi_{\sigma}\} \cup  (\phi_{\sigma} \bigtriangledown {\bf{dis}}_M(a, b))$, where we are using the same symbols for models and their atomizations. It follows that the atoms of $M_{\sigma}$ are a subset of the atoms of $M_{\sigma'}$. In addition, each atom of $M_{\sigma'}$ that is not an atom of $M_{\sigma}$ contains the constant $g_{\sigma}$ in its upper constant segment. $M_{\sigma'}^{|C}$ is the restriction of $M_{\sigma'}$ to the subalgebra spawn by the terms and constants over $C$, and it is a semilattice model that can be atomized by taking the atoms of $M_{\sigma'}$ and intersecting their upper constant segments with $C$; we get $M_{\sigma'}^{|C}= M_{\sigma}^{|C}  \cup \{\phi^{|C}_{\sigma}\} \cup  (\phi_{\sigma} \bigtriangledown {\bf{dis}}_M(a, b))^{|C} = M_{\sigma}  \cup {\bf{dis}}_M(a, b)$. Substituting $M_{\sigma}$ we obtain $M_{\sigma'}^{|C} =  (M -  {\bf{dis}}_M(a, b)) \, \cup \, ({\bf{L}}^{a}_M(b) \bigtriangledown {\bf{dis}}_M(a, b)) \cup {\bf{dis}}_M(a, b)$ and, since ${\bf{dis}}_M(a, b) \subseteq M$ then $M_{\sigma'}^{|C} = M \, \cup \, ({\bf{L}}^{a}_M(b) \bigtriangledown {\bf{dis}}_M(a, b))$. The atoms of the discriminant are atoms of $M$ as well as the atoms of the lower atomic segment ${\bf{L}}^{a}_M(b) \subseteq M$, so every atom in ${\bf{L}}^{a}_M(b) \bigtriangledown {\bf{dis}}_M(a, b)$ is redundant with $M$ and then it follows that $M_{\sigma'}^{|C}$ can be atomized by the same atoms than $M$, hence, they are isomorphic $M_{\sigma'}^{|C} \approx M$. \\
The model $M_{\sigma'} = M_{\sigma}  \cup \{\phi_{\sigma}\} \cup  (\phi_{\sigma} \bigtriangledown {\bf{dis}}_M(a, b))$ is atomized by the atoms in $M_{\sigma}$ and a bunch of other atoms that all contain $g_{\sigma}$ in their upper constant segment. We argue that a non-redundant atom $\phi$ of $M_{\sigma}$ is also non-redundant in $M_{\sigma'}$ because $\phi$ does not contain $g_{\sigma}$ in its upper constant segment and therefore it cannot be written as a union of atoms with some atom containing $g_{\sigma}$, so the extra atoms of  $M_{\sigma'}$ cannot make $\phi$ redundant.  Therefore, the non-redundant atoms of $M_{\sigma}$ are non-redundant atoms of $M_{\sigma'}$ and we can write $M_{\sigma}  \sqsubset M_{\sigma'}$.  \\
We have proven that given a model $M$ of a semilattice with constants in $C$ and a duple $\sigma$ there is an extended model $M_{\sigma'}$ with constants in $C' \supset C$ such that $M_{\sigma'}^{|C} \approx M$ and the set of non-redundant atoms of $M_{\sigma} \equiv F_{C}(Th^{+}_0(M) \cup \{\sigma\})$ is a subset of the set of non-redundant atoms of $M_{\sigma'}$. We can represent this property with a triangular cell:
\[
\begin{tikzcd}
M_{\sigma} \arrow[r,  symbol=\sqsubset] & M_{\sigma'} \arrow[l, "{}^{\vee}C"', bend right=49] \arrow[ld, "|C"] \\
M \arrow[u, "\square_{\sigma}"]    &                                                                     
\end{tikzcd}
\]
where the vertical arrows represent a full crossing operation with $\sigma$, written $\square_{\sigma}$. From $M^{{}^{\vee}C}_{\sigma} = M_{\sigma}$ and $\{\phi^{{}^{\vee}C}_{\sigma}\} \cup  (\phi_{\sigma} \bigtriangledown {\bf{dis}}_M(a, b))^{{}^{\vee}C} = \emptyset\,$ follows that $M_{\sigma'}$ grounded to $C$ satisfies $M^{{}^{\vee}C}_{\sigma'} = M^{{}^{\vee}C}_{\sigma}  \cup \{\phi^{{}^{\vee}C}_{\sigma}\} \cup  (\phi_{\sigma} \bigtriangledown {\bf{dis}}_M(a, b))^{{}^{\vee}C} = M_{\sigma}$, which is represented in the triangular cell.  \\
If $R$ contains more than one duple, we can sort the duples in any order and let the full diagram be: 
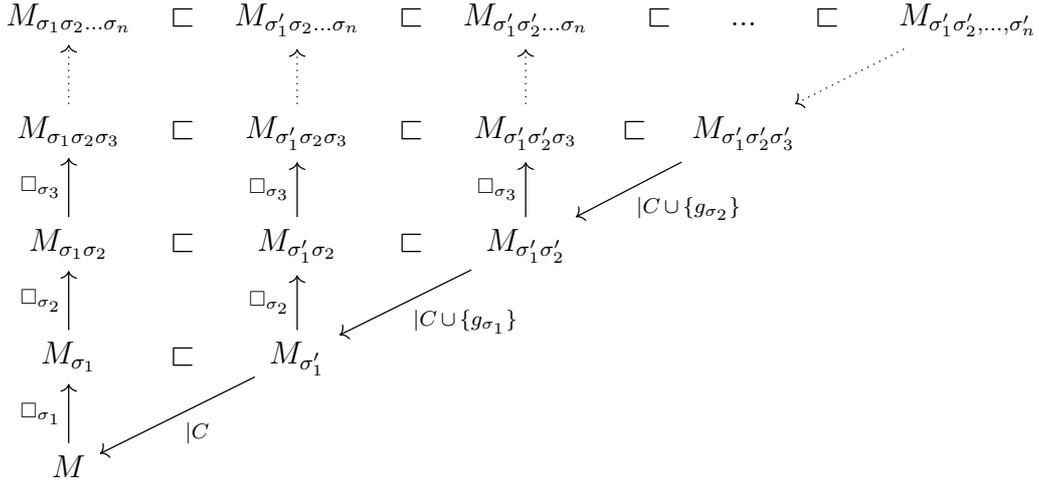
\[
\begin{tikzcd}
M_{\sigma_1 \sigma_2...\sigma_n} \arrow[r, symbol=\sqsubset]                    & M_{\sigma'_1 \sigma_2...\sigma_n} \arrow[r, symbol=\sqsubset]                    & M_{\sigma'_1 \sigma'_2...\sigma_n} \arrow[r, symbol=\sqsubset]                                            & ... \arrow[r, symbol=\sqsubset]                                                  & {M_{\sigma'_1 \sigma'_2,...,\sigma'_n}} \arrow[ld, dotted] \\
M_{\sigma_1 \sigma_2 \sigma_3} \arrow[r, symbol=\sqsubset] \arrow[u, dotted]    & M_{\sigma'_1 \sigma_2 \sigma_3} \arrow[r, symbol=\sqsubset] \arrow[u, dotted]    & M_{\sigma'_1 \sigma'_2 \sigma_3} \arrow[r, symbol=\sqsubset] \arrow[u, dotted]                            & M_{\sigma'_1 \sigma'_2 \sigma'_3} \arrow[ld, "{|C \,\cup\, \{g_{\sigma_2}\} }"] &                                                                     \\
M_{\sigma_1 \sigma_2} \arrow[r, symbol=\sqsubset] \arrow[u, "{\square_{\sigma_3}}"] & M_{\sigma'_1 \sigma_2} \arrow[r, symbol=\sqsubset] \arrow[u, "{\square_{\sigma_3}}"] & M_{\sigma'_1 \sigma'_2} \arrow[u, "{\square_{\sigma_3}}"] \arrow[ld, "{|C \,\cup\, \{g_{\sigma_1}\} }"] &                                                                            &                                                                     \\
M_{\sigma_1} \arrow[r, symbol=\sqsubset] \arrow[u, "{\square_{\sigma_2}}"]     & M_{\sigma'_1} \arrow[u, "{\square_{\sigma_2}}"] \arrow[ld, "|C"]          &                                                                                                     &                                                                            &                                                                     \\
M \arrow[u, "{\square_{\sigma_1}}"]                                      &                                                                            &                                                                                                     &                                                                            &                                                                    
\end{tikzcd}                                                            
\]
% https://tikzcd.yichuanshen.de/#N4Igdg9gJgpgziAXAbVABwnAlgFyxMJZABgBoBmAXVJADcBDAGwFcYkQBZAfWAB1fsAcwC29LgEYAviEml0mXPkIpxFanSat23PgKwj6AcgnTZ87HgJEyAFnUMWbRJxlyQGC0uukATPc1OnDz8QqISAAQh+mE+pm4eilYqvv6O2sF6BsbikZkxceaJysg+KTQOWs46UVkRNaLGsa6FlsVk4qmVQbqhYjn1Yj65vVzkBe4KrUSqHeUB6T3RRnV5g8NLo+MJUyilsxppVRm92eu1QwObzROeScjkpPsVgdWrpwONZw1XZjdFRDZSMROoEAHK5GAAR2YWFo4VeJxWiNKADo0aQPlwwFtJl4UGRgXNDt1Lv1VlwfGiUZdsddtnjkKpCQcugilu9yZS0TScbdiqVmc8FpiyciqTy6bi7g9BfNnFSZOoYFBBPAiKAAGYAJwgwiQZBAOAgSFUhvoWEY7EgYDYNEYWBt7CgEGYACNGLaQAALGD0KBIMDMRiMX7a3VIUqG42IA1C5whaH0LUwXKkL59a5hvUxmhGpAPFmBBPMJMp-hpy4+TM67MFvOIQEge2O5zOt0ekA0H1+gNBxi582W5zWtihmsm3PRxtxkDF0up9MU6vhhuTpAAVjtDsCbfdnu7-sQgeDA4tVoIo7cWfza8QADYiV058mF5dyMvs436wB2R9FgSJi+5aLu+Y4rpuUZIAAHH+7DPmWvAVuSoFXuO963r+Tbbk6Lp7p23q+oex79maZ7DheH5IJh9YwVhLYgLuHZdoRvYnqRQ7gBRYHZrR9YAJxbvRjH7ixR59qeHEjpRiAQfWprNoEcAQPa-qwc4AA+ADCqb8AAxswaDAfwwCCMcGxSPwkjhNJAmQTJgmKcpWCqYW7BaTpvD6YZiG5CZZkGBSkiWdZ3Emga-EOU69BwD6-qhYgpr1pGM5adJ4iJdG4iRgpOHtiJPZiWxOCDuejrxVlt7iAWOWtrhTEEQVxESaVl6amhVWVY2NUMXV+VEeJ7EtWlX6ZQa3XCfhB6sSRxVkZxZWoSuD52elkWttFsXSdRmXZdhtU4DgyrSbxmXVXtPUHUdkiUJIQA
Consider the square diagrams:
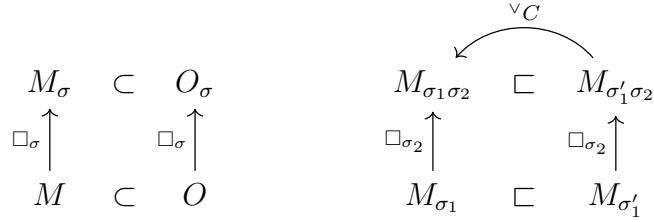
\[
\begin{tikzcd}
M_{\sigma} \arrow[r, symbol=\subset]                    & O_{\sigma}                    &  & M_{ \sigma_1 \sigma_2} \arrow[r, symbol=\sqsubset]                & M_{ \sigma_1' \sigma_2} \arrow[l, "{}^{\vee}C"', bend right=49] \\
M \arrow[u, "\square_{\sigma}"] \arrow[r, symbol=\subset] & O \arrow[u, "\square_{\sigma}"] &  & M_{\sigma_1} \arrow[u,"\square_{\sigma_2}"] \arrow[r, symbol=\sqsubset] & M_{\sigma_1'} \arrow[u, "\square_{\sigma_2}"]               
\end{tikzcd}
\]
The square cell on the left says that if $M$ is atomized by a subset of the atoms of $O$ then the result of full crossing $\sigma$ over $M$ produces a model $M_{\sigma}$ that can be atomized by a subset of the atoms of $O_{\sigma}$, which is straightforward from the definition of full crossing. For the triangular cells we have proven that the non-redundant atoms of $M_{\sigma}$ are also non-redundant atoms of $M_{\sigma'}$. For a square cell this property does not hold in general, however it does for all the square cells in the full diagram; on each square cell we have some $O = M \, \cup \, X^g$ where $X^g$ is a set of atoms that contain each some constant $g$ that is not in the upper segment of any atom of $M$. In addition, the full crossing in the square cells occurs with some duple $\sigma \equiv (a \leq b)$ with $a$ and $b$ terms over a set of constants that do not contain $g$. It is easy to show that in this situation $O_{\sigma} = M_{\sigma} \, \cup \, Y^g$ where $Y^g$ is again a set of atoms that contain $g$ in their upper constant segments and, hence, cannot make any atom of  $M_{\sigma}$ redundant in $O_{\sigma}$ and we can write  $M_{\sigma} \sqsubset O_{\sigma}$. Also, the grounded model to $C$ of $O_{\sigma}$ when $C$ contains the constants of $O_{\sigma}$  except $g$ produces $O^{{}^{\vee}C}_{\sigma} = M^{{}^{\vee}C}_{\sigma} \, \cup \, (Y^g)^{{}^{\vee}C} = M_{\sigma} $.   \\
Finally, since $F_{C}(Th^{+}_0(M_{\sigma_1}) \cup \{\sigma_2\}) = F_{C}(Th^{+}_0(M)  \cup \{\sigma_1\} \cup \{\sigma_2\})$ we can identify $M_{\sigma'_1 \sigma'_2,...,\sigma'_n}$ with the model $F_{C'}(Th^{+}_0(M) \cup R)'$ and this proves the lemma. 
\end{proof}

\bigskip
\bigskip

\begin{theorem}  \label{solutionsAsSubsets} 
Let $R$ be an embedding set over the constants $C$ for a problem $P$ with interpretation $(Q, {\bf{\varphi}}, \Gamma)$ and a scope sentence $\Xi$ that is a first order sentence without quantifiers and with atomic subclauses in the separator set $\,\Gamma$. Create a new constant $g_{\sigma}$ for each atomic sentence $\sigma \equiv (a \leq b)$ of $\Gamma$ and define $\sigma' \equiv(a \leq b \odot g_{\sigma})$ and the set $\Gamma' = \cup_{\sigma \in \Gamma}  \{ \sigma'  \}$: \\
i) The extended embedding set $R^E = R  \cup \Gamma'$ over the embedding constants $C  \, \cup_{\sigma \in \Gamma} \{ g_{\sigma} \}$ is an embedding of $P$ with the same interpretation and scope sentence. \\
ii) The extended embedding $R  \cup \Gamma'$ is an explicit and also tight embedding. \\
iii) For each solution $S$ the freest solution model of the extended embedding restricted to $C$ is equal to the freest solution model of the original embedding.  \\
iv) The extended embedding is concise if and only if the embedding $R$ is concise. \\
v) The extended embedding is complete if and only if the embedding $R$ is complete.
\end{theorem}
\begin{proof}
For a solution $S$ of $P$ consider the following diagram obtained applying lemma \ref{extensionLemma}:
\[
% https://tikzcd.yichuanshen.de
\begin{tikzcd}
                                                                     & {F_{C_S}(R^{+}   \cup (\Gamma \cap  \overline{S})' \, \cup \, S)} \arrow[r, symbol=\sqsubset]      & {F_{C  \, \cup_{\sigma \in \Gamma} \{ g_{\sigma} \}}(R^{+}  \cup \Gamma')} \arrow[ld, "{|C_S}"] \\
{F_{C}(R^{+}  \cup  (\Gamma \cap  \overline{S}))} \arrow[r, symbol=\sqsubset] & {F_{C_S}(R^{+} \cup (\Gamma \cap  \overline{S})')} \arrow[ld, "|C"] \arrow[u, "\square_{S}"] &                                                                \\
F_C(R^{+}) \arrow[u, "\square_{\Gamma \cap  \overline{S}}"]                 &                                                                                                        &                                                                \\
F_C(\emptyset) \arrow[u, "\square_{R^{+}}"]                                             &                                                                                                        &                                                               
\end{tikzcd}
\]
where $C_S \equiv C \, \cup \{ g_{\sigma} : \sigma \in \Gamma \cap  \overline{S} \}$. \\
(i) We argue that, just like $F_S \equiv F_C(R^{+} \cup S)$ is a model solution of $S$ that satisfies $\,\Xi$ (see theorem \ref{freestSolutionModel}), the model $N_S  \equiv F_{C \, \cup \{ g_{\sigma} : \sigma \in \Gamma \cap \overline{S} \}}(R^{+} \cup (\Gamma \cap  \overline{S})' \, \cup \, S)$ is also a solution of $P$ that can be identified as $S$ (we use $\overline{S}$ for the complement of $S$ in $\Gamma$, i.e. $\Gamma \cap \overline{S} = \Gamma - S$). Since a positive atomic sentence $\sigma' \in \Gamma'$ has no consequences in the subalgebra of $N_S$ spawned by terms over $C$, i.e. in the subalgebra $N_S^{|C}$, then it follows that any positive or negative atomic sentence over the constants of $C$ is true in $N_S$ if and only if it is true in $F_C(R^{+} \cup S)$; we can write $N_S^{|C} \approx F_S \equiv F_C(R^{+} \cup S)$. From the diagram above and lemma \ref{extensionLemma} we have that $N_S = F_{C  \, \cup_{\sigma \in \Gamma} \{ g_{\sigma} \}}(R^{+}  \cup \Gamma')^{{}^{\vee}C_S}$ and also $N_S \sqsubset F_{C  \, \cup_{\sigma \in \Gamma} \{ g_{\sigma} \}}(R^{+}  \cup \Gamma')$. \\
We have built an extended embedding with embedding set $R^{E} \equiv R \, \cup \, \Gamma'$ and embedding constants $C^{E} \equiv C \, \cup \{ g_{\sigma} : \sigma \in \Gamma \}$. Notice that every pair $\sigma, \sigma'$ satisfies $\sigma \Rightarrow \sigma'$ and then the freest solution model for $S$ in the extended embedding $F_S^E \equiv F_{C^{E}}((R^{E})^{+} \cup S) = F_{C^{E}}(R^{+} \cup \Gamma' \cup S) = F_{C^{E}}(R^{+} \cup (\Gamma \cap  \overline{S})' \, \cup \, S)$ a model that, we are about to show, differs from $N_S$ in a bunch of non-redundant atoms. Since the constants in $\{ g_{\sigma} : \sigma \in S \}$ are not mentioned in $R^{+} \cup (\Gamma \cap  \overline{S})' \, \cup \, S$, a non-redundant atom of $F_{C^{E}}(\emptyset)$ with an upper constant segment equal to a single constant in $\{ g_{\sigma} : \sigma \in S \}$ remains unaltered after the full crossing of the duples in $R^{+} \cup (\Gamma \cap  \overline{S})' \, \cup \, S$ as a non-redundant atom of $F_S^E$. It follows that $F_S^E$ is the model spawned by the atoms of $N_S$ plus a set of atoms $\Lambda$ each with a single constant in its upper constant segment, a constant in the set $\{ g_{\sigma} : \sigma \in S \}$, so we can write $F_S^E = N_S + \Lambda$. Applying theorem \ref{restrictionLemma},  $(F_S^E)^{|(C \, \cup \{ g_{\sigma} : \sigma \in \Gamma \cap \overline{S} \})} = N_S$ and also  $(F_S^E)^{|C} = N_S^{|C} \approx F_S$.  This proves that if we can use $\Gamma$ to build a model for each solution of $P$ using the embedding set $R$ we can also use $\Gamma$ to build a model for each solution of $P$ using the embedding set $R^{E}$.  \\
(ii) Since the constants $g_{\sigma}$ only appear in the right hand side of the duples of $R^{E}$ then the atoms in the freest model $F_{C^{E}}( \emptyset)$ with upper constant segments in the set $\{ g_{\sigma} : \sigma \in \Gamma \}$ remain all in $F_{C^{E}}(R^{E+})$ as they survive every full crossing operation with a duple in $R^{E+}$. This implies the atoms of $\Lambda$ are all atoms in $F_{C^{E}}(R^{E})$ and, since they have a single constant in their upper constant segment then they are non-redundant atoms, i.e. $\Lambda \sqsubset F_{C^{E}}(R^{E+})$. We showed above that $N_S \sqsubset F_{C^{E}}(R^{E+})$. Since $F_S^E = N_S + \Lambda$ and both $N_S$ and $\Lambda$ are tight subset models of $F_{C^{E}}(R^{E+})$ then $F_S^E \sqsubset F_{C^{E}}(R^{E+})$. Therefore, for every solution $S$ of $P$ the freest solution model $F_S^E$ is spawned by a subset of the non-redundant atoms of $F_{C^{E}}(R^{E+})$ and this proves that the extended embedding $R^{E}$ is tight. \\
Since $N_S$ a model over $C \, \cup \{ g_{\sigma} : \sigma \in \Gamma \cap \overline{S} \}$ and $\Lambda$ is a set of atoms over $\{ g_{\sigma} : \sigma \in S \}$ that spawn the freest model $F_{\{ g_{\sigma} : \sigma \in S \}}( \emptyset)$, we can write $F_S^E = N_S + F_{\{ g_{\sigma} : \sigma \in S \}}( \emptyset)$ and, since the sets of constants of both models are disjoint $F_S^E = N_S \oplus F_{\{ g_{\sigma} : \sigma \in S \}}( \emptyset)$. It follows that the extended embedding is explicit. \\
(iii)  We showed above that $(F_S^E)^{|C} = N_S^{|C} \approx F_S$.  \\
(iv) Theorem \ref{conciseEmbeddingTheorem} says that an embedding is concise if and only if for each solution $S$ the freest solution model $F_S$ restricted to the interpretation constants $Q$ is equal to $F_{Q}(S)$ for each solution $S$ of $P$. Using (iii), from $(F_S^E)^{|C} = F_S$ we get $(F_S^E)^{|Q} = ((F_S^E)^{|C})^{|Q}  = F_S^{|Q}$ so the extended embedding is concise if and only if the original embedding is concise.  \\
(v) The diagram above shows that $F_{C^{E}}(R^{E+})^{|C} = F_{C}(R^{+})$, so for every positive atomic sentence $r^{+}$ over the interpretation constants $Q \subseteq C$ we have that $R \vdash r^{+}$ if and only if $R^{E} \vdash r^{+}$ and then the extended embedding is complete if and only if the original embedding is complete.
\end{proof}

\bigskip
\bigskip

\begin{theorem}  \label{minimalModelsTheroem} 
An irreducible model of a set $R$ of duples has at most $\vert R^{-} \vert + 1$ atoms.
\end{theorem}
\begin{proof}
It is a consequence of two facts: it is enough with one atom to discriminate any negative duple and removing atoms from a model of $R^{+}$ always produces a model of $R^{+}$. Atoms can be removed from a model as long as $R^{-}$ is still satisfied and there is an atom in the lower atomic segment of every constant. Since $\ominus_C$ is in every model, adding $\ominus_C$ to an atomization suffices to ensure that the sixth axiom of atomized semilattices is satisfied. Hence, it is always possible to remove some atom from an atomized model that has more than $\vert R^{-} \vert$ atoms and still get a model of $R$ if we add $\ominus_C$ when needed. It follows that an irreducible model of $R$ should always have fewer than $\vert R^{-} \vert + 1$ atoms.  
\end{proof}

\bigskip

\begin{theorem}  \label{positiveScopeTheroem} 
Consider an algebraic embedding for a problem $P$ with embedding set $R$ and a scope sentence $\Xi$ that can be written as a conjunction of disjunctions, i.e. $\Xi = \wedge_i \xi_i$, with each $\xi_i$ a disjunction of positive atomic sentences. \\ 
i) Every irreducible model of a solution $S$ is an irreducible model of $R$. \\
ii) For each solution $S$ of $P$ there is an irreducible model of $R$ that can be identified with $S$. \\
iii) If $R$ is a tight embedding, for each solution $S$ of $P$ there is an irreducible model of $R$ atomized with non-redundant atoms of $F_C(R^{+})$ that can be identified with $S$. 
\end{theorem}
\begin{proof}
 A solution model is any model of $R \cup  \{  \Xi \}$. If $\,\Xi = \wedge_i \xi_i$ is a conjunction of disjunctions of positive atomic sentences in a set $\Lambda$, we can express $\Xi$ in disjunctive normal form as a disjunction $\Xi = \vee_k \Xi_k$ of conjunctions $\Xi_k$ of positive atomic sentences in the same set $\Lambda$. The index $k$ usually runs along a much larger set of values than the original index $i$ and most values of $k$ correspond to clauses $\Xi_k$ satisfied by no solution of $P$. However, every solution $S$ must satisfy at least one clause $\Xi_k$. If $S$ satisfies more than one $\Xi_k$ we can do the conjunction of these so, without loss of generality, we can assume there is a single value of $k$, the value $k(S)$, that maps to the solution $S$. Because $\Xi_{k(S)}$ contains only positive atomic sentences $R^{+} \cup S \Rightarrow  \Xi_{k(S)}$. It follows that a model of a solution $S$ is any model of the set of atomic sentences $R_S = R^{+} \cup R^{-} \cup S$.\\
(i) Suppose $M$ is an irreducible model of $R_S$, i.e. it is not possible to remove a non-redundant atom of $M$ and still model $R_S$. It is clear that $M$ is a model of $R = R^{+} \cup R^{-}$. We argue that $M$ is also an irreducible model of $R$. The elimination of one or many atoms from a model preserves the positive atomic sentences entailed by the model, i.e. any elimination of atoms from $M$ still models $R^{+} \cup S$. Therefore, since $M$ is an irreducible model of $R_S$ after the elimination of an atom the model spawned by the remaining atoms should not satisfy at least one negative atomic sentence in $R^{-}$, which means that $M$ is an irreducible model of $R$.  \\ 
(ii) We have assumed $R$ is an embedding set of $P$ so there are models of $R$ for each solution $S$. These solution models are the models of $R_S$. We can find irreducible models for any set of atomic sentences, so $R_S$ should have irreducible models and according to part (i) these models are also irreducible for $R$. Therefore, we can find for each solution $S$ an irreducible model of $R$ that models $S$.  \\ 
(iii) If the embedding is tight the non-redundant atoms of the freest solution model $F_S = F_C(R^{+} \cup S)$ are all non-redundant atoms of $F_C(R^{+})$. It is clear that $F_S  \models R_S$. We start with $F_S$ atomized with non-redundant atoms and select a subset $M$ of the atoms of $F_S$ that suffice to satisfy $R^{-}$ and such that it cannot be made smaller, i.e. an irreducible model of $R_S$. From part (i), $M$ is an irreducible model of $R$ and is atomized by non-redundant atoms of $F_C(R^{+})$. 
\end{proof}

\section{Context Constants} \label{sec:context_constants}

Context constants are useful to build explicit embeddings that can be used in practice.  In lemma \ref{extensionLemma} a different context constant $g_{\sigma}$ is used for each duple $\sigma$. We will see in this section that we can reuse some of our context constants and still obtain an explicit embedding.

Consider the following observation:

\begin{lemma}  \label{extensionLemmaPartB} 
Lemma \ref{extensionLemma} holds if duples $\sigma'$ are defined as $\sigma' \equiv(a \leq b \odot g)$ where the same constant $g$ is used for every $\sigma$ (instead of a different constant $g_{\sigma}$).
\end{lemma}
\begin{proof}
The restriction to $C$ of $M_{R'}$ should be equal to $M$ as the duples of the form $\sigma'$ have no consequences in the space spawned by the constants in $C$. With the new definition for $\sigma'$ the proof of lemma \ref{extensionLemma} still holds step by step. We get to a model $M_{R'} = M_{R} \cup \{\phi_{g} \bigtriangledown X\}$ where $X$ is some set of atoms and $\phi_{g}$ is an atom with an upper constant segment equal to $g$. From here it follows immediately that $M_{R} \sqsubset M_{R'}$ and $M_{R} = M^{{}^{\vee}C}_{R'}$. 
\end{proof}
\bigskip

Lemma \ref{extensionLemmaPartB} is simpler than lemma \ref{extensionLemma} as it uses a single ``context constant'' but it is also weaker than lemma \ref{extensionLemma}. In fact, we cannot prove theorem \ref{solutionsAsSubsets} from it.  A stronger version of \ref{extensionLemma} is the following:

\bigskip
\begin{lemma}  \label{extensionLemmaPartC} 
Let M be a semilattice model with constants $C$ and $R = \{\sigma_i : i \in I\}$ a set of positive duples over $C$ indexed by a set $I$ and let $M_R  \equiv F_C(Th^{+}_0(M) \cup R)$. Extend the constants with an additional set of (context) constants $G$ disjoint with $C$, to an extended set $C' = C \cup G$ and for every duple $\sigma_i \equiv (a_i \leq b_i)$ of $R$ define one or many duples $\sigma_{ij}' \equiv(a_i \leq b_i \odot g_{ij})$ where $g_{ij} \in G$. Let $R'$ be the set of such duples $\sigma_{ij}'$ and let $M_{R'} \equiv F_{C'}(Th^{+}_0(M) \cup R')$; then $M_R \sqsubset M_{R'}$ and the restriction and grounding to $C$ satisfy $M_{R'}^{|C} = M$ and $M_{R'}^{{}^{\vee}C} \approx M_R$ respectively.
\end{lemma}
\bigskip

Lemma \ref{extensionLemmaPartC} can be proven in the same way as lemma \ref{extensionLemma}. This lemma says that we can use the same context constant for all the duples or one context constant for each duple or any choice of context constants in between. In addition, a duple $\sigma$ of $R$ can have in $R'$ one or many $\sigma'$ by using one or many constants in $G = C' - C$. 

Theorem \ref{solutionsAsSubsets} implies that we can find the freest solution model of every solution $S$ by calculating the grounding  $N_S = F_{C'}(R^{+}  \cup \Gamma')^{{}^{\vee}C_S}$ to some subset $C_S \subseteq C'$. We proved theorem \ref{solutionsAsSubsets} from lemma \ref{extensionLemma} and we can also prove it from lemma \ref{extensionLemmaPartC} although not every choice of context constants works in this case. We will need a set of context constants sufficiently large to ``separate'' every solution $S$ of the problem $P$ as the grounded model to some subset $C_S$:

\bigskip
\begin{theorem}  \label{contextConstantReuse} 
Let $R$ be an embedding set over $C$ for a problem $P$ with interpretation $(Q, {\bf{\varphi}}, \Gamma)$ and a scope sentence $\Xi$ that is a first order sentence without quantifiers and with atomic subclauses in the separator set $\,\Gamma$. Assume $C$ is extended to $C' = C \cup G$ with an additional set of (context) constants $G$, disjoint with $C$ and such that for each $\sigma_i \equiv (a_i \leq b_i)$ in the separator set $\Gamma$ there is one or many $\sigma_{ij}' \equiv(a_i \leq b_i \odot g_{ij})$ in a set $\Gamma'$ where $g_{ij} \in G$.  \\
The embedding $R  \cup \Gamma'$ over $C'$ is explicit if and only if for each solution $S \subseteq \Gamma$ there is at least one set $G_S \subseteq G$ of context constants such that: \\
 - if $\sigma \in \Gamma$ appears as $\sigma' \in \Gamma'$ using a context constant from $G_S$ then $\sigma \in S$, and \\ 
 - for every $\sigma \in S$ there is at least one $\sigma'  \in \Gamma'$ that uses a context constant in $G_S$.
\end{theorem}
\begin{proof}
We can follow the steps of the proof of theorem \ref{solutionsAsSubsets}. To build this proof we relied on a diagram that requires the existence of a subset of constants $C_S$. We can relax the requirement of having a different context constant for each atomic sentence $\sigma$ in $\Gamma$ as long as there is still a suitable subset $C_S$ of the extended set of constants $C'$. We claim that $C_S = C - G_S$ is such suitable subset. \\ 
In the proof of theorem \ref{solutionsAsSubsets} we had  $N_S = F_{C'}(R^{+} \cup \Gamma')^{{}^{\vee}C_S}$ where $N_S$ is a freest solution model for $S$. For this equality to hold, a $\sigma' \in \Gamma'$ that uses a context constant in the set $C' - C_S$ (i.e. in $G_S$) should correspond to a $\sigma \in S$, otherwise the grounded model $F_{C'}(R^{+}  \cup \Gamma')^{{}^{\vee}C_S}$ would satisfy a positive duple that is not satisfied by the freest solution model of $S$. We also need for every $\sigma \in S$ at least one $\sigma' \in \Gamma'$ that uses a context constant in $C' - C_S$ otherwise the grounded model $F_{C'}(R^{+}  \cup \Gamma')^{{}^{\vee}C_S}$ does not satisfy $\sigma \in S$ and it cannot be a solution model of $S$. 
\end{proof}
\bigskip
\bigskip

It may seem at first glance that in order to use theorem \ref{contextConstantReuse} we need to know the solutions $S$ in advance, but this is not true. We use theorem \ref{contextConstantReuse} to reduce the number of context constants in the vertical bar and the hamiltonian graph embeddings without knowing the solutions.

\section{Acknowledgments} 

We thank Nabil Abderrahaman Elena, Antonio Ricciardo, Emilio Suarez and David Mendez for critical comments on the
manuscript. We are grateful for the support from Champalimaud Foundation (Lisbon, Portugal) provided through the Algebraic Machine Learning project, from Portuguese national funds through FCT in the context of the project UIDB/04443/2020, and from the European Commission provided through projects H2020 ICT48 \emph{Humane AI; Toward AI Systems That Augment and Empower Humans by Understanding Us, our Society and the World Around Us} (grant $\# 820437$) and the H2020 ICT48 project \emph{ALMA: Human Centric Algebraic Machine Learning} (grant $\# 952091$).

\bigskip
\bigskip

\newpage

\section{Notation}

We use repetition of subscripts (Einstein notation) such as $\vee_i x_i$ to mean $\vee^b_{i=a} x_i$ for some omitted limits $a$ and $b$ when $a$ and $b$ are obvious. Sometimes we just omit one of the limits, for example we write $\vee^b_{i} x_i$.

We use $Th^{+}_0(M)$ to refer to the set of all positive atomic sentences satisfied by $M$ and $Th^{-}_0(M)$ to refer to the set of all positive atomic sentences not satisfied by $M$. We read them as the positive or negative theory of $M$ with 0 quantifiers. 

We say that a model $M$ has an atom $\phi$ if either $\phi$ is a non-redundant atom of $M$ or $\phi$ is redundant with $M$. If a model $M$ has an atom or not is independent of its atomization. We sometimes write $\phi \in M$. If $M$ has an atom $\phi$ then we say that it $\phi$ is in $M$ or is an atom of $M$.

If an atom $\phi$ is not in $M$ we say that it is out of $M$ or that is external to $M$.

With respect to a model an atom can be in three possible situations: it can be external to the model, it can be a non-redundant atom of the model or it can be a redundant atom with the model. 

We say that an atom is in $M$, written $\phi \in M$, if $\phi$ is a non-redundant atom in $M$ or if it is redundant with $M$. An atom is universally defined independent of the model. With respect to a model an atom can be in three situations: it is either a non-redundant atom in the model OR redundant with the model OR out of the model.  

We often use the same letter for a model and its atomization. Since redundant atoms can be discarded from an atomization and the remaining set of atoms still spawns the same model there are many valid atomizations of a model. We use the same letter for a model and for its atomization when the particular atomization chosen is not relevant. Alternatively, we can consider that a set of atoms that has the same name than the model they spawn contains all the atoms of the model.

A model $M$ is a subset model of $N$, written $M \subset N$ if every atom that is in model $M$ is in model $N$.  A model $M$ is a tight subset model of $N$, written $M \sqsubset N$ if the non-redundant atoms of $M$ are a subset of the non-redundant atoms of $N$.  
We use $\square_w M$ to represent the result of full-crossing $w$ on $M$. The crossing of a duple $w$ on a model $M$ produces a strictly less free model. Any atom of $\square_w M$ is either a non-redundant atom of $M$ or redundant with $M$, i.e. is in $M$ and we can write $\square_w M  \subset  M$. 

We say a sentence $\Xi$ has no negative subclauses if once is expressed in disjunctive normal form it only has positive subclauses.

% \bibliographystyle{unsrt}
% \bibliography{main}
\printbibliography[heading=bibintoc]
\end{document}